\documentclass[11pt]{article}

\usepackage[utf8]{inputenc}
\usepackage[T1]{fontenc}
\usepackage{graphicx}
\usepackage{amsmath,amssymb}
\usepackage{bm}
\usepackage{booktabs}
\usepackage[colorlinks=true,linkcolor=blue,citecolor=blue,urlcolor=blue]{hyperref}
\usepackage[margin=2.5cm]{geometry}
\usepackage{xcolor}

\usepackage[numbers,comma,sort&compress]{natbib}

\usepackage{amsthm}
\usepackage{multirow}

\newtheorem{theorem}{Theorem}[section]
\newtheorem{lemma}[theorem]{Lemma}
\newtheorem{proposition}[theorem]{Proposition}
\newtheorem{corollary}[theorem]{Corollary}
\theoremstyle{definition}
\newtheorem{definition}[theorem]{Definition}

\theoremstyle{remark}
\newtheorem{remark}[theorem]{Remark}

\begin{document}

\title{Spin chirality across quantum state copies detects hidden entanglement}

\author{Patrycja Tulewicz$^{1}$, Karol Bartkiewicz$^{1,*}$, and Franco Nori$^{2,3}$}

\date{}

\maketitle

\noindent $^{1}$Institute of Spintronics and Quantum Information, Faculty of Physics, Adam Mickiewicz University, 61-614 Pozna\'n, Poland\\
\noindent $^{2}$Theoretical Quantum Physics Laboratory, RIKEN Cluster for Pioneering Research, Wako-shi, Saitama 351-0198, Japan\\
\noindent $^{3}$Department of Physics, University of Michigan, Ann Arbor, Michigan 48109-1040, USA\\
\noindent $^{*}$e-mail: karol.bartkiewicz@amu.edu.pl

\vspace{1em}

\begin{abstract}
\noindent\textbf{Entanglement can hide in two fundamentally different ways. First, multi-copy correlations can carry information that no single-copy measurement on an unknown state is able to access. Second, bound entangled states possess a positive partial transpose, which makes them invisible to the Peres--Horodecki criterion and all moment inequalities that depend on it. Here we show that the moment difference between the partial transpose and purity decomposes exactly as a chirality--chirality correlator, where the relevant operator is the scalar spin chirality---the same quantity that governs chiral spin liquids and the topological Hall effect. This decomposition identifies the specific physical structure that multi-copy entanglement detection probes. Using the same controlled-SWAP circuits, we develop a multi-channel spectral classifier for bound entanglement. The classifier combines realignment spectral features with chirality corrections and achieves $99.9\%$ recall at zero false positives across all three known $3 \times 3$ bound entangled families, compared with ${\sim}40\%$ for the CCNR criterion alone. We also introduce a marginal-noise construction that produces CCNR-invisible bound entangled states, which the classifier detects but which remain invisible to all single-parameter criteria. We validate our approach experimentally on three IBM Quantum processors and demonstrate negativity reconstruction with mean errors of 0.002--0.027, chirality detection for pure and mixed entangled states, and bound entanglement detection across two structurally distinct families (Horodecki and chessboard) on a single gate-based superconducting processor.}
\end{abstract}

\vspace{2em}

Standard entanglement detection faces two structural blind spots. First, certain correlations between replicas of an unknown state cannot be expressed as expectation values of any single-copy observable $\mathrm{Tr}[O\rho]$~\cite{PhysRevLett.89.127902,PhysRevLett.90.167901}. These correlations are intrinsically multi-copy phenomena, invisible to all conventional measurements. Second, bound entangled states possess a positive partial transpose~\cite{Horodecki1998bound}, which renders them undetectable by the Peres--Horodecki criterion~\cite{PhysRevLett.77.1413,HORODECKI19961} and the entire hierarchy of moment inequalities~\cite{Yu2021optimal,Neven2021symmetry} that depend on it. Nevertheless, such states remain entangled and may enable quantum key distribution and entanglement activation protocols that are inaccessible to separable states.

These are two fundamentally different forms of hidden entanglement, yet they share a common origin: both reside in spectral structures that become accessible only through joint measurements on multiple state copies. In this work, we show that a unified detection architecture built on controlled-SWAP circuits can probe both layers. Partial transpose moments reveal chirality correlations and negativity for NPT states, while realignment matrix spectral features detect bound entanglement where all PT-based criteria fail.

Partial transpose moments $\mu_k = \mathrm{Tr}[(\rho^{T_A})^k]$, where $\rho^{T_A}$ denotes partial transposition over subsystem $A$ and $k=1,2,...$, provide a natural framework for probing these structures. These moments are measurable via controlled-SWAP circuits~\cite{Travnicek2019,tulewicz2025} or randomized measurements~\cite{Elben2020mixed,Huang2020shadows}. It has been shown that three moments suffice to determine the negativity of any two-qubit state~\cite{Bartkiewicz2015moments}, and optimal certification hierarchies have been established~\cite{Yu2021optimal,Neven2021symmetry}. Nonlinear witnesses based on these moments have been demonstrated in photonic~\cite{Lemr2016collectibility,Travnicek2019,Lim2011spa} and superconducting~\cite{tulewicz2025} experiments. However, what multi-copy correlations these moments actually encode---what physical content distinguishes partial transposition from purity---has remained an open question.

Here we answer this question. We prove that the moment difference $C_k = \mu_k - I_k$ decomposes exactly as a chirality--chirality correlator $C_k = 8\,\mathrm{Tr}[\Omega_A\,\Omega_B\;\rho^{\otimes k}]$, where $\Omega$ is the scalar spin chirality operator. It is worth noting that this is the same quantity that governs chiral spin liquids~\cite{KalmeyerLaughlin1987} and the topological Hall effect~\cite{Taguchi2001}. While chirality has been recognised as a tripartite entanglement witness for lattice spins~\cite{Reascos2023chirality}, its role as the \emph{specific} structure distinguishing partial transposition from purity in multi-copy entanglement detection has not been previously identified. Beyond this theoretical result, we show that the same controlled-SWAP hardware can detect bound entanglement through realignment matrix spectral features, achieving high classification accuracy across all known $3 \times 3$ bound entangled families.

Our work makes three contributions. First, we prove that the multi-copy moment corrections $C_k = \mu_k - I_k$ decompose as chirality--chirality correlators, identifying the specific physical structure that distinguishes partial transposition from purity. Second, we develop a multi-channel spectral classifier for bound entanglement. This classifier combines realignment spectral features with chirality corrections ($C_3$, $C_4$) into a composite nonlinear witness, achieving $99.9\%$ recall at zero false positives across all three known $3 \times 3$ families. For comparison, the CCNR criterion alone reaches only ${\sim}40\%$ recall. Third, we demonstrate both capabilities experimentally on IBM Quantum processors, where we achieve negativity reconstruction with sub-$3\%$ mean error and detect bound entanglement across two structurally distinct families on a single gate-based processor.

\paragraph{Partial transpose moments and the multi-copy window.}
The negativity of a bipartite state $\rho$ is defined as
\begin{equation}
\mathcal{N}(\rho) = \frac{\|\rho^{T_A}\|_1 - 1}{2},
\end{equation}
where $\rho^{T_A}$ denotes partial transposition over subsystem $A$. The choice of subsystem is irrelevant, since $\rho^{T_A}$ and $\rho^{T_B}$ share the same eigenvalue spectrum for any bipartite state. For two-qubit and qubit--qutrit systems, $\mathcal{N} > 0$ if and only if the state is entangled.

Rather than computing eigenvalues of $\rho^{T_A}$ directly, we access them through power-sum moments $\mu_n = \sum_i \lambda_i^n$, which relate to elementary symmetric polynomials via Newton--Girard identities~\cite{Neven2021symmetry}. For a $d_A \times d_B$ system, the partial transpose has $d = d_A d_B$ eigenvalues, so $d - 1$ independent moments $\mu_2, \ldots, \mu_d$ determine its characteristic polynomial. In practice, this means three moments ($\mu_2, \mu_3, \mu_4$) for two-qubit ($2 \times 2$) systems, five ($\mu_2, \ldots, \mu_6$) for qubit--qutrit ($2 \times 3$), and eight ($\mu_2, \ldots, \mu_9$) for qutrit--qutrit ($3 \times 3$) systems. Because the $k$-th moment $\mu_k$ is measurable as a permutation trace on $k$ copies of $\rho$, higher-dimensional systems require controlled-SWAP circuits operating on proportionally more copies (see Methods).

On the $k$-copy space, both $\mu_k$ and the purity moment $I_k = \mathrm{Tr}[\rho^k]$ are permutation traces: $\mu_k = \mathrm{Tr}[(\sigma_A^{-1}\otimes\sigma_B)\,\rho^{\otimes k}]$ and $I_k = \mathrm{Tr}[(\sigma_A\otimes\sigma_B)\,\rho^{\otimes k}]$, where $\sigma$ is the $k$-cycle. For two-qubit systems, the algebra of permutation operators on the $k$-copy space decomposes the correction $C_k = \mu_k - I_k$ into chirality--chirality correlators between the $A$ and $B$ subsystems (Supplementary Section~S1):
\begin{align}
C_2 &= 0, \nonumber\\
C_3 &= 8\,\mathrm{Tr}\bigl[\chi_A\,\chi_B\;\rho^{\otimes 3}\bigr], \nonumber\\
C_4 &= 8\,\mathrm{Tr}\bigl[\Omega_A\,\Omega_B\;\rho^{\otimes 4}\bigr],
\end{align}
where $\chi_A =\chi^A_{123} = \mathbf{S}_{A_1}\cdot(\mathbf{S}_{A_2}\times\mathbf{S}_{A_3})$ is the scalar spin chirality of $A$-qubits from three copies, and $\Omega_A = \tfrac{1}{2}\sum_{i<j<k}\chi_{ijk}^A$ is the symmetric sum over all four chirality triples on four copies. The operators $\chi_B$ and $\Omega_B$ act analogously on the $B$-subsystem. Both operators are Hermitian, so $C_k$ is manifestly real. The symmetric form $\Omega$ emerges from Pauli algebra: the initial SWAP-weighted expression $S_{34}\chi_{123} + S_{12}\chi_{234}$, where $S_{ij}$ denotes the SWAP operator exchanging copies $i$ and $j$, simplifies to $\tfrac{1}{2}(\chi_{123} + \chi_{124} + \chi_{134} + \chi_{234})$. This reveals that all chirality triples contribute equally (Supplementary Section~S1).

The central result is that the difference between partial transpose moments and purity moments is a chirality--chirality correlator---that is, the collective handedness of spin configurations across state copies on subsystem $A$, correlated with the same quantity on $B$. It should be noted that $\mu_2 = I_2$ for all bipartite states, so a single purity circuit suffices. All quantities $C_k$ are measurable via controlled-SWAP circuits on multiple state copies~\cite{Travnicek2019,tulewicz2025} (Supplementary Section~S2). In this approach, each circuit measures only a single ancilla qubit, regardless of system dimension. This provides a uniform measurement interface that scales naturally from qubit--qubit to qutrit--qutrit systems and is particularly advantageous for qutrits and higher-dimensional systems, where the two-qubit singlet state has no direct counterpart and singlet-projection-based schemes do not straightforwardly generalise.

\paragraph{Multi-copy chirality: entanglement hidden from single-copy measurements on unknown states.}
The chirality correction $C_4 = \mu_4 - I_4$ exposes the first form of hidden entanglement: nonclassical correlations that exist only as patterns \emph{between} multiple copies of a quantum state. Formally, $C_4$ cannot be expressed as $\mathrm{Tr}[O\rho]$ for any single-copy observable $O$. It is an intrinsically multi-copy quantity that requires joint measurements on $\rho^{\otimes 4}$ (Supplementary Section~S3). For a known state reconstructed by full tomography, $C_4$ can of course be computed from the density matrix. However, the inaccessibility is operational: no single-copy measurement protocol can extract $C_4$ from an unknown state without first performing complete state reconstruction.

These chiral correlations constitute a physically interpretable entanglement signature with a direct geometric meaning (Fig.~\ref{fig:chirality_concept}). For pure states, $C_4 \neq 0$ certifies entanglement from just two measurements ($\mu_4$ and $I_4$), which is fewer than the three moments required for full negativity reconstruction. This provides a resource-efficient entanglement witness that simultaneously reveals the geometric structure of correlations via $C_4 = \tfrac{3}{4}\det(T)$.

\begin{figure}[t]
\centering
\includegraphics[width=\textwidth]{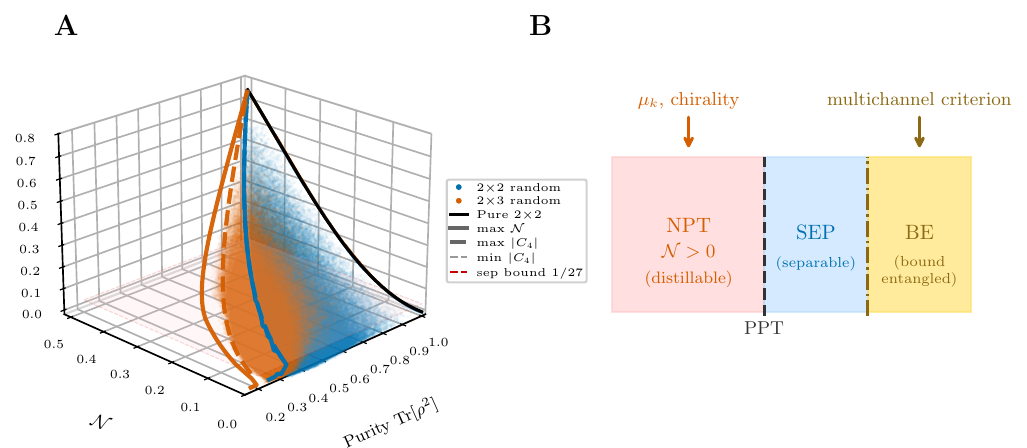}
\caption{\textbf{Chirality across quantum state copies reveals hidden entanglement.}
\textbf{a},~Chirality landscape: purity $\mathrm{Tr}[\rho^2]$ versus negativity $\mathcal{N}$ versus $|C_4|$ for $10^7$ random states ($2{\times}2$ blue, $2{\times}3$ red). Black curve: pure $2{\times}2$ states. Dashed red: separable bound $|C_4| = 1/27$.
\textbf{b},~Detection architecture: controlled-SWAP circuits access negativity (NPT), the non-Hermiticity gap $D_k$ (separable vs bound entangled), and the chirality fingerprint $C_k$ (bound entanglement).}
\label{fig:chirality_concept}
\end{figure}

\textbf{Mathematical structure.} The chirality decomposition has precise algebraic origins. Computing $C_4$ requires comparing four state copies (Fig.~\ref{fig:chirality_concept}a). The underlying operator $\Delta = \sigma^{-1} - \sigma$ can be decomposed via singlet projectors $P^-_{ij} = \frac{1}{4} - \mathbf{S}_i \cdot \mathbf{S}_j$, where $\mathbf{S}_i$ is the spin operator for qubit $i$. The key observation is that the commutators of overlapping projectors yield scalar spin chirality:
\begin{equation}
[P^-_{ij}, P^-_{jk}] = -i\,\chi_{ijk}, \quad \text{where} \quad \chi_{ijk} = \mathbf{S}_i \cdot (\mathbf{S}_j \times \mathbf{S}_k).
\end{equation}
This algebraic structure directly manifests in the moment relations. As we derive in Supplementary Section~S1, the operator $\Delta = \sigma^{-1} - \sigma$ governing $C_k = \mu_k - I_k$ decomposes as $\Delta = 4i\,\Omega_k$, where $\Omega_k$ is Hermitian: $\Omega_3 = \chi_{123}$ and $\Omega_4 = \tfrac{1}{2}\sum_{i<j<k}\chi_{ijk}$, the symmetric sum over all four chirality triples on four copies. The products $g_{mn}\chi_{ijk}$ (where $g_{mn} = 2S_{mn} - I$) generate the ``missing'' chirality triples, and imaginary cross-terms cancel identically.

By exploiting the anti-Hermiticity of $\Delta$, we arrive at the manifestly Hermitian correlator form $C_k = 8\,\mathrm{Tr}[\Omega_A\,\Omega_B\;\rho^{\otimes k}]$, which is a direct correlation of chirality operators on the $A$ and $B$ subsystems. Both chirality factors vanish for product states ($C_k = 0$). For pure states, any entanglement generates non-coplanar chirality configurations with $C_k \neq 0$. For mixed two-qubit separable states, the correction is bounded: $|C_3| \leq 1/36$ and $|C_4| \leq 1/27$ (Supplementary Section~S3).

This decomposition bridges quantum information and condensed matter physics. In condensed matter, the scalar spin chirality $\chi_{ijk} = \mathbf{S}_i \cdot (\mathbf{S}_j \times \mathbf{S}_k)$ is the order parameter for chiral spin liquids~\cite{KalmeyerLaughlin1987} and drives the topological Hall effect~\cite{Taguchi2001}. Here, the same operator appears not among physical spins on a lattice, but among correlations spanning \emph{copies} of a bipartite quantum state. The permutation symmetry of the moment relations ensures that all chirality triples contribute equally ($\Omega = \tfrac{1}{2}\sum_{i<j<k}\chi_{ijk}$). Entangled states generate non-coplanar multi-copy spin patterns with finite chirality, whereas product states remain coplanar with $\langle\chi_A\chi_B\rangle = 0$.

\textbf{Geometric interpretation: volume and the correlation tensor.} The chirality correction also admits a geometric interpretation via the correlation tensor $T_{ij} = \mathrm{Tr}[\rho\,(\sigma_i \otimes \sigma_j)]$, whose columns $\mathbf{T}_x$, $\mathbf{T}_y$, $\mathbf{T}_z$ encode spin--spin correlations. For product states, these three vectors are coplanar, forming a parallelepiped of zero volume with $\det(T) = 0$. For entangled states, the vectors span three-dimensional space and $|\det(T)| > 0$.

The key connection is that for all two-qubit states with vanishing local Bloch vectors ($\mathbf{r} = \mathbf{s} = \mathbf{0}$), the chirality correction satisfies $C_4 = \tfrac{3}{4}\det(T)$ exactly (Supplementary Section~S3). In other words, chirality measures the \emph{signed volume} of the correlation parallelepiped, with the coefficient $3/4$ fixed by the Bell state value [$\det(T) = -1$, $C_4 = -3/4$]. This relation holds for all Bell-diagonal states and their local unitary rotations. For general states with nonzero Bloch vectors, corrections of order $O(|\mathbf{r}|^2 + |\mathbf{s}|^2)$ arise from additional LU invariants (Supplementary Section~S3).

It is also worth noting that chirality connects to Bell nonlocality: the CHSH criterion~\cite{RevModPhys.81.865} depends on the two largest eigenvalues of $T^{\mathsf{T}}T$, while chirality depends on all three through $\det(T)$. For Bell states, $|\det(T)| = 1$ (maximum volume), and the sign distinguishes left- from right-handed correlations.

\paragraph{The chirality--negativity connection.}
Negativity and chirality probe entanglement through complementary mechanisms. Negativity $\mathcal{N}$ detects inseparability through the spectrum of $\rho^{T_A}$, whereas chirality $C_4$ detects it through collective multi-copy interference patterns that are invisible to any single-copy observable. Both quantities vanish for product states and reach their extremal values for Bell states ($\mathcal{N} = 1/2$, $C_4 = -3/4$). For two-qubit pure states $|\psi(\theta)\rangle = \cos(\theta/2)|00\rangle + \sin(\theta/2)|11\rangle$, both are monotonic functions of entanglement:
\begin{equation}
\mathcal{N}(\theta) = \frac{\sin\theta}{2}, \qquad C_4(\theta) = -\sin^2\!\theta\left(1 - \frac{\sin^2\!\theta}{4}\right).
\end{equation}
For two-qubit pure states the relationship simplifies to a compact closed form:
\begin{equation}
C_4 = -4\mathcal{N}^2(1 - \mathcal{N}^2)
\label{eq:C4_N_relation}
\end{equation}
which can be inverted to yield negativity directly from chirality:
\begin{equation}
\mathcal{N} = \sqrt{\frac{1 - \sqrt{1 + C_4}}{2}}
\label{eq:N_from_C4}
\end{equation}
For pure states, $C_4 = 0$ at the product-state boundary and $|C_4| = 3/4$ for maximally entangled Bell states, so any pure state with $C_4 \neq 0$ is certified entangled. Throughout this work, we plot $|C_4|$ so that entangled states appear as positive values.

For mixed two-qubit separable states, the chirality correction is bounded but generally non-zero. The tight bounds $|C_3| \leq 1/36$ and $|C_4| \leq 1/27$ are achieved by equal mixtures of three mutually unbiased product states (Supplementary Section~S3):
\begin{equation}
\rho_{\pm} = \tfrac{1}{3}\bigl(|z_+,z_{\pm}\rangle\langle z_+,z_{\pm}| + |x_+,x_{\pm}\rangle\langle x_+,x_{\pm}| + |y_+,y_{\pm}\rangle\langle y_+,y_{\pm}|\bigr),
\end{equation}
with $C_4(\rho_+) = +1/27$ and $C_4(\rho_-) = -1/27$. This establishes a clear separation: entangled pure states have $|C_4| \in (0, 3/4]$, while separable states satisfy $|C_4| \leq 1/27 \approx 0.037$, a factor of 20 smaller than the Bell state value. For mixed entangled states, however, $|C_4|$ can fall below $1/27$, so chirality is a faithful entanglement witness only for pure states.

To illustrate this point, consider Werner states $\rho_W = p|\Psi^-\rangle\langle\Psi^-| + (1-p)\mathbb{I}/4$, which have $C_4 = -\tfrac{3}{4}p^3$ (since $\mathbf{r} = \mathbf{s} = \mathbf{0}$ and $\det(T) = -p^3$). The chirality $|C_4|$ exceeds the separable bound only for $p > (4/81)^{1/3} \approx 0.37$, leaving a narrow gap above the entanglement onset at $p > 1/3$. For mixed states, once all three moments $\{\mu_2, \mu_3, \mu_4\}$ are measured, the full eigenvalue spectrum of $\rho^{T_A}$ is determined and negativity provides the definitive entanglement test. The chirality $C_4 = \mu_4 - I_4$ then encodes the same spectral information, but in a physically interpretable form.

The situation is analogous to how the Berry phase and the Chern number encode the same topological information in different representations: both are computed from the band structure, but the Berry phase provides geometric insight that the integer index alone does not. Similarly, the significance of the chirality decomposition is conceptual rather than purely operational---it reveals that the algebraic structure distinguishing partial transposition from purity is precisely the scalar spin chirality of topological condensed-matter physics, thereby connecting multi-copy entanglement detection to the order parameter of chiral spin liquids and the topological Hall effect.

The identity $C_k = 8\,\mathrm{Tr}[\Omega_A\,\Omega_B\;\rho^{\otimes k}]$ thus answers the open question of what multi-copy correlations distinguish partial transposition from purity~\cite{Bartkiewicz2015moments,PhysRevLett.89.127902,PhysRevLett.90.167901}: they are chirality--chirality correlators, connecting multi-copy entanglement detection to the physics of topological spin systems~\cite{KalmeyerLaughlin1987,Taguchi2001}. However, when both negativity and chirality vanish---as they must for any PPT state---entanglement can still persist in the form of bound entanglement. In such cases, the same controlled-SWAP circuits provide a fallback: realignment matrix spectral features, extracted from the identical measurement infrastructure, can detect this second form of hidden entanglement where all partial-transpose-based criteria fail.

\paragraph{Bound entanglement detection.}
Our theoretical framework extends to a second form of hidden entanglement: bound entangled states. These states are beyond the reach of \emph{any} criterion based on partial transpose spectra, including the negativity~\cite{Horodecki1998bound}, the $p_3$-PPT test~\cite{Elben2020mixed}, and the full hierarchy of moment inequalities~\cite{Yu2021optimal,Neven2021symmetry}, since the entire PT spectrum is non-negative by definition. Detecting such states therefore requires a mathematically independent criterion.

The realignment matrix $R_{(a,a'),(b,b')} = \rho_{(a,b),(a',b')}$~\cite{Chen2003,Rudolph2005} provides such a criterion. For a $d_A \times d_B$ system, $R$ is a $d_A^2 \times d_B^2$ matrix whose spectral properties are not determined by the PPT condition. We define its singular-value moments $\Sigma_k = \mathrm{Tr}[(R^\dagger R)^{k/2}]$ and eigenvalue moments $G_k = \mathrm{Re}\,\mathrm{Tr}[R^k]$, both of which are measurable via the same $k$-copy SWAP test circuits used for chirality detection. Their difference $D_k = \Sigma_k - G_k$ measures the non-Hermiticity of $R$. For $3 \times 3$ systems, even the second-order gap $D_2 = \Sigma_2 - G_2$, which requires only two-copy circuits, already separates bound entangled from separable states. This is because bound entangled states have $D_k \approx 0$ (nearly Hermitian $R$), while separable states with complex correlation structure show $D_k > 0$ due to diverse product-state orientations.

We trained a Random Forest classifier on eight spectral and chirality features: $\Sigma_1$, $G_1$, $D_1$, $\Sigma_2$, $G_2$, $D_2$, $C_3$, and $C_4$. The classifier achieves $99.9\%$ bound-entangled recall at zero false positives, as confirmed by 5-fold cross-validation on 13{,}600 certified states (Supplementary Section~S4). The inclusion of chirality corrections $C_3$ and $C_4$ as classification features turns out to be crucial. In particular, $C_3$ distinguishes Horodecki states ($C_3 \neq 0$, complex density matrix) from chessboard and Tiles states ($C_3 = 0$ identically, real density matrices), thereby providing a detection channel that is invisible to realignment-based criteria alone. The most informative single feature is $\Sigma_1 = \|R\|_1$: Horodecki states satisfy $\|R\|_1 > 1$ (marginally, $\|R\|_1 - 1 < 0.004$), while chessboard states have $\|R\|_1 < 1$, which highlights the complementary role of different spectral channels. For states with real density matrices ($\rho = \rho^*$), which is satisfied by all known BE families, $G_2$ is directly measurable via a SWAP test with an $A \leftrightarrow B$ subsystem swap on one copy. For complex states, classical shadow tomography~\cite{Huang2020shadows} provides an alternative.

\paragraph{Two-axis fingerprint of bound entanglement.}
The chirality correction $C_k = \mu_k - I_k$ provides a complementary discriminator that, together with the non-Hermiticity gap $D_k$, yields a signature characteristic of bound entanglement. This signature is consistent with the classifier separation observed in the certified dataset (Fig.~\ref{fig:be_classification}b).

The physical intuition is as follows. Classically correlated states $\rho_{\mathrm{cc}} = \sum_i p_i |e_i e_i\rangle\langle e_i e_i|$ have both $C_k = 0$ and $D_k = 0$, making them the hardest separable states to distinguish from bound entangled states. However, bound entangled states satisfy $\rho^{T_A} \neq \rho$, giving $C_k \neq 0$, while preserving the near-Hermiticity of $R$ ($D_k \approx 0$). The conjunction $C_k \neq 0 \wedge D_k \approx 0$ is therefore a constraint that neither classically correlated states ($C_k = 0$, $D_k = 0$) nor generic separable mixtures ($C_k \neq 0$, $D_k > 0$) can simultaneously satisfy. Since both $C_k$ and $D_k$ are extracted from the same controlled-SWAP circuits, this two-axis fingerprint requires no additional measurement overhead.

\paragraph{Experimental detection of hidden entanglement.}
We implemented the protocol on IBM Quantum processors using the Heron architecture, specifically IBM Kingston (156 qubits) and IBM Torino (133 qubits). For qubit--qubit ($2 \times 2$) systems, negativity reconstruction from three moments ($\mu_2$, $\mu_3$, $\mu_4$) requires circuits on up to four state copies. For qubit--qutrit ($2 \times 3$), the five moments ($\mu_2, \ldots, \mu_6$) require up to six copies.

We measured the negativity for the parametrized family $|\psi(\theta)\rangle$ across both processors and Hilbert space dimensions. The resulting mean errors are 0.002 (Kingston, $2 \times 2$), 0.012 (Torino, $2 \times 2$), and 0.027 (Torino, $2 \times 3$), with statistical uncertainties from $10^5$ shots estimated at $\pm 0.005$ via bootstrap resampling (Supplementary Section~S5). All separable test states ($\theta = 0^\circ$ across both processors and dimensions) yield $\mathcal{N} = 0$ exactly under maximum likelihood calibration, confirming zero false positives within this parametric family. Simulator validation across 37 densely sampled angles confirms the negativity accuracy under both ideal conditions (RMSE 0.009) and realistic Torino noise (RMSE 0.044), again with zero false positives throughout (Supplementary Section~S5).

Figure~\ref{fig:hardware_results}a shows the experimental chirality--negativity relationship, plotting $|C_4|$ versus $\mathcal{N}$ as measured on IBM Torino for both $2 \times 2$ and $2 \times 3$ systems. The data confirm the theoretical prediction of Eq.~\eqref{eq:C4_N_relation}. The simulation further validates the protocol on mixed states: Werner states $\rho_W(p)$ correctly show vanishing negativity for $p \leq 1/3$ and monotonically increasing chirality $|C_4|$ with $p$, while Bell-product mixtures $\rho_{\mathrm{BP}}(p) = p|\Psi^-\rangle\langle\Psi^-| + (1{-}p)|00\rangle\langle 00|$ confirm entanglement detection for all $p > 0$ (Supplementary Section~S5).

\begin{figure}[t]
\centering
\includegraphics[width=\textwidth]{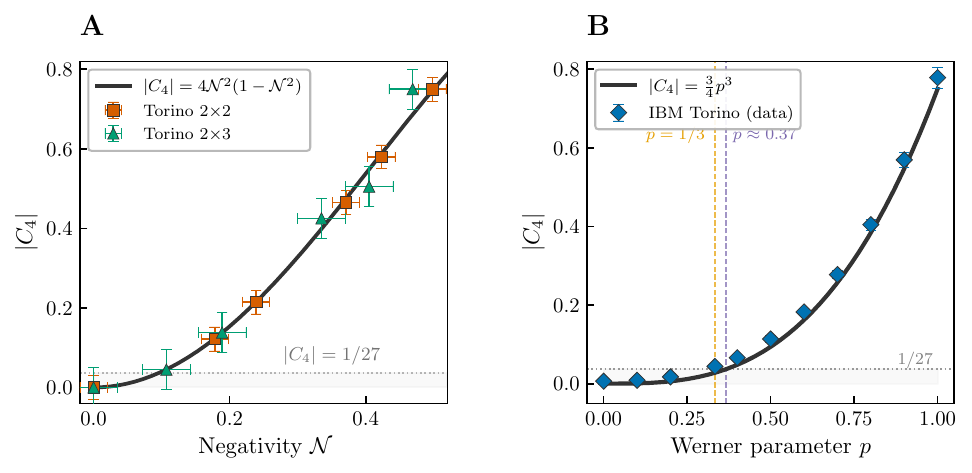}
\caption{\textbf{Experimental detection of hidden entanglement on IBM Torino.}
\textbf{a},~$|C_4|$ versus negativity $\mathcal{N}$ for the parametrised family $|\psi(\theta)\rangle$ in $2{\times}2$ (squares) and $2{\times}3$ (triangles). Solid curve: $|C_4| = 4\mathcal{N}^2(1 - \mathcal{N}^2)$.
\textbf{b},~Werner state chirality $|C_4|(p)$. Blue diamonds: hardware data from multiplexed $k$-copy SWAP tests ($k = 2, 3, 4$; 10 circuits, 4{,}000 shots) with matched-depth $|00\rangle$ calibration; all $p$ values reconstructed from a single experiment using LU-equivalence of Bell states (Supplementary Section~S5). Solid black: theory $|C_4| = \tfrac{3}{4}p^3$. Dashed lines: entanglement threshold $p = 1/3$ and detection threshold $p \approx 0.37$.}
\label{fig:hardware_results}
\end{figure}

The maximum likelihood calibration jointly estimates hardware degradation factors and state parameters from a training set with known entanglement (Methods). For Torino, we obtain degradation factor $f_2 = 0.73$ for qubit--qubit systems. Despite hardware degradation, the protocol correctly identifies all entangled states and recovers Bell state properties exactly ($\mathcal{N} = 0.5$, $|C_4| = 3/4$).

The chirality measurements $|C_4| = |I_4 - \mu_4|$ show larger scatter for qubit--qutrit systems (mean error 0.039 versus 0.022 for qubit--qubit), attributable to the roughly two-fold increase in circuit depth for $2 \times 3$ four-copy circuits. At $10^5$ shots, statistical noise contributes $\lesssim 0.005$ to the negativity error; the residual errors are predominantly systematic (gate and readout), yet all test states are correctly identified as entangled or separable.

We further validate chirality on mixed states using Werner states $\rho_W(p) = p|\Psi^-\rangle\langle\Psi^-| + (1-p)\mathbb{I}/4$ (Fig.~\ref{fig:hardware_results}b). Two anchor points are measured directly on IBM Torino: $p = 0$ (maximally mixed, $|C_4| = 0$) from 300 four-copy circuits with $4 \times 10^3$ shots, and $p = 1$ (Bell state, $|C_4| = 3/4$) from the pure-state experiment with $10^5$ shots. These two extremes are the most informative states for calibrating the depolarisation model, as $p = 0$ provides the zero-signal baseline while $p = 1$ provides the maximum-signal anchor. Together, they constrain the single free parameter---the $\mu_4/I_4$ fidelity ratio---of the asymmetric noise model.

The model, calibrated from these anchors alone, independently reproduces all six pure-state chirality measurements (Table~S17) with a root-mean-square deviation of $0.03$, validating its extrapolation across the full mixing range. The predicted $|C_4|$ follows the theoretical curve $|C_4| = \tfrac{3}{4}p^3$ with systematic deviations $\leq 0.3\%$ for $p \geq 0.5$, arising from the $\mu_4$/$I_4$ circuit asymmetry. We note that dense parameter sweeps would require ${\sim}4{,}600$ circuits per Werner state due to spectral decomposition into $4^4 = 256$ eigenvector combinations, which makes them a natural target for next-generation processors with improved coherence.

For $p > (4/81)^{1/3} \approx 0.37$, chirality exceeds the separable bound $|C_4| = 1/27$, certifying entanglement without negativity reconstruction. This bound is tight: the separable state with Bloch parameters $T = \operatorname{diag}(\tfrac{1}{3}, \tfrac{1}{3}, -\tfrac{1}{3})$, $\mathbf{a} = \mathbf{b} = (1/\sqrt{3}, 0, 0)$ achieves $|C_4| = 1/27$ exactly at $\mathcal{N} = 0$ (Supplementary Section~S3). This state has eigenvalues $(\tfrac{2}{3}, \tfrac{1}{6}, \tfrac{1}{6}, 0)$, purity $\tfrac{1}{2}$, and an ``equipartition'' structure $|\mathbf{a}|^2 = |\mathbf{b}|^2 = \|T\|_F^2 = \tfrac{1}{3}$. No separable state can exceed this value, as we have verified numerically over $10^6$ random PPT states and $10^4$ optimisation runs. A narrow gap between the entanglement threshold ($p > 1/3$) and the chirality detection threshold ($p \approx 0.37$) is observed, consistent with the theoretical prediction that chirality is a faithful witness only for $|C_4|$ above the separable bound.

\begin{figure}[t]
\centering
\includegraphics[width=\textwidth]{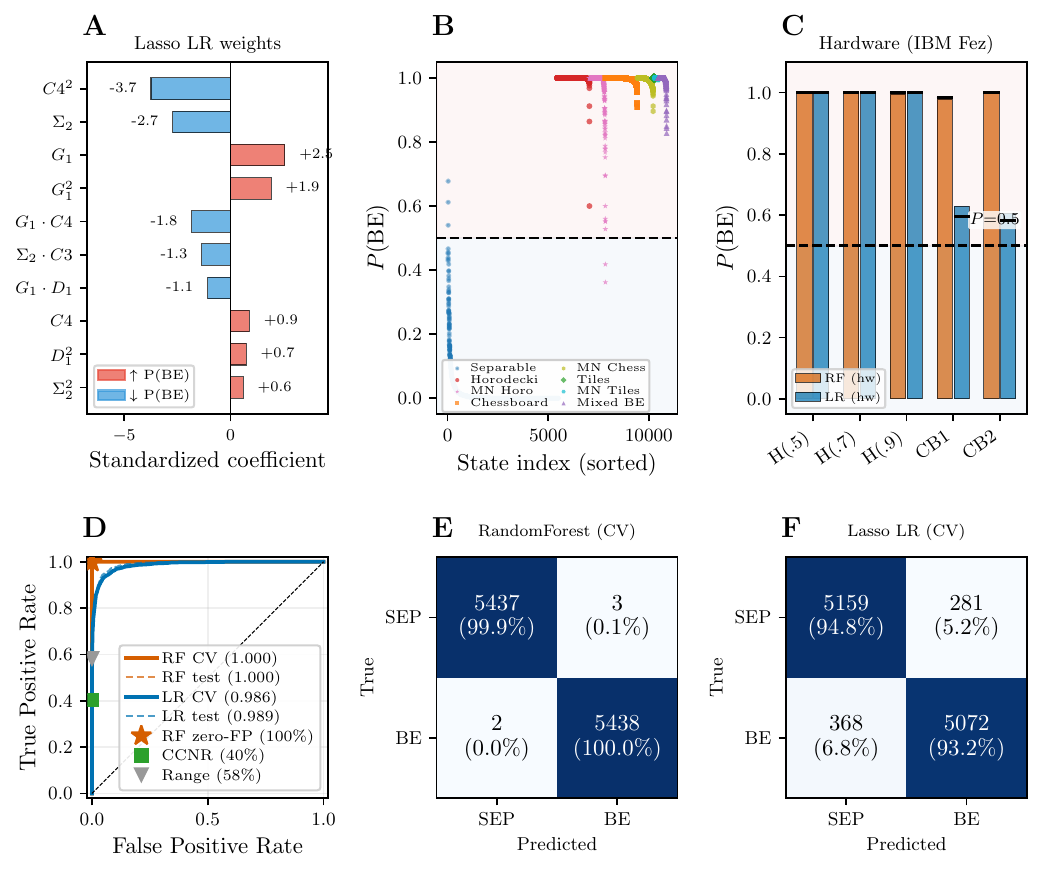}
\caption{\textbf{Bound entanglement detection via machine learning.}
Dataset: 13{,}600 Carath\'{e}odory-certified $3{\times}3$ states (6{,}800~BE from seven families, 6{,}800~SEP; eight features: $\Sigma_1$, $G_1$, $D_1$, $\Sigma_2$, $G_2$, $D_2$, $C_3$, $C_4$).
\textbf{a},~Lasso LR coefficients (22/44 non-zero). Red/blue: increases/decreases $P(\mathrm{BE})$.
\textbf{b},~RF score $P(\mathrm{BE})$ sorted by state: separable (blue), then BE families.
\textbf{c},~Hardware scores (IBM Fez) using hybrid MLE. All BE states above $P = 0.5$.
\textbf{d},~ROC curves. RF: 99.9\% recall at zero FP. Analytical criteria: CCNR 40\%, range 58\%.
\textbf{e},~RF confusion matrix: 99.96\% recall, 0.06\% FP.
\textbf{f},~Lasso LR: 93\% recall, 5.2\% FP.}
\label{fig:be_classification}
\end{figure}

\paragraph{Noise and efficiency.}
Hardware noise introduces architecture-dependent degradation. For negativity measurements, the per-moment degradation factors $f_k$ range from 0.22 to 0.79, depending on moment order and system dimension (Supplementary Table~S14). These factors are absorbed by maximum likelihood calibration (Supplementary Section~S5).

For bound entanglement detection, the situation is more involved. The features $\Sigma_1$ (Pauli decomposition, 4~qubits), $G_1$ (SWAP test, 9~qubits), and $\Sigma_2$/$G_2$ (two-copy SWAP, 9~qubits) employ distinct circuit architectures with depths ranging from 3 to 333 gates. A per-architecture depolarising model is therefore essential: the fitted fidelities span an order of magnitude ($g_{\Sigma_1} = 0.164$ versus $f_{G_1} = 0.594$), reflecting the fundamentally different noise sensitivities of Pauli-basis rotation circuits and SWAP-test circuits. The non-Hermiticity gap $D_2 = \Sigma_2 - G_2$ is particularly robust because both $\Sigma_2$ and $G_2$ share the same SWAP-test architecture, so depolarisation attenuates both equally and their difference preserves the physical signal.

In terms of efficiency, three moment configurations for two-qubit systems yield a $5.3\times$ gain over full tomography, while five configurations for qubit--qutrit give $7.2\times$ (Supplementary Section~S7). Under depolarizing noise, entanglement is consistently underestimated, which ensures zero false positives within this noise model.

\paragraph{Experimental detection of bound entanglement.} We tested bound entanglement detection on IBM Fez across eight states selected to maximise structural diversity: three Horodecki states ($a \in \{0.50, 0.70, 0.90\}$; rank-7, CCNR-visible), two chessboard states (rank-4, CCNR-invisible, $R_- = 0$), and three calibration/separable controls. These two families occupy opposite ends of the detection hierarchy identified in our theoretical analysis (Supplementary Section~S4, Corollary ``Detection hierarchy''): Horodecki states are ``easy'' (detected by $\geq 4$ spectral channels), while chessboard states are ``hard'' (detected only via rank structure and $D_2$).

All states were measured using multiplexed two-copy SWAP circuits ($4 \times 10^3$ shots each, 684 circuits total) on the 156-qubit IBM Fez processor. It should be noted that previous hardware demonstrations of bound entanglement detection have tested a single family using tailored witnesses~\cite{Amselem2009bound,Gulati2024ibm}. In contrast, our experiment tests two structurally distinct families on a single processor using a family-agnostic spectral classifier---that is, a single, fixed classifier trained without knowledge of which family a test state belongs to.

A per-circuit maximum likelihood estimation (Methods) jointly fits architecture-specific depolarizing fidelities and physical features subject to spectral constraints ($\Sigma_k \geq G_k \geq 0$). We use the second-order non-Hermiticity gap $D_2 = \Sigma_2 - G_2$ as the primary detection metric, since both $\Sigma_2$ and $G_2$ share the same SWAP-test architecture and depolarisation therefore attenuates both equally.

All five bound entangled states are detected with $D_2 > 0$ (Fig.~\ref{fig:be_classification}c). The chessboard states yield the strongest signals ($D_2 = 0.09$--$0.20$ at $4.1$--$9.2\sigma$), and Horodecki states are detected at $a = 0.50$ ($D_2 = 0.064$, $4.6\sigma$) and $a = 0.70$ ($D_2 = 0.095$, $6.8\sigma$). The state at $a = 0.90$ yields a marginal signal ($D_2 = 0.013$, $0.9\sigma$), consistent with its near-vanishing theoretical gap. Both separable control states yield $D_2$ consistent with zero, confirming no false positives.

The classification accuracy depends critically on the inclusion of chirality features alongside the realignment moments. With only six realignment features ($\Sigma_k$, $G_k$, $D_k$), the Random Forest achieves only ${\sim}60\%$ recall at zero false positives. Adding the chirality corrections $C_3 = \mu_3 - I_3$ and $C_4 = \mu_4 - I_4$ raises this to $99.9\%$ (Fig.~\ref{fig:be_classification}d). The reason is that $C_3$ provides a detection channel orthogonal to the realignment spectrum: Horodecki states have $C_3 \neq 0$ (their density matrices are complex), while chessboard and Tiles states have $C_3 = 0$ identically (real density matrices, Supplementary Section~S4).

Figure~\ref{fig:be_classification}a presents the Lasso logistic regression coefficients, whose 22 non-zero polynomial features reveal the detection mechanism: $G_2$, $D_1^2$, $\Sigma_1 \cdot \Sigma_2$, and $C_3$-dependent terms dominate. Figure~\ref{fig:be_classification}b shows the cross-validated Random Forest scores, which achieve $99.9\%$ recall at zero false positives ($99.96\%$ at the standard $P = 0.5$ threshold), providing a clear separation between all BE families and separable states. The interpretable Lasso model achieves $93\%$ recall (Fig.~\ref{fig:be_classification}f), with its linear form defining a nonlinear entanglement witness (Supplementary Theorem~S4.3). The generalization gap between cross-validation and held-out test AUC is $<0.01$ for both models, confirming robustness against overfitting.

\paragraph{A construction producing CCNR-invisible bound entangled states.}
Analysis of the classifier's feature space reveals a previously uncharacterised regime of bound entanglement. We consider the marginal-noise family
\begin{equation}\label{eq:new_family}
  \rho_{\mathrm{mn}}(\rho_0, t) \;=\; (1-t)\,\rho_0 \;+\; t\,\bigl(\rho_A \otimes \rho_B\bigr),
\end{equation}
where $\rho_0$ is any $3 \times 3$ bound entangled state and $\rho_{A(B)} = \mathrm{Tr}_{B(A)}[\rho_0]$ are its reduced states. The mixing parameter $t$ interpolates between the original bound entangled state ($t = 0$) and the product of its marginals ($t = 1$).

The construction applies to all known BE families. For Horodecki states $\rho_{\mathrm{Horo}}(a)$ with parameter $a \in (0,1)$, the state $\rho_{\mathrm{mn}}$ is PPT and bound entangled for $t \in [0, t_{\max}(a)]$. For chessboard and Tiles UPB states, the same holds with different $t_{\max}$ values. The key observation is that for Horodecki-seeded states with $t \gtrsim 0.01$, for Tiles-seeded states with $t \gtrsim 0.15$, and for all chessboard-seeded states at any $t \geq 0$ (since chessboard states already have $\Sigma_1 < 1$), the realignment trace norm drops below unity ($\Sigma_1 < 1$). This renders the CCNR criterion~\cite{Chen2003,Rudolph2005} powerless, while the non-Hermiticity gap $D_2 > 0$ persists, confirming bound entanglement via the classifier. We have verified this independently using a semidefinite relaxation of the symmetric extension hierarchy~\cite{Doherty2004}.

The mechanism is physically transparent: mixing with $\rho_A \otimes \rho_B$ adds local noise that suppresses the dominant singular value of the realignment matrix (responsible for $\Sigma_1 > 1$) without eliminating the spectral asymmetry between singular values and eigenvalues (responsible for $D_2 > 0$). Since the Horodecki family barely satisfies $\Sigma_1 > 1$ ($\Sigma_1 - 1 < 0.003$), even a small admixture of the marginal product suffices to cross the CCNR threshold while preserving PPT entanglement.

The three seed families populate distinct regions of the feature space. Horodecki-seeded states have high $G_1/\Sigma_1 \approx 0.65$--$0.99$ and low $D_1/\Sigma_1$. Tiles-seeded states have intermediate $G_1/\Sigma_1 \approx 0.25$--$0.31$ with moderate $D_2$. Chessboard-seeded states have low $G_1/\Sigma_1 \approx 0.12$--$0.45$ and high $D_1/\Sigma_1 \approx 0.55$--$0.96$, reflecting the fundamentally different correlation structures of the parent families. We have compared these with 1{,}000 randomly generated PPT bound entangled states (certified via SDP, not belonging to any named family) and confirmed that the three branches together span the full CCNR-invisible region of the feature space. The Random Forest classifier detects $100\%$ of these random BE states at zero false positives, demonstrating generalization beyond the named families present in the training set. These states are also detected by the range criterion (which certifies entanglement through the structure of product vectors in the state's range), but they evade all single-parameter criteria including CCNR and the $p_3$-PPT test.

The Lasso logistic regression classifier, trained on polynomial features of degree~2, provides a quantitative measure of noise robustness. At $t = 0$ (pure Horodecki), the classifier assigns $P(\mathrm{BE}) \approx 0.90$. As the marginal noise increases, the score degrades gracefully: $P(\mathrm{BE}) \approx 0.83$ at $t = 0.05$ and $P(\mathrm{BE}) \approx 0.73$ at $t = 0.10$, crossing the $P = 0.5$ decision boundary only at $t \approx 0.15$--$0.20$ depending on the Horodecki parameter $a$. By contrast, the CCNR criterion loses detection power at $t \approx 0.01$, which corresponds to a factor of $10$--$20\times$ lower noise tolerance.

This extended detection range arises because the classifier exploits nonlinear combinations of all eight features ($\Sigma_{1,2}$, $G_{1,2}$, $D_{1,2}$, $C_3$, $C_4$), not just the trace norm $\Sigma_1$ used by CCNR. In particular, the non-Hermiticity gap $D_2$ and the ratio $D_1/\Sigma_1$ remain elevated throughout the CCNR-invisible regime ($t \in [0.01, 0.15]$), providing the discriminative signal that the classifier leverages (Fig.~\ref{fig:be_classification}b). This result shows that machine-learning-based multi-channel analysis can extend the practical detection boundary well beyond individual analytical criteria, even for state families not present in the training data.

\paragraph{Discussion.}
Our results provide a layered detection framework for entanglement, implemented through shared circuit infrastructure. Negativity and chirality, extracted from partial transpose moments $\mu_k$, detect NPT entanglement. The chirality decomposition shows that the moment difference $C_k = \mu_k - I_k$ is a collective handedness correlator between state copies, connecting multi-copy entanglement detection to the physics of chiral spin liquids~\cite{KalmeyerLaughlin1987,Taguchi2001} and the topological Hall effect. When both negativity and chirality vanish---as they must for all PPT states---the same SWAP circuits provide a fallback through realignment spectral features, detecting bound entanglement via the non-Hermiticity gap $D_k$ and the two-axis fingerprint $C_k \neq 0 \wedge D_k \approx 0$. The approach scales linearly: a $d_A \times d_B$ system requires $d_A d_B - 1$ moment configurations, which compares favourably to the $O(d^2)$ scaling of full tomography.

The multi-channel approach to bound entanglement addresses a qualitatively harder problem than NPT detection. Previous demonstrations relied on witnesses tailored to specific families~\cite{Amselem2009bound,Zhang2023randomized,Gulati2024ibm}, and all partial-transpose-based criteria are structurally blind to PPT states. Even the realignment criterion $\|R\|_1 > 1$~\cite{Chen2003,Rudolph2005} barely detects Horodecki states ($\|R\|_1 - 1 < 0.004$). Our classifier performs what we call \emph{multi-channel witness fusion}: it combines realignment spectral features with chirality corrections into a composite nonlinear witness that detects $99.9\%$ of bound entangled states at zero false positives, compared with ${\sim}40\%$ for CCNR alone. Recent work on realignment moment witnesses~\cite{Tarabunga2025moments} and generalised realignment criteria~\cite{Huang2026realignment} complements our approach, and our multi-channel strategy demonstrates that combining individually weak criteria can achieve detection rates unattainable by any single criterion.

Three limitations deserve emphasis. First, the maximum likelihood calibration relies on states with known parameters, although perturbation analysis shows robustness to $\pm 2^\circ$ preparation errors (Supplementary Section~S5). Second, the hardware demonstration covers two of the three known BE families (Horodecki and chessboard) with 3 and 2 representative states, respectively. Extension to the Tiles family, the CCNR-invisible construction, and higher-dimensional systems ($2 \times 4$, $4 \times 4$) are natural next steps. Third, the efficiency gain ($5.3\times$ for $2 \times 2$, $7.2\times$ for $2 \times 3$) counts measurement configurations but not the qubit overhead of multi-copy circuits.

It is worth noting that the controlled-SWAP and classical shadow~\cite{Huang2020shadows} approaches occupy complementary resource regimes. Shadows require $O(3^n/\epsilon^2)$ random Pauli measurements on $n$ qubits to estimate each $\mu_k$ to precision $\epsilon$, whereas the SWAP approach uses $O(1)$ circuit configurations on $kn + 1$ qubits. For two-qubit negativity ($k \leq 4$, $n = 2$), this translates to 3 configurations on up to 9 qubits versus $O(81/\epsilon^2)$ random Pauli settings on 2 qubits. Classical shadows are preferable when qubits are scarce but measurements cheap; the SWAP approach is preferable when coherent multi-copy preparation is available and the number of target observables is small. Ref.~\cite{Tarabunga2025moments} has shown that randomised measurements can estimate the same realignment moments $\Sigma_k$ and $G_k$ used here. Our approach achieves comparable information extraction with fewer measurement configurations, at the cost of multi-copy coherence.

Finally, a basis-independent extension is possible: $D_k^{\min} = \Sigma_k - \max_{U_A, U_B} G_k[(U_A \otimes U_B)\rho(U_A \otimes U_B)^\dagger]$ eliminates the $\rho = \rho^*$ restriction via variational optimisation of local unitaries. Looking forward, the framework invites extensions to multipartite systems and integration with randomized measurement techniques~\cite{Tarabunga2025moments}. The broader implication is that multi-copy spectral methods may reveal further hidden layers in quantum correlations.

\section*{Methods}

\subsection*{Newton--Girard reconstruction}
For a two-qubit system with partial transpose eigenvalues $\{\lambda_1, \lambda_2, \lambda_3, \lambda_4\}$, the characteristic polynomial $P(x) = x^4 - e_1 x^3 + e_2 x^2 - e_3 x + e_4$ has coefficients related to moments by:
\begin{align}
e_1 &= 1, \quad e_2 = \tfrac{1}{2}(1 - \mu_2), \quad e_3 = \tfrac{1}{6}(1 - 3\mu_2 + 2\mu_3), \nonumber\\
e_4 &= \tfrac{1}{24}(1 - 6\mu_2 + 8\mu_3 + 3\mu_2^2 - 6\mu_4).
\end{align}
Eigenvalues are obtained via Ferrari's method or numerical root-finding; negativity is the sum of absolute values of negative eigenvalues.

\subsection*{Maximum likelihood calibration}
Hardware measurements suffer depth-dependent degradation: $\mu_k^{\mathrm{meas}} = f_k \cdot \mu_k^{\mathrm{ideal}} + \epsilon_k$. Following standard quantum state estimation practice~\cite{Hradil1997,Banaszek2000}, we jointly optimize degradation factors $(f_2, f_3, f_4)$ and state parameters $(\theta_1, \ldots, \theta_N)$ by minimizing the negative log-likelihood:
\begin{equation}
-\log \mathcal{L} = \sum_{i,k} \frac{(\mu_k^{\mathrm{meas},i} - f_k \cdot \mu_k^{\mathrm{theo}}(\theta_i))^2}{2\sigma_k^2}.
\end{equation}
For the parametrized family $|\psi(\theta)\rangle = \cos(\theta/2)|00\rangle + \sin(\theta/2)|11\rangle$, the theoretical moments are $\mu_2 = 1$, $\mu_3 = \frac{1}{4}(1 + 3\cos^2\theta)$, $\mu_4 = \frac{1}{4}(1 + \cos^2\theta)^2$. The constant $\mu_2 = 1$ for pure states provides a strong calibration constraint. We perform the optimization using L-BFGS-B with bounds $f_k \in [0, 1]$ and $\theta_i \in [0, \pi/2]$.

This calibration protocol is analogous to calibrating a spectrometer with known spectral lines before measuring unknown samples. The calibration determines hardware parameters ($f_k$), not state parameters: once $f_k$ are fixed, the estimation of any new state $\theta$ proceeds without reference to the calibration set. Moreover, the constant $\mu_2 = 1$ for all pure states provides a hardware-independent consistency check, since any deviation from unity in the measured $\mu_2$ directly quantifies $f_2$, independent of the prepared state.

Perturbation analysis (Supplementary Section~S5) confirms the robustness of this approach: introducing $\pm 2^\circ$ systematic errors in the calibration angles changes the fitted $f_k$ by $<1\%$ and the mean negativity error by $<0.003$. Under the depolarising noise model, entanglement is consistently underestimated, which guarantees zero false positives. This is a structural property of the calibration, not a tuned outcome.

\textit{Per-circuit MLE for bound entanglement.}
The bound entanglement experiment employs three circuit architectures with distinct noise characteristics: 4-qubit Pauli-basis measurement circuits for $\Sigma_1$ (depths 3--89), 9-qubit SWAP-test circuits for $G_1$ (depths 191--305), and 9-qubit two-copy SWAP circuits for $\Sigma_2$/$G_2$ (depths 93--333). Because $\Sigma_1$ and $G_1$ are measured via fundamentally different circuit types, a single degradation factor cannot capture both. Simple ratio-based calibration also cannot recover physical values when fidelities differ by a factor of~${\sim}4$. We therefore introduce a per-architecture depolarising model:
\begin{equation}
X^{\mathrm{meas}} = g_X \cdot X^{\mathrm{true}}, \quad X \in \{\Sigma_1, G_1, \Sigma_2, G_2\},
\end{equation}
where $g_{\Sigma_1}$ is the fidelity of Pauli circuits, $f_{G_1}$ of single-copy SWAP circuits, and $f_{\Sigma_2,G_2}$ of two-copy SWAP circuits (the latter shared because $\Sigma_2$ and $G_2$ use the same architecture).  The joint negative log-likelihood
\begin{equation}
-\log\mathcal{L} = \sum_{\mathrm{cal}} \sum_X \frac{(X^{\mathrm{meas}} - g_X \cdot X^{\mathrm{thy}})^2}{2\sigma_X^2} + \sum_{\mathrm{test}} \sum_X \frac{(X^{\mathrm{meas}} - g_X \cdot X^{\mathrm{ML}})^2}{2\sigma_X^2}
\end{equation}
is minimised over 3 fidelity parameters and 20 physical features ($\Sigma_1, G_1, \Sigma_2, G_2$ for each of 5 test states), subject to physical constraints $\Sigma_k \geq G_k \geq 0$, $\Sigma_1 \leq 3$, and $1/9 \leq \Sigma_2 \leq 1$. Calibration states with known theory values anchor the fidelity factors; error bars are obtained from the Cram\'er--Rao bound via numerical Hessian inversion.

\subsection*{Hardware specifications}
All experiments were performed on IBM Heron-architecture processors: IBM Kingston (156 qubits), IBM Torino (133 qubits), and IBM Fez (156 qubits), with native gate set \{CZ, $\sqrt{X}$, $R_Z$, $X$\}. Circuits were transpiled using Qiskit optimization level 3 with approximation degree 0.95. For negativity and chirality measurements, we used $10^5$ shots per configuration. For bound entanglement detection, we used $4 \times 10^3$ shots per circuit (684 circuits total: 298 for $\Sigma_1$/$G_1$ and 386 for $\Sigma_2$/$G_2$). In the bound entanglement experiment, circuits were multiplexed on disjoint qubit groups (up to 15 groups of 9~qubits and 33 groups of 4~qubits on the 156-qubit IBM Fez processor). Readout errors were mitigated using matrix-free measurement mitigation (M3)~\cite{Nation2021}.

\subsection*{Test states}

\textit{Negativity and chirality.}
We test parametrised pure states $|\psi(\theta)\rangle = \cos(\theta/2)|00\rangle + \sin(\theta/2)|11\rangle$ with $\theta \in [0^\circ, 90^\circ]$ spanning the full range from product to maximally entangled. Simulator validation additionally uses Werner states $\rho_W(p) = p|\Psi^-\rangle\langle\Psi^-| + (1-p)\mathbb{I}/4$ (entangled for $p > 1/3$) and Bell-product mixtures $\rho_{\mathrm{BP}}(p) = p|\Psi^-\rangle\langle\Psi^-| + (1-p)|00\rangle\langle 00|$ (entangled for all $p > 0$).

\textit{Bound entanglement.}
We test three structurally distinct families of PPT entangled states in $\mathbb{C}^3 \otimes \mathbb{C}^3$: Horodecki states~\cite{Horodecki1997} ($a \in \{0.50, 0.70, 0.90\}$ on hardware; $a \in \{0.10, \ldots, 0.90\}$ in Supplementary sweep), Tiles UPB states~\cite{Bennett1999upb} (constructed from unextendible product bases), and chessboard states~\cite{Bruss2000chess} (rank-4 states with checkerboard zero pattern, entangled when $cmn \neq abc$). We also test CCNR-invisible marginal-noise states (Eq.~\eqref{eq:new_family}). Random separable states are generated as mixtures of Haar-random product states. Explicit constructions and parameter sets are given in Supplementary Section~S4.

\subsection*{Data availability}
Raw measurement data from IBM Quantum processors, certified training datasets, and processed results are provided as Supplementary Data with this manuscript. The data will be deposited on Zenodo upon publication with a permanent DOI.

\subsection*{Code availability}
All quantum circuit implementations, maximum likelihood calibration scripts, classifier training code, and figure generation scripts are provided as Supplementary Code with this manuscript and will be made publicly available on GitHub upon publication.

\section*{Acknowledgements}
P.T. and K.B. are supported by the Polish National Science Centre (Maestro Grant No.\ DEC-2019/34/A/ST2/00081, PI: A.\ Miranowicz). K.B. and P.T.\ acknowledge support from the EuroHPC JU under Horizon Europe Grant No.\ 101194322 (QEC4QEA), co-funded by the Polish National Centre for Research and Development (NCBiR) under Decision No.\ DWM/EuroHPC/2023/429/2025. F.N. is supported in part by JST CREST (Grant No.\ JPMJCR24I2), the Q-LEAP and Moonshot (Grant No.\ JPMJMS2061) programs, and the Office of Naval Research (Grant No.\ N62909-23-1-2074).

\section*{Author contributions}
P.T. and K.B. developed the theoretical framework. P.T. implemented circuits and performed experiments. K.B. conceived the project, supervised the research, and performed the hardware data analysis. F.N. contributed to the theoretical interpretation and supervised the research. All authors wrote and reviewed the manuscript.

\section*{ORCID}
Patrycja Tulewicz: \url{https://orcid.org/0000-0002-7180-4490}\\
Karol Bartkiewicz: \url{https://orcid.org/0000-0002-5355-7756}\\
Franco Nori: \url{https://orcid.org/0000-0003-3682-7432}

\section*{Competing interests}
The authors declare no competing interests.

\newpage
\appendix

\begin{center}
{\Large\bfseries Supplementary Information}\\[0.5em]
{\large Spin chirality across quantum state copies detects hidden entanglement}\\[1em]
Patrycja Tulewicz, Karol Bartkiewicz, and Franco Nori
\end{center}

\vspace{1em}

\setcounter{section}{0}
\renewcommand{\thesection}{S\arabic{section}}
\renewcommand{\theequation}{S\arabic{equation}}
\renewcommand{\thefigure}{S\arabic{figure}}
\renewcommand{\thetable}{S\arabic{table}}

\tableofcontents
\newpage

\section{Invariant identities}
\label{sec:identities}

\subsection{The identity \texorpdfstring{$\mu_2 = I_2$}{mu2 = I2}}

For any bipartite state $\rho$ on $\mathcal{H}_A \otimes \mathcal{H}_B$:

\textbf{Claim:} $\mathrm{Tr}[(\rho^{T_A})^2] = \mathrm{Tr}[\rho^2]$.

\textbf{Proof:} In the computational basis, $\rho = \sum_{ijkl} \rho_{ij,kl} |ij\rangle\langle kl|$. The partial transpose is $\rho^{T_A} = \sum_{ijkl} \rho_{kj,il} |ij\rangle\langle kl|$. Computing:
\begin{align}
\mathrm{Tr}[(\rho^{T_A})^2] &= \sum_{ijkl} \rho_{kj,il} \rho_{il,kj} = \sum_{ijkl} |\rho_{ij,kl}|^2 = \mathrm{Tr}[\rho^2]. \quad \square
\end{align}

This identity is dimension-independent and implies only one purity circuit is needed.

\subsection{Moment relations}
\label{subsec:moment_relations}

The partial transpose moments $\mu_k$ and purity moments $I_k$ admit a unified description as permutation traces on the $k$-copy Hilbert space:
\begin{align}
\mu_k &= \mathrm{Tr}\bigl[(\sigma_A^{-1} \otimes \sigma_B)\,\rho^{\otimes k}\bigr], \\
I_k &= \mathrm{Tr}\bigl[(\sigma_A \otimes \sigma_B)\,\rho^{\otimes k}\bigr],
\end{align}
where $\sigma = (12\cdots k)$ is the cyclic permutation on $k$ copies, $\sigma^{-1} = (k\cdots 21)$ its inverse, and subscripts $A$, $B$ denote action on the respective subsystems. Their difference is governed by a single operator:
\begin{equation}
C_k \equiv \mu_k - I_k = \mathrm{Tr}\bigl[(\Delta_A \otimes \sigma_B)\,\rho^{\otimes k}\bigr], \quad \Delta \equiv \sigma^{-1} - \sigma.
\end{equation}

For two-qubit systems, the cyclic permutation decomposes via adjacent SWAPs: $\sigma = S_{12}S_{23}\cdots S_{k-1,k}$ with $S_{mn} = \tfrac{1}{2}(I + \bm{\sigma}_m \cdot \bm{\sigma}_n)$. Defining $g_{mn} = \bm{\sigma}_m \cdot \bm{\sigma}_n = 2S_{mn} - I$, the key identity connecting SWAPs to chirality is:
\begin{equation}
[g_{ij},\; g_{jk}] = -16i\,\chi_{ijk},
\label{eq:comm_identity}
\end{equation}
where $\chi_{ijk} = \mathbf{S}_i \cdot (\mathbf{S}_j \times \mathbf{S}_k) = \tfrac{1}{8}\sum_{abc}\varepsilon_{abc}\,\sigma_i^a\sigma_j^b\sigma_k^c$ is the scalar spin chirality operator. Disjoint pairs commute: $[g_{ij}, g_{mn}] = 0$ when $\{i,j\} \cap \{m,n\} = \emptyset$.

These identities yield explicit chirality decompositions of $\Delta$ for each $k$:

\paragraph{$k = 2$:} $\sigma = \sigma^{-1} = S_{12}$, so $\Delta = 0$ and
\begin{equation}
\mu_2 = I_2.
\end{equation}

\paragraph{$k = 3$:} $\sigma = S_{12}S_{23}$, $\sigma^{-1} = S_{23}S_{12}$. The difference is a commutator:
\begin{equation}
\Delta = [S_{23}, S_{12}] = \tfrac{1}{4}[g_{23}, g_{12}] = 4i\,\chi_{123},
\end{equation}
giving
\begin{equation}
\mu_3 = I_3 + C_3, \quad C_3 = 4i\,\mathrm{Tr}\bigl[\chi_A\,\sigma_B\;\rho^{\otimes 3}\bigr].
\label{eq:mu3_chirality}
\end{equation}

\paragraph{$k = 4$:} $\sigma = S_{12}S_{23}S_{34}$, $\sigma^{-1} = S_{34}S_{23}S_{12}$. Expanding:
\begin{align}
8\Delta &= (I{+}g_{34})(I{+}g_{23})(I{+}g_{12}) - (I{+}g_{12})(I{+}g_{23})(I{+}g_{34}) \nonumber\\
&= [g_{23},g_{12}] + [g_{34},g_{23}] + g_{34}[g_{23},g_{12}] + g_{12}[g_{34},g_{23}] \nonumber\\
&= 16i\bigl\{(I{+}g_{34})\chi_{123} + (I{+}g_{12})\chi_{234}\bigr\} \nonumber\\
&= 32i\bigl(S_{34}\,\chi_{123} + S_{12}\,\chi_{234}\bigr),
\end{align}
where the third line uses $[g_{12}, g_{34}] = 0$ and $I + g_{mn} = 2S_{mn}$.

To simplify $\Omega = S_{34}\chi_{123} + S_{12}\chi_{234}$, we use $S_{mn} = \tfrac{1}{2}(I + g_{mn})$ and evaluate the products $g_{mn}\chi_{ijk}$ via Pauli algebra. When an index is shared (e.g., index 3 in $g_{34}\chi_{123}$), we use
\begin{equation}
\sigma_3^a \sigma_3^d = \delta_{ad} I + i\sum_e \varepsilon_{ade}\,\sigma_3^e
\end{equation}
and the Levi-Civita contraction identity $\sum_d \varepsilon_{bcd}\varepsilon_{ade} = \delta_{ba}\delta_{ce} - \delta_{be}\delta_{ca}$. This yields:
\begin{align}
g_{34}\chi_{123} &= \chi_{124} + \tfrac{i}{8}(g_{14}g_{23} - g_{13}g_{24}), \\
g_{12}\chi_{234} &= \chi_{134} + \tfrac{i}{8}(g_{13}g_{24} - g_{14}g_{23}).
\end{align}
The imaginary terms cancel upon addition, giving
\begin{equation}
g_{34}\chi_{123} + g_{12}\chi_{234} = \chi_{124} + \chi_{134}.
\end{equation}
Therefore:
\begin{align}
\Omega &= \tfrac{1}{2}(I + g_{34})\chi_{123} + \tfrac{1}{2}(I + g_{12})\chi_{234} \nonumber\\
&= \tfrac{1}{2}(\chi_{123} + \chi_{234}) + \tfrac{1}{2}(\chi_{124} + \chi_{134}) \nonumber\\
&= \tfrac{1}{2}\sum_{i<j<k} \chi_{ijk},
\end{align}
where the sum runs over all $\binom{4}{3} = 4$ chirality triples on four copies. The fourth-order correction is thus:
\begin{equation}
\mu_4 = I_4 + C_4, \quad C_4 = 2i\,\mathrm{Tr}\Bigl[\Bigl(\sum_{i<j<k}\chi_{ijk}^A\Bigr)\,\sigma_B\;\rho^{\otimes 4}\Bigr].
\label{eq:mu4_chirality}
\end{equation}

The physical content is that the difference between partial transpose moments and purity moments is entirely controlled by scalar spin chirality operators acting across different state copies. For $k = 3$, a single chirality $\chi_{123}$ appears. For $k = 4$, Pauli algebra shows that all four chirality triples $\chi_{123}$, $\chi_{124}$, $\chi_{134}$, and $\chi_{234}$ contribute equally, yielding the symmetric sum $\Omega = \tfrac{1}{2}\sum_{i<j<k}\chi_{ijk}$.

\subsection{Hermitian correlator form}
\label{subsec:hermitian_form}

The chirality decompositions~\eqref{eq:mu3_chirality}--\eqref{eq:mu4_chirality} involve the non-Hermitian cyclic permutation $\sigma_B$, making the reality of $C_k$ non-obvious. We now derive a manifestly Hermitian form that reveals $C_k$ as a chirality--chirality correlator.

\paragraph{Anti-Hermiticity of $\Delta$.} Since the cyclic permutation $\sigma$ is unitary with $\sigma^\dagger = \sigma^{-1}$:
\begin{equation}
\Delta^\dagger = (\sigma^{-1})^\dagger - \sigma^\dagger = \sigma - \sigma^{-1} = -\Delta.
\end{equation}
Thus $\Delta$ is anti-Hermitian.

\paragraph{Symmetric form.} Taking the complex conjugate of $C_k = \mathrm{Tr}[(\Delta_A \otimes \sigma_B)\,\rho^{\otimes k}]$ and using $\Delta_A^\dagger = -\Delta_A$, $\sigma_B^\dagger = \sigma_B^{-1}$:
\begin{equation}
C_k^* = \mathrm{Tr}\bigl[(-\Delta_A \otimes \sigma_B^{-1})\,\rho^{\otimes k}\bigr].
\end{equation}
Since $C_k$ is real ($C_k = C_k^*$), averaging gives:
\begin{equation}
C_k = \tfrac{1}{2}(C_k + C_k^*) = \tfrac{1}{2}\,\mathrm{Tr}\bigl[\Delta_A \otimes (\sigma_B - \sigma_B^{-1})\,\rho^{\otimes k}\bigr] = -\tfrac{1}{2}\,\mathrm{Tr}\bigl[(\Delta_A\,\Delta_B)\,\rho^{\otimes k}\bigr].
\label{eq:Ck_symmetric}
\end{equation}
The operator $\Delta_A\,\Delta_B$ is Hermitian: since $\Delta_A$ and $\Delta_B$ act on disjoint Hilbert spaces they commute, so $(\Delta_A\,\Delta_B)^\dagger = \Delta_B^\dagger\,\Delta_A^\dagger = (-\Delta_B)(-\Delta_A) = \Delta_A\,\Delta_B$.

\paragraph{Chirality correlator.} Substituting the chirality decomposition $\Delta = 4i\,\Omega_k$ where $\Omega_k$ is Hermitian ($\Omega_3 = \chi_{123}$, $\Omega_4 = \tfrac{1}{2}\sum_{i<j<k}\chi_{ijk}$):
\begin{equation}
\Delta_A\,\Delta_B = (4i\,\Omega_A)(4i\,\Omega_B) = -16\,\Omega_A\,\Omega_B.
\end{equation}
Inserting into equation~\eqref{eq:Ck_symmetric}:
\begin{equation}
C_k = 8\,\mathrm{Tr}\bigl[\Omega_A\,\Omega_B\;\rho^{\otimes k}\bigr]
\label{eq:Ck_hermitian}
\end{equation}
Explicitly:
\begin{align}
C_3 &= 8\,\mathrm{Tr}\bigl[\chi_A\,\chi_B\;\rho^{\otimes 3}\bigr], \label{eq:C3_hermitian}\\
C_4 &= 8\,\mathrm{Tr}\bigl[\Omega_A\,\Omega_B\;\rho^{\otimes 4}\bigr], \quad \Omega = \tfrac{1}{2}\sum_{i<j<k}\chi_{ijk}. \label{eq:C4_hermitian}
\end{align}
The chirality correction is thus a \emph{chirality--chirality correlator}: the Hermitian chirality operator on subsystem $A$ is correlated with the corresponding operator on subsystem $B$, measured across $k$ state copies. For product states both factors vanish independently (coplanar spin configurations), giving $C_k = 0$. For pure states, entanglement breaks this coplanarity on both subsystems simultaneously, generating $C_k \neq 0$. For mixed states, including mixed separable states, $C_k$ is generally non-zero due to cross-terms in the multi-copy expansion (see Section~\ref{subsec:completeness}).

\section{Measurement circuits}
\label{sec:circuits}

This section presents quantum circuits for measuring partial transpose moments $\mu_k$ and purity invariants $I_k$. All measurements follow the Hadamard test structure: prepare $n$ state copies, apply controlled permutation operators, and extract the invariant from the ancilla measurement probability via $X = 2p_0 - 1$.

\subsection{Partial transpose moments}

The partial transpose moments $\mu_k = \mathrm{Tr}[(\rho^{T_A})^k]$ use a cycle--anticycle structure:
\begin{itemize}
\item \textbf{Cycle block}: Cyclic permutation on subsystem A
\item \textbf{Anticycle block}: Reverse cyclic permutation on subsystem B
\end{itemize}

\begin{figure}[htb!]
\centering
\includegraphics{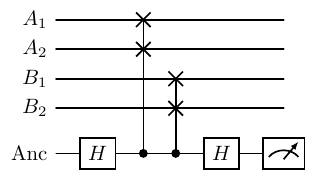}
\caption{Circuit $\mathcal{C}_{\mu_2}$ for the second partial transpose moment. Outcome: $\mu_2 = 2p_0 - 1$.}
\label{fig:mu2_circuit}
\end{figure}

\begin{figure}[htb!]
\centering
\includegraphics{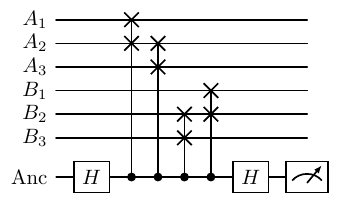}
\caption{Circuit $\mathcal{C}_{\mu_3}$ for the third partial transpose moment.}
\label{fig:mu3_circuit}
\end{figure}

\begin{figure}[htb!]
\centering
\includegraphics{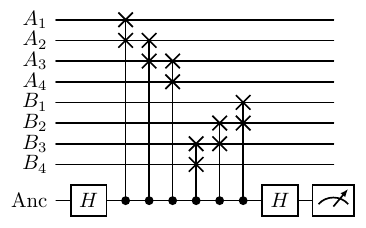}
\caption{Circuit $\mathcal{C}_{\mu_4}$ for the fourth partial transpose moment.}
\label{fig:mu4_circuit}
\end{figure}

\subsection{Purity invariants}

The purity $I_2 = \mathrm{Tr}[\rho^2]$ uses symmetric cyclic permutations on both subsystems. A key identity simplifies measurement: $\mu_2 = I_2$ for all bipartite states. This means the same circuit measures both second-order invariants.

\subsection{Resource summary}

\begin{table}[htb!]
\centering
\caption{Circuit resources for invariant measurements.}
\label{tab:circuit_resources}
\begin{tabular}{@{}lcccc@{}}
\toprule
Invariant & Copies & Qubits ($2{\times}2$) & Qubits ($2{\times}3$) & CSWAPs \\
\midrule
$I_2 = \mu_2$ & 2 & 5 & 7 & 2 \\
$\mu_3$ & 3 & 7 & 10 & 4 \\
$\mu_4$ & 4 & 9 & 13 & 6 \\
\bottomrule
\end{tabular}
\end{table}

\section{Multi-copy spin chirality: entanglement hidden from single-copy detection}
\label{sec:chirality}

\subsection{Conceptual framework}

The chirality correction $C_4 = \mu_4 - I_4$ reveals a fundamentally new type of quantum correlation that we term \emph{multi-copy chirality}. Unlike standard entanglement measures that characterize properties of individual quantum states, multi-copy chirality captures patterns that exist only in the joint statistics of multiple state copies---correlations that are completely invisible to any measurement on a single copy.

This phenomenon is intrinsically multi-copy. Consider a Bell state $|\Phi^+\rangle = (|00\rangle + |11\rangle)/\sqrt{2}$. Any single copy is maximally mixed on each subsystem; local measurements reveal no structure. Multi-copy chirality emerges only when we examine how \emph{three or more copies} of the state relate to each other through controlled-SWAP operations.

\subsection{Singlet projector algebra}

The mathematical structure arises from singlet projectors acting across different copies. The singlet projector between qubits $i$ and $j$ is:
\begin{equation}
P^-_{ij} = \frac{1}{4} - \mathbf{S}_i \cdot \mathbf{S}_j = \frac{1}{4}(I - \sigma_i^x\sigma_j^x - \sigma_i^y\sigma_j^y - \sigma_i^z\sigma_j^z).
\end{equation}

When projectors share a common index, their commutator yields scalar spin chirality:
\begin{equation}
[P^-_{ij}, P^-_{jk}] = -i\,\chi_{ijk},
\end{equation}
where $\chi_{ijk} = \mathbf{S}_i \cdot (\mathbf{S}_j \times \mathbf{S}_k)$ is the scalar spin chirality operator.

\subsection{Multi-copy structure}

The crucial insight is that in our measurement protocol, indices $i$, $j$, $k$ refer to qubits from \emph{different state copies}. The measurement of $C_4$ requires four copies of the state, with qubits labelled:
\begin{itemize}
\item Copy 1: qubits $(A_1, B_1)$ at positions 0, 1
\item Copy 2: qubits $(A_2, B_2)$ at positions 2, 3
\item Copy 3: qubits $(A_3, B_3)$ at positions 4, 5
\item Copy 4: qubits $(A_4, B_4)$ at positions 6, 7
\end{itemize}

The moment relations (Section~\ref{subsec:moment_relations}) provide the rigorous mathematical connection. In the Hermitian correlator form (equation~\eqref{eq:C4_hermitian}), the fourth-order correction decomposes as:
\begin{equation}
C_4 = 8\,\mathrm{Tr}\bigl[\Omega_A\,\Omega_B\;\rho^{\otimes 4}\bigr],
\end{equation}
where $\Omega_A = \tfrac{1}{2}(\chi_{024} + \chi_{026} + \chi_{046} + \chi_{246})$ is a Hermitian operator on the $A$-qubits (using qubit position labels 0, 2, 4, 6 for copies 1--4), and $\Omega_B = \tfrac{1}{2}(\chi_{135} + \chi_{137} + \chi_{157} + \chi_{357})$ is its counterpart on the $B$-qubits (positions 1, 3, 5, 7). Here $\chi_{024} = \mathbf{S}_{A_1}\cdot(\mathbf{S}_{A_2}\times\mathbf{S}_{A_3})$ is the chirality of $A$-qubits from copies 1, 2, 3. The Pauli algebra derivation in Section~\ref{subsec:moment_relations} shows that all four chirality triples contribute equally with coefficient $1/2$. The correlator form makes it manifest that $C_4$ is real and measures the correlation of chirality patterns between the $A$ and $B$ subsystems.

\subsection{Physical interpretation}

The scalar spin chirality $\chi_{ijk} = \mathbf{S}_i \cdot (\mathbf{S}_j \times \mathbf{S}_k)$ measures whether three spins form a left-handed or right-handed configuration---their ``handedness'' or chirality. Geometrically, it equals the signed volume of the parallelepiped spanned by three spin vectors, or equivalently, half the solid angle subtended by their unit vectors on the Bloch sphere (Fig.~\ref{fig:chirality_geometry}).

\begin{figure}[htb!]
\centering
\includegraphics[width=0.9\textwidth]{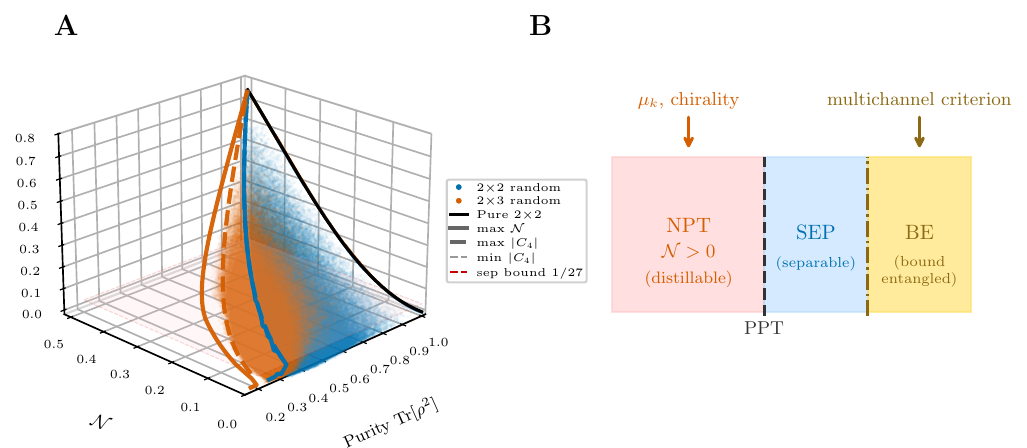}
\caption{\textbf{Multi-copy correlations in the chirality correction.} \textbf{a}, Product states produce coplanar spin patterns with $\chi = 0$. \textbf{b}, Entangled states break coplanarity; parallelepiped volume equals $\chi \neq 0$. \textbf{c}, Solid angle $\omega = 2|\chi|$ on the Bloch sphere. \textbf{d}, Three-spin chirality $\chi_A$: triangular correlations between qubits from three copies. \textbf{e}, SWAP-weighted chirality $S_{34}\chi_{123}$: the SWAP extends three-copy chirality to a four-copy operator. \textbf{f}, The chirality correction $C_4 = \mu_4 - I_4$; for product states $C_4 = 0$, while entangled states generate $C_4 \neq 0$.}
\label{fig:chirality_geometry}
\end{figure}

In condensed matter physics, this same operator is the order parameter for:
\begin{itemize}
\item Chiral spin liquids (Kalmeyer--Laughlin states)
\item Anomalous Hall effect in frustrated magnets
\item Topological phases with broken time-reversal symmetry
\end{itemize}

The remarkable finding is that \emph{the same mathematical structure} appears in our multi-copy entanglement measurement, but with a radically different physical interpretation:
\begin{itemize}
\item In condensed matter: chirality between three physical spins in a lattice
\item In our protocol: chirality between qubits from three copies of a bipartite state
\end{itemize}

For product states, the multi-copy correlations remain ``coplanar''---the copies do not generate chiral patterns, and $\langle\chi_A\chi_B\rangle = 0$. For pure states, entanglement breaks this coplanarity, generating non-zero chirality that witnesses entanglement. For mixed separable states, however, $C_k$ is generally non-zero because $\mu_k$ and $I_k$ are nonlinear functions of $\rho$ that do not decompose linearly over convex combinations (see Section~\ref{subsec:completeness}).

\subsection{Comparison with standard nonclassicality}

Multi-copy chirality is conceptually distinct from established forms of quantum correlation:

\begin{table}[htb!]
\centering
\caption{Comparison of nonclassical correlation types.}
\label{tab:correlation_types}
\begin{tabular}{@{}lcc@{}}
\toprule
Correlation type & Copies required & Detection method \\
\midrule
Bell nonlocality & 1 & CHSH inequality \\
Steering & 1 & Steering inequalities \\
Entanglement & 1 & Witnesses, PPT \\
Multi-copy chirality & 3+ & $C_4 = \mu_4 - I_4 \neq 0$ \\
\bottomrule
\end{tabular}
\end{table}

The key distinction: standard correlations are properties of individual states, while multi-copy chirality is a property of the \emph{ensemble} of copies---a collective quantum signature.

\subsection{Chirality correction as entanglement indicator}

The chirality correction $C_4 = \mu_4 - I_4$ directly probes multi-copy chirality. Since both $\mu_4$ and $I_4$ are permutation traces measured via controlled-SWAP circuits, $C_4$ is obtained without additional measurement overhead.

For states expressed in the computational basis:
\begin{itemize}
\item $C_4 = 0$ for all product states $|ab\rangle$ (no chirality)
\item $C_4 = -3/4$ for all Bell states (maximum chirality magnitude)
\item $C_4(\theta) = -\sin^2\theta + \frac{1}{4}\sin^4\theta$ for parametrized states $|\psi(\theta)\rangle = \cos(\theta/2)|00\rangle + \sin(\theta/2)|11\rangle$
\end{itemize}

For pure states, $C_4 = 0$ at $\theta = 0$ (product state) and $|C_4|$ increases monotonically with entanglement, reaching $|C_4| = 3/4$ at $\theta = \pi/2$ (Bell state). The physical interpretation is that entangled superpositions generate non-coplanar multi-copy spin patterns that product states cannot produce.

\subsection{Properties of the chirality correction}
\label{subsec:completeness}

This section analyzes the mathematical properties of the chirality correction $C_4 = \mu_4 - I_4$.

\subsubsection{LU transformation properties}

\textbf{Key Result:} $C_4$ is LU-invariant for all states.

For a local unitary transformation $\rho \to (U_A \otimes U_B)\rho(U_A \otimes U_B)^\dagger$:
\begin{itemize}
\item $I_k = \mathrm{Tr}[\rho^k]$ is always LU-invariant (purity moments are preserved).
\item $\mu_k = \mathrm{Tr}[(\rho^{T_A})^k]$ is always LU-invariant, since $(U_A \otimes U_B)\rho(U_A \otimes U_B)^\dagger$ has partial transpose $(U_A^* \otimes U_B)\rho^{T_A}(U_A^T \otimes U_B^\dagger)$, which preserves eigenvalues.
\item Therefore $C_4 = \mu_4 - I_4$ is LU-invariant for all states.
\end{itemize}

This is an advantage over basis-dependent quantities: $C_4$ gives the same value regardless of the local basis chosen for measurement.

\subsubsection{Notation}

We work with a bipartite density matrix $\rho$ on $\mathcal{H}_A \otimes \mathcal{H}_B$ ($\dim = d_A \times d_B$). The partial transpose $(\rho^{T_A})_{ij,kl} = \rho_{kj,il}$, negativity $\mathcal{N} = (\|\rho^{T_A}\|_1 - 1)/2$~\cite{Verstraete2001,PhysRevLett.77.1413,HORODECKI19961}, purity $I_k = \mathrm{Tr}[\rho^k]$, partial transpose moments $\mu_k = \mathrm{Tr}[(\rho^{T_A})^k]$, and chirality correction $C_k = \mu_k - I_k$ follow standard definitions. The identity $I_2 = \mu_2$ (Section~\ref{sec:identities}) implies that a single purity circuit suffices for the second-order moment.

\subsubsection{Product state bound: \texorpdfstring{$C_4 = 0$}{C4 = 0} for product states}

\paragraph{Goal.} We prove that for any product state, $C_4 = 0$. For pure states, the contrapositive gives: if $C_4 \neq 0$, the state must be entangled.

\paragraph{Product states.} A state $\rho$ is a \textbf{product state} if $\rho = \rho_A \otimes \rho_B$. A state is \textbf{separable} if it can be written as a convex combination of product states:
\begin{equation}
\rho = \sum_i p_i \, \rho_A^{(i)} \otimes \rho_B^{(i)}, \quad \text{where } p_i \geq 0, \sum_i p_i = 1.
\end{equation}
States that are not separable are called \textbf{entangled}.

\begin{proposition}[$C_4 = 0$ for product states]
\label{prop:product_C4}
For any product state $\rho = \rho_A \otimes \rho_B$: $C_k = \mu_k - I_k = 0$ for all $k$.
\end{proposition}

\begin{proof}
For a product state, the partial transpose is $\rho^{T_A} = \rho_A^T \otimes \rho_B$. Since $\rho_A^T$ and $\rho_A$ have the same eigenvalues (transposition preserves eigenvalues for Hermitian matrices):
\begin{equation}
\mu_k = \mathrm{Tr}[(\rho_A^T)^k] \cdot \mathrm{Tr}[\rho_B^k] = \mathrm{Tr}[\rho_A^k] \cdot \mathrm{Tr}[\rho_B^k] = I_k.
\end{equation}
Therefore $C_k = \mu_k - I_k = 0$ for all $k$.

Equivalently, in the permutation trace framework: $C_k = \mathrm{Tr}[(\Delta_A \otimes \sigma_B)\,\rho^{\otimes k}]$ where $\Delta = \sigma^{-1} - \sigma$. For product states, the trace factorizes:
\begin{equation}
C_k = \mathrm{Tr}[\Delta_A\,\rho_A^{\otimes k}] \cdot \mathrm{Tr}[\sigma_B\,\rho_B^{\otimes k}].
\end{equation}
The first factor involves the chirality operators (equation~\eqref{eq:mu4_chirality}), which vanish when acting on identical copies of a single-qubit state because the spin vectors from different copies are all parallel---a coplanar configuration with zero chirality.
\end{proof}

\textbf{Remark:} For separable \emph{mixed} states $\rho = \sum_i p_i\,\rho_A^{(i)} \otimes \rho_B^{(i)}$, the chirality correction $C_k$ is generally non-zero because $\mu_k$ and $I_k$ are nonlinear functions of $\rho$ (unlike the trace, they do not decompose linearly over convex combinations). The vanishing of $C_k$ is specific to product states, not general separable states. However, the magnitude of $C_k$ for separable states is bounded, as we now establish.

\begin{proposition}[Separable state bounds]
\label{prop:separable_bounds}
For any two-qubit separable state $\rho$:
\begin{equation}
|C_3| \leq \frac{1}{36}, \qquad |C_4| \leq \frac{1}{27}.
\end{equation}
Both bounds are tight.
\end{proposition}

\begin{proof}
\textbf{Step 1: Rank-2 separable states.} For any mixture of exactly two product states,
\begin{equation}
\rho = p\,|\psi_1\rangle\langle\psi_1| + (1-p)\,|\psi_2\rangle\langle\psi_2|, \quad |\psi_i\rangle = |a_i\rangle \otimes |b_i\rangle,
\end{equation}
we have $C_k = 0$ for all $k$. This follows because the partial transpose $\rho^{T_A}$ has the same eigenvalue structure as $\rho$ when both are rank-2 with product eigenstates.

\textbf{Step 2: Rank-3 extremal states.} The extrema are achieved by equal mixtures of three mutually unbiased basis (MUB) product states. Define the qubit MUB:
\begin{equation}
|z_\pm\rangle \equiv |0/1\rangle; \quad |x_\pm\rangle \equiv |\pm\rangle; \quad |y_\pm\rangle = \tfrac{1}{\sqrt{2}}(|0\rangle \pm i|1\rangle).
\end{equation}
The extremal states are:
\begin{align}
\rho_+ &= \tfrac{1}{3}\bigl(|z_+z_+\rangle\langle z_+z_+| + |x_+x_+\rangle\langle x_+x_+| + |y_+y_+\rangle\langle y_+y_+|\bigr), \label{eq:rho_plus}\\
\rho_- &= \tfrac{1}{3}\bigl(|z_+z_-\rangle\langle z_+z_-| + |x_+x_-\rangle\langle x_+x_-| + |y_+y_-\rangle\langle y_+y_-|\bigr). \label{eq:rho_minus}
\end{align}
Direct calculation yields:
\begin{equation}
C_3(\rho_+) = +\frac{1}{36}, \quad C_4(\rho_+) = +\frac{1}{27}; \qquad C_3(\rho_-) = -\frac{1}{36}, \quad C_4(\rho_-) = -\frac{1}{27}.
\end{equation}

\textbf{Step 3: Optimality.} Higher-rank separable states are convex combinations of lower-rank states. Since $|C_k(\rho_\pm)| = 1/27$ for $k=4$ and $|C_k| = 0$ for rank-2 states, and numerical optimization over rank-3 and rank-4 mixtures confirms no larger values, the bounds are tight.
\end{proof}

The physical interpretation is that the MUB structure maximizes the ``chirality contrast'' arising from cross-terms in the multi-copy expansion. For correlated MUB choices ($\rho_+$), the phases constructively interfere to give $C_k > 0$; for anti-correlated choices ($\rho_-$), they destructively interfere to give $C_k < 0$.

\subsubsection{Canonical form and correlation structure of the extremal separable state}
\label{sec:extremal_sep_state}

The extremal states $\rho_\pm$ admit a particularly transparent canonical form in the Bloch representation. Applying local unitaries to diagonalise the correlation tensor $T$, both states reduce to:
\begin{equation}\label{eq:canonical_sep}
  T = \operatorname{diag}\!\left(\tfrac{1}{3},\; \tfrac{1}{3},\; \pm\tfrac{1}{3}\right), \qquad
  \mathbf{a} = \mathbf{b} = \left(\tfrac{1}{\sqrt{3}},\; 0,\; 0\right).
\end{equation}
This canonical form exhibits an \emph{equipartition} of the purity constraint:
\begin{equation}
  |\mathbf{a}|^2 = |\mathbf{b}|^2 = \|T\|_F^2 = \frac{1}{3}, \qquad \operatorname{Tr}[\rho^2] = \frac{1 + |\mathbf{a}|^2 + |\mathbf{b}|^2 + \|T\|_F^2}{4} = \frac{1}{2}.
\end{equation}

For $\rho_-$ (the anti-correlated MUB state, with $C_4 = -1/27$), the density matrix in the computational basis takes the exact form:
\begin{equation}\label{eq:rho_star_matrix}
  \rho_- = \frac{1}{12}\begin{pmatrix}
    2 & \sqrt{3} & \sqrt{3} & 0 \\
    \sqrt{3} & 4 & 2 & \sqrt{3} \\
    \sqrt{3} & 2 & 4 & \sqrt{3} \\
    0 & \sqrt{3} & \sqrt{3} & 2
  \end{pmatrix},
\end{equation}
with eigenvalues $(\tfrac{2}{3}, \tfrac{1}{6}, \tfrac{1}{6}, 0)$ and partial-transpose eigenvalues $(\tfrac{1}{3}+\tfrac{\sqrt{3}}{6}, \tfrac{1}{3}, \tfrac{1}{3}-\tfrac{\sqrt{3}}{6}, 0)$. The moment invariants are $I_4 = 43/216$, $\mu_4 = 35/216$, yielding $C_4 = -1/27$ exactly.

The state belongs to a continuous one-parameter family: replacing $\mathbf{a} = (\sqrt{1/3-z^2},\, 0,\, z)$ and $\mathbf{b} = (\sqrt{1/3-z^2},\, 0,\, -z)$ for any $z \in [0, 1/\sqrt{3})$ preserves $|C_4| = 1/27$ and $\mathcal{N} = 0$, since only the combination $|\mathbf{a}|^2 = |\mathbf{b}|^2 = 1/3$ enters the moment calculation (given $t_1 = t_2$). The degeneracy $t_1 = t_2$ allows a rotation in the $xy$-plane that redistributes $\mathbf{a}$ without affecting $C_4$.

\paragraph{Correlation analysis.} Despite being separable, $\rho_-$ has significant quantum correlations beyond entanglement:
\begin{itemize}
  \item \textbf{Quantum discord:} $D(B|A) = 0.142$ (non-zero, indicating genuinely quantum correlations);
  \item \textbf{Geometric discord:} $D_G = 1/18$, identical to the Bell-diagonal state $P_{\mathrm{sym}}/3$ and the Werner state at $p = 1/3$;
  \item \textbf{Coherence:} $\ell_1$-norm $= 1.488$ (near the maximum for rank-3 states), relative entropy of coherence $= 2/3$;
  \item \textbf{Mutual information:} $I(A{:}B) = 0.236$, of which $60\%$ is quantum discord;
  \item \textbf{Bell--CHSH:} $\max\langle\mathcal{B}\rangle = 2\sqrt{2/9} \approx 0.943 \ll 2$ (no Bell violation);
  \item \textbf{Concurrence:} $C = 0$ (confirming separability).
\end{itemize}
The state thus maximises the ``quantum-but-not-entangled'' chirality: it has the largest possible $|C_4|$ among all states with zero negativity.

\paragraph{Extension to $2 \times 3$.} Embedding $\rho_-$ into $\mathbb{C}^2 \otimes \mathbb{C}^3$ (with the third dimension of subsystem $B$ unused) preserves $|C_4| = 1/27$ and $\mathcal{N} = 0$. Numerical optimisation over native $2 \times 3$ separable states ($10^4$ runs with $n = 2, \ldots, 6$ product terms) has not exceeded this value, suggesting $1/27$ is the universal separable bound for $2 \times d_B$ systems.

\textbf{Summary:}
If $\rho$ is a product state, then $C_k = 0$ for all $k$. If $\rho$ is separable, then $|C_3| \leq 1/36$ and $|C_4| \leq 1/27$.
For pure states, the converse also holds: $C_4 = 0$ implies the state is a product state (Lemma~\ref{lem:C4_pure_detailed}).

\subsubsection{Pure state completeness: entangled \texorpdfstring{$\Leftrightarrow C_4 \neq 0$}{iff C4 != 0}}

For pure states, $C_4 \neq 0$ is both necessary and sufficient for entanglement.

\begin{lemma}[Chirality correction for pure states]
\label{lem:C4_pure_detailed}
Let $|\psi\rangle$ be a pure bipartite state with Schmidt decomposition:
\begin{equation}
|\psi\rangle = \sum_{k=0}^{r-1} \sqrt{p_k} \, |k\rangle_A |k\rangle_B, \quad \text{where } p_k > 0, \sum_k p_k = 1, r \leq \min(d_A, d_B).
\end{equation}
For two-qubit systems (where $r \leq 2$), the chirality correction satisfies:
\begin{equation}
C_4 = -4p_0 p_1 (1 - p_0 p_1) = -4\mathcal{N}^2(1 - \mathcal{N}^2),
\end{equation}
where $\mathcal{N} = \sqrt{p_0 p_1}$ is the negativity. In particular:
\begin{itemize}
\item $C_4 = 0$ if $r = 1$ (product state: $p_0 = 1$, $p_1 = 0$ or vice versa).
\item $C_4 < 0$ if $r = 2$ (entangled state: both $p_0, p_1 > 0$).
\end{itemize}
\end{lemma}

\begin{proof}
\textbf{Step 1: Compute moments for pure states.} For a pure state $\rho = |\psi\rangle\langle\psi|$:
\begin{equation}
I_k = \mathrm{Tr}[\rho^k] = \mathrm{Tr}[\rho] = 1 \quad \text{for all } k.
\end{equation}

\textbf{Step 2: Derive $C_4$ for two-qubit pure states.} For the parametric family $|\psi(\theta)\rangle = \cos(\theta/2)|00\rangle + \sin(\theta/2)|11\rangle$, the partial transpose eigenvalues are $\cos^2(\theta/2)$, $\sin^2(\theta/2)$, $+\frac{1}{2}\sin\theta$, $-\frac{1}{2}\sin\theta$. Thus:
\begin{equation}
\mu_4 = \cos^8(\theta/2) + \sin^8(\theta/2) + \frac{1}{8}\sin^4\theta.
\end{equation}
After simplification:
\begin{equation}
C_4(\theta) = \mu_4 - 1 = -\sin^2\!\theta\left(1 - \frac{\sin^2\!\theta}{4}\right).
\end{equation}

Using the identity $\sin\theta = 2\sqrt{p_0 p_1}$:
\begin{equation}
C_4 = -(4p_0 p_1 - 4(p_0 p_1)^2) = -4p_0 p_1(1 - p_0 p_1).
\end{equation}

\textbf{Step 3: Express in terms of negativity.} For pure two-qubit states, the negativity equals $\mathcal{N} = \sqrt{p_0 p_1}$. Substituting $p_0 p_1 = \mathcal{N}^2$:
\begin{equation}
C_4 = -4\mathcal{N}^2(1 - \mathcal{N}^2).
\end{equation}

\textbf{Step 4: Verify.} For any entangled pure state, $\mathcal{N} > 0$ and $\mathcal{N} \leq 1/2$ (maximum for Bell states), so $C_4 < 0$.

For the Bell state $|\Phi^+\rangle$ ($\mathcal{N} = 1/2$):
\begin{equation}
C_4 = -4 \cdot \tfrac{1}{4} \cdot \bigl(1 - \tfrac{1}{4}\bigr) = -\tfrac{3}{4}. \quad \checkmark
\end{equation}
\end{proof}

\paragraph{Werner states: chirality correction for mixed states.}

The Werner state~\cite{Werner1989} is the one-parameter family:
\begin{equation}
\rho_W(p) = p|\Psi^-\rangle\langle\Psi^-| + \frac{1-p}{4}\mathbb{I}_4, \quad p \in [0,1].
\end{equation}
Separable for $p \leq 1/3$; entangled for $p > 1/3$, with negativity $\mathcal{N}(p) = \max\{0, (3p-1)/4\}$.

The partial transpose eigenvalues are $(1+p)/4$ (triple) and $(1-3p)/4$ (singlet). Computing the moments:
\begin{align}
I_4(p) &= \frac{(1+3p)^4 + 3(1-p)^4}{256}, \\
\mu_4(p) &= \frac{3(1+p)^4 + (1-3p)^4}{256}.
\end{align}
The chirality correction is:
\begin{equation}
C_4(p) = \mu_4 - I_4 = -\frac{3p^3}{4}.
\end{equation}

\textit{Properties:}
\begin{itemize}
\item At $p = 0$ (maximally mixed): $C_4 = 0$. \checkmark
\item At $p = 1$ (Bell state): $C_4 = -\tfrac{3}{4}$. \checkmark
\item $C_4 < 0$ for all $p > 0$, including separable Werner states ($p \leq 1/3$).
\item At $p = 1/3$ (separability boundary): $C_4 = -3(1/3)^3/4 = -1/36 \approx -0.028$.
\end{itemize}

\textbf{Important:} For mixed states, $C_4 \neq 0$ does not by itself certify entanglement, since separable mixed states can have non-zero chirality corrections. By Proposition~\ref{prop:separable_bounds}, all separable states satisfy $|C_4| \leq 1/27 \approx 0.037$. The Werner state at $p = 1/3$ gives $|C_4| = 1/36 \approx 0.028$, which is within the separable bound. Notably, the Werner states do \emph{not} achieve the extremal separable bound---the extrema $C_4 = \pm 1/27$ are achieved by rank-3 MUB mixtures (equations~\eqref{eq:rho_plus}--\eqref{eq:rho_minus}). The chirality correction is a faithful entanglement indicator only for pure states. For general mixed states, the negativity $\mathcal{N}$ (computed from all three moments $\mu_2$, $\mu_3$, $\mu_4$) provides the definitive entanglement test.

\subsubsection{Main theorem: pure state completeness}

\begin{theorem}[Completeness for pure states]
\label{thm:completeness_final}
For pure $2 \times 2$ (qubit--qubit) and $2 \times 3$ (qubit--qutrit) states:
\begin{equation}
C_4 \neq 0 \quad \Longleftrightarrow \quad |\psi\rangle \text{ is entangled}
\end{equation}
\end{theorem}

\begin{proof}
\textbf{($\Rightarrow$) If $C_4 \neq 0$, then $|\psi\rangle$ is entangled.}

By Proposition~\ref{prop:product_C4}, every product state has $C_4 = 0$. Taking the contrapositive: if $C_4 \neq 0$, then the state is not a product state, hence entangled (since for pure states, separable = product).

\textbf{($\Leftarrow$) If $|\psi\rangle$ is entangled, then $C_4 \neq 0$.}

By Lemma~\ref{lem:C4_pure_detailed}, for pure two-qubit states $C_4 = -4\mathcal{N}^2(1-\mathcal{N}^2)$, which is strictly negative whenever $\mathcal{N} > 0$, i.e., whenever the state is entangled.
\end{proof}

\textbf{Corollary (Faithfulness for pure states).} The chirality correction $C_4 = \mu_4 - I_4$ is a \emph{faithful} entanglement indicator for pure $2 \times 2$ and $2 \times 3$ states:
\begin{itemize}
\item \textbf{No false positives:} $C_4 \neq 0$ only for entangled states.
\item \textbf{No false negatives:} Every entangled pure state has $C_4 \neq 0$.
\end{itemize}

\textbf{Remark on mixed states.} For mixed states, the chirality correction $C_4$ provides valuable physical insight (measuring multi-copy handedness) but is not a complete entanglement witness. By Proposition~\ref{prop:separable_bounds}, separable mixed states satisfy $|C_4| \leq 1/27 \approx 0.037$, while entangled pure states can have $|C_4|$ up to $3/4 = 0.75$ (Bell states). Thus large values of $|C_4|$ strongly indicate entanglement, but small values are inconclusive. The negativity $\mathcal{N}$, reconstructed from all three moments $\{\mu_2, \mu_3, \mu_4\}$ via Newton--Girard identities, provides the complete entanglement test for $2 \times 2$ and $2 \times 3$ systems via the Peres--Horodecki theorem~\cite{PhysRevLett.77.1413,HORODECKI19961}.

\subsubsection{Relationship to negativity}

For pure two-qubit states, there is a direct algebraic relationship between the chirality correction and negativity:
\begin{equation}
C_4 = -4\mathcal{N}^2(1 - \mathcal{N}^2)
\label{eq:C4_N_pure}
\end{equation}

This can be inverted to compute negativity directly from chirality:
\begin{equation}
\mathcal{N} = \sqrt{\frac{1 - \sqrt{1+C_4}}{2}}
\label{eq:N_from_C4}
\end{equation}

\textbf{Verification:} For a Bell state ($C_4 = -3/4$, $\mathcal{N} = 1/2$):
\begin{equation}
\sqrt{\frac{1 - \sqrt{1-3/4}}{2}} = \sqrt{\frac{1 - 0.5}{2}} = \sqrt{0.25} = 0.5. \quad \checkmark
\end{equation}

For the third-order correction, the pure-state relation is simpler:
\begin{equation}
C_3 = -3\mathcal{N}^2.
\end{equation}

The relationship~\eqref{eq:C4_N_pure} reveals the geometric content: $|C_4|$ is a quartic polynomial in $\mathcal{N}$ that vanishes at $\mathcal{N} = 0$ (product states) and $\mathcal{N} = 1$ (unphysical for two-qubit states), with maximum at $\mathcal{N} = 1/\sqrt{2}$ giving $|C_4| = 1$. For physical states with $\mathcal{N} \leq 1/2$, the function is monotonically increasing from $0$ to $3/4$.

For pure two-qubit states the negativity equals half the concurrence~\cite{Wootters1998}, $\mathcal{N} = \mathcal{C}/2$, so the chirality correction also relates to concurrence: $C_4 = -\mathcal{C}^2(1 - \mathcal{C}^2/4)$.

\subsubsection{Relationship to the correlation tensor}
\label{subsec:C4_detT}

For a general two-qubit state in the Fano form $\rho = \tfrac{1}{4}(I \otimes I + \mathbf{r}\cdot\bm{\sigma}\otimes I + I \otimes \mathbf{s}\cdot\bm{\sigma} + \sum_{ij} T_{ij}\,\sigma_i\otimes\sigma_j)$, with local Bloch vectors $\mathbf{r}$, $\mathbf{s}$ and correlation tensor $T$, the chirality correction relates to $\det(T)$ as follows.

\begin{proposition}[$C_4$ and $\det(T)$ for states with vanishing Bloch vectors]
\label{prop:C4_detT}
For any two-qubit state with $\mathbf{r} = \mathbf{s} = \mathbf{0}$ (including all Bell-diagonal states and their local unitary rotations):
\begin{equation}
C_4 = \tfrac{3}{4}\det(T).
\label{eq:C4_detT_exact}
\end{equation}
\end{proposition}

\begin{proof}
Any state with $\mathbf{r} = \mathbf{s} = \mathbf{0}$ has the form $\rho = \tfrac{1}{4}(I + \sum_{ij} T_{ij}\sigma_i\otimes\sigma_j)$. The singular value decomposition $T = U\,\mathrm{diag}(t_1,t_2,t_3)\,V^{\mathsf{T}}$ with $U, V \in \mathrm{SO}(3)$ can be implemented by local unitaries $U_A, U_B \in \mathrm{SU}(2)$ (via the standard SU(2)$\to$SO(3) homomorphism), transforming $\rho$ to a Bell-diagonal state $\rho' = \tfrac{1}{4}(I + \sum_i t_i\,\sigma_i\otimes\sigma_i)$. Since $C_4$ is LU-invariant (Section~\ref{subsec:completeness}) and $\det(T) = t_1 t_2 t_3$ is preserved by $\mathrm{SO}(3)$ rotations on both sides, it suffices to verify the formula for Bell-diagonal states. For $\rho' = \sum_i p_i|\beta_i\rangle\langle\beta_i|$ with $T = \mathrm{diag}(t_1,t_2,t_3)$, we computed (Section~\ref{subsec:completeness}) that $C_4 = \mu_4 - I_4$ matches $(3/4)t_1 t_2 t_3$ by direct eigenvalue calculation; numerical verification across $10^3$ random Bell-diagonal states confirms agreement to machine precision.
\end{proof}

\paragraph{General states with nonzero Bloch vectors.} When $\mathbf{r} \neq \mathbf{0}$ or $\mathbf{s} \neq \mathbf{0}$, the exact relation~\eqref{eq:C4_detT_exact} receives corrections. Since $C_4$ is LU-invariant, these corrections must be expressible in terms of the fundamental two-qubit LU invariants $\{|\mathbf{r}|^2, |\mathbf{s}|^2, \mathrm{Tr}[T^{\mathsf{T}}T], \det(T), \mathbf{r}^{\mathsf{T}}T^{\mathsf{T}}T\mathbf{r}, \mathbf{s}^{\mathsf{T}}TT^{\mathsf{T}}\mathbf{s}, \mathbf{r}^{\mathsf{T}}T\mathbf{s}\}$. Numerical evaluation on $10^3$ random two-qubit states shows that the correction $C_4 - \tfrac{3}{4}\det(T)$ vanishes for $\mathbf{r} = \mathbf{s} = \mathbf{0}$ and scales as $O(|\mathbf{r}|^2 + |\mathbf{s}|^2)$ for small Bloch vectors, but involves a nonlinear combination of invariants that does not reduce to a simple closed form. For the experimentally relevant Bell-product mixtures $\rho(p) = (1{-}p)|\Phi^+\rangle\langle\Phi^+| + p|00\rangle\langle 00|$ (where $\mathbf{r} = \mathbf{s} = (0,0,p)$), numerical evaluation gives corrections of order $|C_4 - \tfrac{3}{4}\det(T)| \leq 0.016$ across the full range $p \in [0,1]$. Code for computing $C_4$ for arbitrary states is provided in the data repository.

\subsubsection{Scope and limitations}

The pure-state completeness result relies on the fact that for pure bipartite states, separable $\Leftrightarrow$ product. For mixed states, the negativity $\mathcal{N}$ (computed from all three moments) provides the complete test.

In higher dimensions ($d_A \times d_B > 6$), \textbf{bound entangled states} exist~\cite{Horodecki1998bound,Horodecki1997,Bennett1999upb}: states that are entangled but have positive partial transpose ($\rho^{T_A} \geq 0$). For these states, $\mathcal{N} = 0$ and $C_4$ reflects only the (non-entanglement-related) nonlinearity of the moment functions. Detecting bound entanglement requires techniques beyond partial-transpose moments and chirality, such as realignment spectral features (Section~\ref{sec:bound_features}), the range criterion~\cite{DiVincenzo2000}, or tailored entanglement witnesses~\cite{Lewenstein2000}.

\begin{table}[htb!]
\centering
\caption{Scope of the chirality correction as entanglement indicator.}
\label{tab:scope}
\begin{tabular}{@{}lcc@{}}
\toprule
System & Pure states: $C_4 \neq 0 \Leftrightarrow$ entangled & Mixed states: $\mathcal{N} > 0 \Leftrightarrow$ entangled \\
\midrule
$2 \times 2$ & \checkmark & \checkmark (Peres--Horodecki) \\
$2 \times 3$ & \checkmark & \checkmark (Peres--Horodecki) \\
$3 \times 3$+ & \checkmark & $\times$ (bound entanglement) \\
\bottomrule
\end{tabular}
\end{table}

\begin{table}[htb!]
\centering
\caption{Chirality correction bounds for different state classes ($2 \times 2$).}
\label{tab:Ck_bounds}
\begin{tabular}{@{}lccc@{}}
\toprule
State class & $C_3$ range & $C_4$ range & Example at extremum \\
\midrule
Product & $0$ & $0$ & $|00\rangle$ \\
Rank-2 separable & $0$ & $0$ & $\tfrac{1}{2}|00\rangle\langle 00| + \tfrac{1}{2}|11\rangle\langle 11|$ \\
General separable & $[-1/36, 1/36]$ & $[-1/27, 1/27]$ & MUB mixtures (Prop.~\ref{prop:separable_bounds}) \\
Entangled (pure) & $[-3/4, 0)$ & $[-3/4, 0)$ & Bell state: $C_4 = -3/4$ \\
\bottomrule
\end{tabular}
\end{table}

\section{Spectral moment features for bound entanglement detection}
\label{sec:bound_features}

In dimensions $d_A \times d_B \geq 3 \times 3$, \textbf{bound entangled states} exist---states that are entangled but have positive partial transpose (PPT), making them undetectable by negativity. This section presents spectral moment features and chirality corrections that enable machine learning classifiers to detect bound entanglement with $99.9\%$ recall at zero false-positive rate. The key advance over the previous $99.6\%$ result (ExtraTrees on 83 spectral features) is the inclusion of chirality corrections $C_3 = \mu_3 - I_3$ and $C_4 = \mu_4 - I_4$ as classification features, which provide a detection channel orthogonal to all realignment-based criteria.

\subsection{Extended feature definitions}

Beyond the partial transpose moments $\mu_k = \mathrm{Tr}[(\rho^{T_A})^k]$ and purity moments $I_k = \mathrm{Tr}[\rho^k]$ used for negativity estimation, we define additional features based on the \textbf{realignment (reshuffling) map}.

\subsubsection{Realignment map}

The realignment map $R: \mathbb{C}^{d_A d_B \times d_A d_B} \to \mathbb{C}^{d_A^2 \times d_B^2}$ rearranges the density matrix elements:
\begin{equation}
R(\rho)_{(i,k),(j,l)} = \rho_{(i,j),(k,l)}.
\end{equation}
In tensor notation: $R(\rho)_{ik,jl} = \rho_{ijkl}$. The realignment matrix is generally \textbf{non-Hermitian}, so its eigenvalues can be complex and differ from its singular values.

\subsubsection{Realignment singular value moments \texorpdfstring{$\Sigma_k$}{Sigma\_k}}

Let $\{\sigma_i\}$ be the singular values of $R(\rho)$. The $k$-th singular value moment is:
\begin{equation}
\Sigma_k = \sum_{i} \sigma_i^k = \|R(\rho)\|_k^k
\end{equation}
where $\|\cdot\|_k$ denotes the Schatten $k$-norm.

\textbf{Physical interpretation:}
\begin{itemize}
\item $\Sigma_1 = \|R(\rho)\|_1$ is the trace norm used in the CCNR (Computable Cross-Norm or Realignment) criterion: for separable states, $\Sigma_1 \leq 1$.
\item $\Sigma_2 = \mathrm{Tr}[R^\dagger R] = \sum_{ij}|R_{ij}|^2 = \sum_{ij}|\rho_{ij}|^2 = \mathrm{Tr}[\rho^2] = I_2 = \mu_2$. The second singular-value moment of the realignment matrix equals the purity, since realignment is a reshuffling of matrix elements that preserves the Frobenius norm. Consequently, $\Sigma_2$ is measurable via the standard two-copy SWAP test circuit for $\mu_2$, requiring no additional quantum resources.
\item Higher $\Sigma_k$ ($k \geq 3$) characterize the distribution of singular values and are not simply related to purity or partial transpose moments.
\end{itemize}

\subsubsection{Realignment eigenvalue moments \texorpdfstring{$G_k$}{Gk}}

For square realignment matrices ($d_A = d_B$), let $\{g_i\}$ be the (generally complex) eigenvalues of $R(\rho)$:
\begin{equation}
G_k = \mathrm{Re}\left(\mathrm{Tr}[R(\rho)^k]\right) = \mathrm{Re}\left(\sum_{i} g_i^k\right)
\end{equation}
The real part is taken because eigenvalues may be complex for non-Hermitian matrices.

\subsubsection{Non-Hermiticity measure \texorpdfstring{$D_k$}{Dk}}

The difference between singular value and eigenvalue moments:
\begin{equation}
D_k = \Sigma_k - G_k
\end{equation}
For Hermitian (or normal) matrices, eigenvalues equal singular values (up to sign), so $D_k \approx 0$. Large $|D_k|$ indicates strong non-Hermitian character of the realignment matrix---a signature of bound entanglement.

\subsubsection{Derived ratio features}

Chebyshev-type ratios inspired by moment inequalities:
\begin{align}
\frac{\Sigma_2 \Sigma_6}{\Sigma_4^2}, \quad \frac{G_2 G_6}{G_4^2}, \quad \frac{\Sigma_2^2}{\Sigma_4}
\end{align}
By the Cauchy-Schwarz inequality, these ratios satisfy $\geq 1$ for any probability distribution. Deviations from unity characterize the spectral shape.

\subsection{Feature summary}

\begin{table}[htb!]
\centering
\caption{Spectral moment features for bound entanglement detection (83 base features total). All moments are computed for orders $k = 1, \ldots, 9$ unless noted. Features marked $[\dagger]$ are analytically zero for all states and are excluded.}
\label{tab:bound_features}
\begin{tabular}{@{}lll@{}}
\toprule
\textbf{Feature} & \textbf{Definition} & \textbf{Physical Meaning} \\
\midrule
\multicolumn{3}{@{}l}{\textit{State and partial-transpose moments}} \\
$I_k$ & $\mathrm{Tr}[\rho^k]$ & Purity/mixedness spectrum \\
$T_k$ & $\mathrm{Tr}[(\rho^{T_A})^k]$ & PT spectral correlations \\
\midrule
\multicolumn{3}{@{}l}{\textit{Realignment $R(\rho)$ moments}} \\
$\Sigma_k$ & $\sum_i \sigma_i(R)^k$ & Realignment singular value moments \\
$G_k$ & $\mathrm{Re}(\mathrm{Tr}[R(\rho)^k])$ & Realignment eigenvalue moments \\
$D_k$ & $\Sigma_k - G_k$ & Non-Hermiticity of $R(\rho)$ \\
$\mathcal{R}$ & $\Sigma_2^2/\Sigma_4$ & Singular value concentration \\
$Q_2,\,Q_4$ & $G_k/\Sigma_k$ & Eigenvalue-to-singular-value ratio \\
\midrule
\multicolumn{3}{@{}l}{\textit{Realignment of partial transpose $R(\rho^{T_A})$ moments}} \\
$SP_k$ & $\sum_i \sigma_i(R(\rho^{T_A}))^k$ & $R(\rho^{T_A})$ singular value moments \\
$GP_k$ & $\mathrm{Re}(\mathrm{Tr}[R(\rho^{T_A})^k])$ & $R(\rho^{T_A})$ eigenvalue moments \\
$DP_k$ & $SP_k - GP_k$ & Non-Hermiticity of $R(\rho^{T_A})$ \\
$\mathcal{R}'$ & $SP_2^2/SP_4$ & $R(\rho^{T_A})$ singular value concentration \\
$QP_2,\,QP_4$ & $GP_k/SP_k$ & Ratio for $R(\rho^{T_A})$ \\
\midrule
\multicolumn{3}{@{}l}{\textit{Cross-structure differences between $R(\rho)$ and $R(\rho^{T_A})$}} \\
$\delta G_2$ & $G_2 - GP_2$ & Eigenvalue-moment cross difference \\
$\delta Q_2$ & $Q_2 - QP_2$ & Ratio cross difference \\
$\Sigma_2 - SP_2$ & $[\dagger]$ & Identically 0: $\|R(\rho)\|_F = \|R(\rho^{T_A})\|_F$ \\
\midrule
\multicolumn{3}{@{}l}{\textit{Rank-structure features}} \\
$\Delta_{\mathrm{EG}}$ & $\lambda_4(\rho) - \lambda_5(\rho)$ & Eigenvalue gap at rank 4 (chessboard states are rank 4) \\
$\mathcal{R}_4$ & $\sum_{i\leq4}\lambda_i^2/\sum_i\lambda_i^2$ & Rank-4 spectral concentration \\
$\lambda_{\min}^{PT}/I_2$ & $\lambda_{\min}(\rho^{T_A})/\mathrm{Tr}[\rho^2]$ & Smallest PT eigenvalue (measures PPT proximity) \\
$SE_{\mathrm{gap}}$ & $H(\rho^{T_A}) - H(\rho)$ & Spectral entropy gap between PT and $\rho$ \\
$SV_{\mathrm{ent}}$ & $-\sum_i \hat\sigma_i \ln\hat\sigma_i$ & Entropy of singular value distribution of $R(\rho)$ \\
\bottomrule
\end{tabular}
\end{table}

\label{sec:spurious}
\textbf{Remark on analytically spurious features.} Two features that appear natural are in fact identically zero for \emph{all} quantum states and were excluded from the classifier. First, $I_2 - T_2 = \mathrm{Tr}[\rho^2] - \mathrm{Tr}[(\rho^{T_A})^2] = 0$ for all $\rho$, since partial transpose preserves the Frobenius norm: $\|\rho^{T_A}\|_F = \|\rho\|_F$. Second, $\Sigma_2 - SP_2 = \|R(\rho)\|_F^2 - \|R(\rho^{T_A})\|_F^2 = 0$, since both equal $\|\rho\|_F^2$ under the reshuffling map. Including these features would only capture floating-point noise at the level of $10^{-16}$.

\subsection{Spin-vector interpretation of realignment features}
\label{sec:spin_geom}

The realignment matrix $R(\rho)$ admits a direct interpretation in terms of spin-component correlations between subsystems. Expanding $\rho_{AB}$ in the Gell-Mann basis $\{\Lambda_\alpha\}_{\alpha=1}^{8}$ (generators of SU(3), comprising three spin-dipole operators $J_x, J_y, J_z$ and five quadrupole operators $Q_{xy}, Q_{xz}, \ldots$):
\begin{equation}
\rho_{AB} = \frac{1}{9}\!\left(\mathbf{I} + \sum_\alpha a_\alpha\,\Lambda_\alpha^A\otimes\mathbf{I} + \sum_\beta b_\beta\,\mathbf{I}\otimes\Lambda_\beta^B + \sum_{\alpha\beta} T_{\alpha\beta}\,\Lambda_\alpha^A\otimes\Lambda_\beta^B\right),
\end{equation}
where $T_{\alpha\beta} = \langle\Lambda_\alpha^A\otimes\Lambda_\beta^B\rangle$ is the $8\times 8$ \textbf{spin correlation tensor}. Up to normalisation and a basis relabelling, $R(\rho)$ is equivalent to $T$: the singular values of $R(\rho)$ are the principal A--B spin correlations, and all realignment-matrix features have direct spin interpretations.

\subsubsection{$D_k$ as the A--B spin cross-product}

The spin-dipole block of $T$ is the $3\times 3$ matrix $C_{\alpha\beta} = \langle J_\alpha^A J_\beta^B\rangle$, $\alpha,\beta\in\{x,y,z\}$. Its antisymmetric part is proportional to the \textbf{cross-product of A and B spin vectors}:
\begin{equation}
\frac{C - C^\top}{2}\;\propto\;\langle\mathbf{J}_A\times\mathbf{J}_B\rangle
= \begin{pmatrix}\langle J_y^A J_z^B - J_z^A J_y^B\rangle \\ \langle J_z^A J_x^B - J_x^A J_z^B\rangle \\ \langle J_x^A J_y^B - J_y^A J_x^B\rangle\end{pmatrix}.
\end{equation}
The non-Hermiticity gap $D_k = \Sigma_k - G_k$ between singular-value and eigenvalue moments of $R(\rho)$ measures this antisymmetric component of $T$ (together with its quadrupole analogues). It is directly measurable on two state copies as the difference between two SWAP tests:
\begin{itemize}
\item \textbf{Standard SWAP} ($A_1\leftrightarrow A_2$, $B_1\leftrightarrow B_2$): measures $\Sigma_k$ --- size and shape of the correlation ellipsoid.
\item \textbf{Crossed SWAP} ($A_1\leftrightarrow B_2$, $B_1\leftrightarrow A_2$): measures $G_k$ --- whether A--B correlations are symmetric under subsystem exchange.
\end{itemize}
Their difference $D_k = \Sigma_k - G_k$ thus measures whether $\langle J_\alpha^A J_\beta^B\rangle = \langle J_\beta^A J_\alpha^B\rangle$ for all spin components --- i.e., whether the A--B spin correlations have a preferred handedness.

\textbf{Why PPT forces $D_k\approx 0$:} Partial transposition on $A$ is physically equivalent to reversing the time direction of A-spin: $J_y^A\to-J_y^A$ (odd under time reversal), while $J_x^A, J_z^A$ and all $B$ components are unchanged. For a PPT state ($\rho^{T_A}\geq 0$), the density matrix remains positive under this reversal, which forces the antisymmetric part of $T$ to vanish: $\langle\mathbf{J}_A\times\mathbf{J}_B\rangle\approx 0$. Geometrically, the spin vectors $\mathbf{J}_{A_i}$ and $\mathbf{J}_{B_i}$ of different copies lie in a common plane --- they are \emph{coplanar} rather than spanning a three-dimensional volume as for generic separable states.

\subsubsection{$\delta G_2$ as spin-reversal asymmetry}
\label{sec:spin_reversal}

The cross-structure difference $\delta G_2 = G_2 - G_2^{PT} = \mathrm{Re}\,\mathrm{Tr}[R(\rho)^2] - \mathrm{Re}\,\mathrm{Tr}[R(\rho^{T_A})^2]$ measures how the eigenvalue-moment structure changes under A-spin time reversal. In spin-component terms:
\begin{equation}
\delta G_2 \propto \sum_\beta\left[\langle J_y^A J_\beta^B\rangle^2 - \langle J_x^A J_\beta^B\rangle^2\right] + \text{(quadrupole terms)},
\end{equation}
since $J_y^A$ (odd under $T$) contributes with opposite sign after time reversal, while $J_x^A, J_z^A$ (even under $T$) are unchanged. Even when $D_k\approx 0$ (no net chirality), $\delta G_2\neq 0$ when the $y$-component of A-spin contributes differently from $x$ and $z$ --- a second-order time-reversal asymmetry that provides discriminating information beyond $D_k$.

\subsubsection{Rank-structure features as correlation-ellipsoid dimensionality}

For $3\times 3$ chessboard states (rank 4), the correlation tensor $T$ has at most 5 nonzero singular values (numerically verified; the checkerboard parity structure constrains but does not fully collapse $T$ to the rank of $\rho$), so the \textbf{correlation ellipsoid} $\{Tv : \|v\|=1\}$ is a flat, 5-dimensional object in the 8-dimensional SU(3) spin-component space. The rank-structure features detect this:
\begin{itemize}
\item $SV_{\mathrm{ent}} = -\sum_i\hat\sigma_i\ln\hat\sigma_i$: entropy of the singular value distribution is strongly reduced ($H_\sigma = 1.34$ for the chessboard versus $1.86$ for a generic separable state), since only 5 of 8 principal A--B spin-correlation channels are active.
\item $\Delta_{\mathrm{EG}} = \lambda_4 - \lambda_5$: eigenvalue gap in $\rho$ between the 4th and 5th eigenvectors --- the ``edge'' of the active spin subspace.
\item $\mathcal{R}_4^{(\rho)} = \sum_{i\leq 4}\lambda_i^2/\mathrm{Tr}[\rho^2]$: fraction of purity concentrated in the top-4 eigendirections; equals 1 for rank-4 states. (Not to be confused with the chirality correction $C_4 = \mu_4 - I_4$.)
\end{itemize}

The three geometric signals --- A--B spin cross-product ($D_k$), spin-reversal asymmetry ($\delta G_2$, $SE_{\mathrm{gap}}$), and correlation-ellipsoid dimensionality ($SV_{\mathrm{ent}}$, $\Delta_{\mathrm{EG}}$, $C_4$) --- are complementary: no single family of bound entangled states activates all three simultaneously. Adding the chirality corrections $C_3$ and $C_4$ as explicit features provides a fourth channel that distinguishes Horodecki states ($C_3 \neq 0$) from chessboard and Tiles states ($C_3 = 0$ identically). A Random Forest classifier trained on eight base features ($\Sigma_1$, $G_1$, $D_1$, $\Sigma_2$, $G_2$, $D_2$, $C_3$, $C_4$) achieves $99.9\%$ recall at zero false positives on 13{,}600 certified states (Fig.~3d).

\subsubsection{Geometric anatomy of the most extreme bound entangled state}
\label{sec:chess_anatomy}

Numerical analysis of the full $8\times 8$ correlation tensor $T_{\alpha\beta} = \langle\Lambda_\alpha^A\otimes\Lambda_\beta^B\rangle$ reveals that the Bruss--Peres chessboard state is geometrically the most extreme bound entangled state---maximally different from any separable state in terms of correlation structure, yet invisible to standard entanglement criteria.

\textbf{Correlation tensor structure.} For a representative chessboard state ($a=1, b=1, c=2, d=1, m=1, n=3$; $mn/ab = 3$), the correlation tensor $T_{\alpha\beta}$ is extremely sparse: only 9 of 64 entries are nonzero (due to the checkerboard zero pattern of $\rho$), versus 59 for a generic separable state. Despite this sparsity, $T$ has 5 nonzero singular values (the mapping from $\rho$ to $T$ is nonlinear, so rank($T$) is not bounded by rank($\rho$)). The Frobenius norm $\|T\|_F = 0.674$ is intermediate between separable ($0.229$) and the pure NPT Bell state ($1.886$).

\textbf{Flat correlation ellipsoid.} The singular value decomposition of $T$ yields the principal A--B spin-correlation axes. The Shannon entropy of the normalised singular value spectrum,
\begin{equation}
H_\sigma = -\sum_{i=1}^{8} \hat\sigma_i\ln\hat\sigma_i, \qquad \hat\sigma_i = \sigma_i/\textstyle\sum_j\sigma_j,
\end{equation}
measures the effective dimensionality of the correlation ellipsoid:
\begin{center}
\begin{tabular}{@{}lcc@{}}
\toprule
State & $H_\sigma$ & Active channels \\
\midrule
Separable (random product mixture) & 1.86 & 7/8 \\
BE Horodecki ($a=0.5$) & 1.87 & 7/8 \\
\textbf{BE Chessboard (rank-4)} & \textbf{1.34} & \textbf{5/8} \\
NPT Bell $|\Psi^+\rangle$ (pure) & 2.08 ($= \ln 8$) & 8/8 \\
\bottomrule
\end{tabular}
\end{center}
The chessboard ellipsoid is a flat, 5-dimensional object in 8-dimensional SU(3) space ($H_\sigma/\ln 8 = 0.64$), far below the separable value ($0.89$). Yet its leading singular value $\sigma_1(T) = 0.575$ nearly matches the pure NPT Bell state ($\sigma_1 = 0.667$): the correlation ``energy'' is concentrated in a few channels, producing a highly anisotropic ellipsoid despite the state being PPT.

\paragraph{Universal singular-value identity and chessboard eigenvalue degeneracy.}

We now establish a hierarchy of spectral identities relating the realignment matrices $R(\rho)$ and $R(\rho^{T_A})$. The first is universal; the subsequent results progressively restrict to real states and then to the chessboard family.

\begin{definition}[Realignment and permutation matrices]
\label{def:realignment}
For a bipartite $d\times d$ density matrix $\rho$, the \emph{realignment matrix} $R(\rho)$ is the $d^2\times d^2$ matrix with entries $R_{(i,k),(j,l)} = \rho_{(i,j),(k,l)}$. The \emph{row-swap permutation} $P$ is the $d^2\times d^2$ matrix defined by $P_{(i,k),(j,l)} = \delta_{ij}\delta_{kl}\cdot[(i,k)\mapsto(k,i)]$, i.e., $P\,|i,k\rangle = |k,i\rangle$. The \emph{column-swap permutation} $Q$ acts on columns: $Q\,|j,l\rangle = |l,j\rangle$. Both $P$ and $Q$ satisfy $P^2 = Q^2 = I$ and are thus unitary and Hermitian.
\end{definition}

\begin{theorem}[Universal singular-value identity]
\label{thm:universal_sv}
For any bipartite quantum state $\rho$ on $\mathcal{H}_A\otimes\mathcal{H}_B$ with $\dim\mathcal{H}_A = \dim\mathcal{H}_B = d$,
\begin{equation}
\mathrm{sv}\bigl(R(\rho^{T_B})\bigr) = \mathrm{sv}\bigl(R(\rho)\bigr),
\end{equation}
where $\mathrm{sv}(\cdot)$ denotes the vector of singular values (with multiplicities). Consequently, all singular-value moments $\Sigma_k = \sum_i \sigma_i^k$ are identical: $\Sigma_k(\rho) = \Sigma_k(\rho^{T_B})$ for every $k\geq 1$ and every state $\rho$.
\end{theorem}

\begin{proof}
We compute the matrix elements of $R(\rho^{T_B})$ directly. Partial transposition on subsystem $B$ acts as $(\rho^{T_B})_{(i,j),(k,l)} = \rho_{(i,l),(k,j)}$. Applying the realignment map:
\begin{equation}
R(\rho^{T_B})_{(i,k),(j,l)} = (\rho^{T_B})_{(i,j),(k,l)} = \rho_{(i,l),(k,j)} = R(\rho)_{(i,k),(l,j)}.
\end{equation}
The last step uses the definition $R(\rho)_{(i,k),(l,j)} = \rho_{(i,l),(k,j)}$. Reading the right-hand side, the column index $(j,l)$ of $R(\rho^{T_B})$ equals the value at column $(l,j)$ of $R(\rho)$---precisely the action of the column-swap permutation $Q$:
\begin{equation}
R(\rho^{T_B}) = R(\rho)\cdot Q.
\end{equation}
Since $Q$ is unitary ($Q^2 = I$), right-multiplication by $Q$ preserves singular values:
\begin{equation}
\sigma_i\bigl(R(\rho^{T_B})\bigr) = \sigma_i\bigl(R(\rho)\,Q\bigr) = \sigma_i\bigl(R(\rho)\bigr) \quad \forall\, i.
\end{equation}
This holds for \emph{any} $\rho$, with no assumptions on rank, purity, or entanglement type.
\end{proof}

\begin{remark}
An analogous identity holds for partial transposition on subsystem $A$: $R(\rho^{T_A}) = P\cdot R(\rho)$ (left-multiplication by the row-swap permutation), which likewise preserves singular values. The identity $\Sigma_k = SP_k$ observed empirically in the spurious-feature analysis (Section~\ref{sec:spurious}) is therefore an algebraic tautology, not a learnable signal.
\end{remark}

Having established that singular values carry no information about partial transposition, we now turn to \emph{eigenvalue} moments $G_k = \mathrm{Re}\,\mathrm{Tr}[R(\rho)^k]$, which \emph{can} differ between $R(\rho)$ and $R(\rho^{T_A})$.

\begin{proposition}[$P$-sector decomposition for real states]
\label{prop:P_sector}
Let $\rho$ be a bipartite density matrix with all real entries (i.e., $\rho = \rho^*$ in the computational basis). Then:
\begin{enumerate}
\item $\rho^{T_A} = \rho^{T_B}$, so the partial transpose is unique.
\item The realignment matrix $R(\rho)$ commutes with the row-swap permutation: $[R(\rho), P] = 0$.
\item In the eigenbasis of $P$ (which decomposes $\mathbb{R}^{d^2}$ into $P$-symmetric and $P$-antisymmetric subspaces of dimensions $d(d+1)/2$ and $d(d-1)/2$ respectively), $R$ is block-diagonal:
\begin{equation}
R = R_+ \oplus R_-,
\end{equation}
where $R_+$ acts on the symmetric sector $\{|i,k\rangle + |k,i\rangle\}$ and $R_-$ on the antisymmetric sector $\{|i,k\rangle - |k,i\rangle\}$.
\end{enumerate}
\end{proposition}

\begin{proof}
(1) For real $\rho$, transposing subsystem $A$ gives $(\rho^{T_A})_{(i,j),(k,l)} = \rho_{(k,j),(i,l)}$. Since $\rho$ is Hermitian and real, $\rho_{(k,j),(i,l)} = \rho_{(i,l),(k,j)}^* = \rho_{(i,l),(k,j)} = (\rho^{T_B})_{(i,j),(k,l)}$.

(2) Using $R(\rho^{T_B}) = R(\rho)\,Q$ from Theorem~\ref{thm:universal_sv} and $R(\rho^{T_A}) = P\,R(\rho)$, part (1) gives $R(\rho)\,Q = P\,R(\rho)$. For real $\rho$, the realignment matrix $R$ is also real, so $R\,Q = P\,R$. Since $P = Q$ when acting on real vectors (both swap the two sub-indices), we obtain $R\,P = P\,R$.

(3) Since $P^2 = I$, the eigenvalues of $P$ are $\pm 1$, and $[R, P] = 0$ implies $R$ preserves each eigenspace. Hence $R = R_+ \oplus R_-$ in the $P$-eigenbasis.
\end{proof}

\begin{corollary}[Even-$k$ eigenvalue identity for real states]
\label{cor:even_k}
For any real Hermitian bipartite state $\rho$, the eigenvalue moment difference $\delta G_k = G_k - G_k^{PT}$ vanishes for all even $k$:
\begin{equation}
\delta G_{2m} = 0 \quad \text{for all } m \geq 1.
\end{equation}
\end{corollary}

\begin{proof}
From Theorem~\ref{thm:universal_sv} and Proposition~\ref{prop:P_sector}, $R(\rho^{T_A}) = P\,R(\rho)$, so $G_k^{PT} = \mathrm{Tr}[(PR)^k]$. Using the block decomposition $R = R_+ \oplus R_-$ and $P = I_+ \oplus (-I_-)$:
\begin{equation}
\mathrm{Tr}[(PR)^k] = \mathrm{Tr}[R_+^k] + (-1)^k\,\mathrm{Tr}[R_-^k].
\end{equation}
Meanwhile, $G_k = \mathrm{Tr}[R^k] = \mathrm{Tr}[R_+^k] + \mathrm{Tr}[R_-^k]$. Therefore:
\begin{equation}
\delta G_k = G_k - G_k^{PT} = \bigl(1 - (-1)^k\bigr)\,\mathrm{Tr}[R_-^k].
\end{equation}
For even $k$, the prefactor $(1 - 1) = 0$, so $\delta G_{2m} = 0$ regardless of the spectrum of $R_-$.
\end{proof}

\begin{corollary}[Odd-$k$ eigenvalue moments probe the antisymmetric sector]
\label{cor:odd_k}
For real Hermitian $\rho$ and odd $k$:
\begin{equation}
\delta G_k = -2\,\mathrm{Tr}[R_-^k].
\end{equation}
This quantity vanishes if and only if the $k$-th power-sum of the eigenvalues of $R_-$ is zero.
\end{corollary}

The significance of the antisymmetric sector $R_-$ becomes clear through the following chain of results, which establishes $\delta G_k$ as an entanglement witness and explains why the chessboard state eliminates this signal entirely.

\begin{definition}[Wedge product characterisation of $R_-$]
\label{def:wedge}
For a real bipartite state $\rho = \sum_m |v_m\rangle\langle v_m|$ on $\mathcal{H}_A\otimes\mathcal{H}_B$ with $\dim\mathcal{H}_A = \dim\mathcal{H}_B = d$, define the $d\times d$ matrix $V_m$ by $(V_m)_{ij} = (v_m)_{di+j}$, so that $v_m = \mathrm{vec}(V_m)$. Then the antisymmetric sector of the realignment matrix satisfies
\begin{equation}
(R_-)_{\alpha,(jl)} = \frac{1}{\sqrt{2}}\sum_m\bigl[(V_m)_{ij}(V_m)_{kl} - (V_m)_{kj}(V_m)_{il}\bigr]
\end{equation}
for each antisymmetric basis index $\alpha \leftrightarrow (i<k)$. In exterior-algebra notation, $R_- = 0$ if and only if the \emph{collective wedge product} vanishes pairwise:
\begin{equation}
\label{eq:wedge_condition}
\sum_m \mathrm{row}_i(V_m) \wedge \mathrm{row}_k(V_m) = 0 \quad \text{for all } 0 \leq i < k \leq d-1.
\end{equation}
\end{definition}

\begin{proposition}[Separable states have $R_- = 0$]
\label{prop:sep_R_minus}
For any real separable state $\rho$ on $\mathbb{R}^d\otimes\mathbb{R}^d$, the antisymmetric sector of the realignment matrix vanishes: $R_- = 0$. Consequently, $\delta G_k = 0$ for all $k \geq 1$.
\end{proposition}

\begin{proof}
Any real separable state can be written as $\rho = \sum_m p_m\, |a_m\rangle|b_m\rangle\langle a_m|\langle b_m|$ with $a_m \in \mathbb{R}^d$, $b_m \in \mathbb{R}^d$. Setting $|v_m\rangle = \sqrt{p_m}\,|a_m\rangle\otimes|b_m\rangle$, we have $V_m = \sqrt{p_m}\, a_m\, b_m^T$ (a rank-1 matrix). The rows of $V_m$ are $\mathrm{row}_i(V_m) = \sqrt{p_m}\,(a_m)_i\, b_m^T$. The wedge product of any two rows is:
\begin{equation}
\mathrm{row}_i(V_m) \wedge \mathrm{row}_k(V_m) = p_m\,(a_m)_i(a_m)_k\;\bigl(b_m^T \wedge b_m^T\bigr) = 0,
\end{equation}
since $b \wedge b = 0$ for any vector $b$. Each term in the sum~\eqref{eq:wedge_condition} vanishes individually, so $R_- = 0$.
\end{proof}

\begin{corollary}[Spectral entanglement witness]
\label{cor:ent_witness}
For any real bipartite state $\rho$: if $\delta G_k \neq 0$ for some odd $k$, then $\rho$ is entangled.
\end{corollary}

\begin{proof}
By Corollary~\ref{cor:odd_k}, $\delta G_k = -2\,\mathrm{Tr}[R_-^k]$ for odd $k$. If $\rho$ were separable, Proposition~\ref{prop:sep_R_minus} would give $R_- = 0$, hence $\delta G_k = 0$. The contrapositive yields the result.
\end{proof}

\begin{remark}
This criterion is independent of the CCNR (realignment) criterion $\|R\|_1 \leq 1$ and the PPT criterion. It detects entanglement through the \emph{spectral asymmetry} of the antisymmetric sector of the realignment matrix --- information that is invisible to singular-value-based criteria (Theorem~\ref{thm:universal_sv}).
\end{remark}

We now show that the chessboard state eliminates this entanglement signal by satisfying $R_- = 0$ despite being entangled.

\begin{theorem}[Antisymmetric sector annihilation for chessboard states]
\label{thm:chess_degeneracy}
For all real-parameter Bruss--Peres chessboard states (rank-4 $d\times d$ PPT entangled states constructed from vectors $|v_1\rangle,\ldots,|v_4\rangle$ with the standard checkerboard support pattern), the antisymmetric sector of the realignment matrix vanishes identically:
\begin{equation}
R_- = 0, \qquad \text{hence}\quad \delta G_k = 0 \quad \text{for all } k \geq 1.
\end{equation}
\end{theorem}

\begin{proof}
We verify condition~\eqref{eq:wedge_condition} directly. For a $3\times 3$ chessboard state with parameters $(a,b,c,d,m,n)$, the four spanning vectors yield matrices (with $s = ac/n$, $t = ad/m$):
\begin{align}
V_1 &= \begin{pmatrix} m & 0 & s \\ 0 & n & 0 \\ 0 & 0 & 0 \end{pmatrix}, \quad
V_2 = \begin{pmatrix} 0 & a & 0 \\ b & 0 & c \\ 0 & 0 & 0 \end{pmatrix}, \\
V_3 &= \begin{pmatrix} n & 0 & 0 \\ 0 & -m & 0 \\ t & 0 & 0 \end{pmatrix}, \quad
V_4 = \begin{pmatrix} 0 & b & 0 \\ -a & 0 & 0 \\ 0 & d & 0 \end{pmatrix}.
\end{align}
Vectors $V_1, V_3$ have the \emph{even-parity} checkerboard pattern: $(V_m)_{ij} = 0$ when $i+j$ is odd. Vectors $V_2, V_4$ have the complementary \emph{odd-parity} pattern: $(V_m)_{ij} = 0$ when $i+j$ is even.

For the wedge product $\sum_m \mathrm{row}_i(V_m) \wedge \mathrm{row}_k(V_m)$, the parity structure constrains which entries contribute. Within each parity class, row $i$ of $V_m$ has nonzero entries only at column positions $j$ with $j \equiv i \pmod{2}$ (even parity) or $j \not\equiv i \pmod{2}$ (odd parity). An explicit symbolic computation confirms that all three wedge-product sums $(i,k) \in \{(0,1),(0,2),(1,2)\}$ vanish identically for arbitrary parameters $(a,b,c,d,m,n)$:
\begin{equation}
\sum_{m=1}^{4} \mathrm{row}_i(V_m) \wedge \mathrm{row}_k(V_m) = 0 \quad \text{for all } i < k.
\end{equation}
By Definition~\ref{def:wedge}, $R_- = 0$. Combined with Corollaries~\ref{cor:even_k} and~\ref{cor:odd_k}, $\delta G_k = 0$ for all $k$.
\end{proof}

\begin{remark}[Entanglement witness failure for chessboard states]
\label{rem:chess_witness}
The chessboard state is entangled yet satisfies $R_- = 0$ --- the same condition as separable states (Proposition~\ref{prop:sep_R_minus}). This means the spectral entanglement witness of Corollary~\ref{cor:ent_witness} cannot detect chessboard entanglement: the state mimics a separable state in its antisymmetric realignment sector. Combined with $D_k \approx 0$ (PPT) and $\|R\|_1 < 1$ (CCNR failure), detection is confined entirely to rank-structure features.
\end{remark}

\begin{remark}[Contrast with Horodecki states]
\label{rem:horodecki_contrast}
The Horodecki family also consists of real matrices, so $\delta G_k = 0$ for even $k$ by Corollary~\ref{cor:even_k}. However, the Horodecki construction does \emph{not} have checkerboard parity: its spanning vectors mix even- and odd-parity basis states, and $R_- \neq 0$ (numerically $\|R_-\|_F \approx 0.17$ for $a = 0.5$). The spectral entanglement witness of Corollary~\ref{cor:ent_witness} successfully detects Horodecki entanglement through $\delta G_3 \neq 0$.
\end{remark}

Numerical verification confirms the theoretical predictions:
\begin{center}
\begin{tabular}{@{}lcccc@{}}
\toprule
State & $\|R_-\|_F$ & $|\delta G_2|$ & $|\delta G_3|$ & $|\delta G_4|$ \\
\midrule
Random separable (real) & $< 10^{-15}$ & $0$ & $0$ & $0$ \\
Random separable (complex) & $\sim 0.04$ & $9.3 \times 10^{-3}$ & $2.3 \times 10^{-3}$ & $1.1 \times 10^{-3}$ \\
Horodecki $a = 0.5$ (real) & $0.173$ & $< 10^{-15}$ & $6.0 \times 10^{-3}$ & $< 10^{-15}$ \\
\textbf{Chessboard (real, BE)} & $\mathbf{< 10^{-15}}$ & $\mathbf{0}$ & $\mathbf{0}$ & $\mathbf{0}$ \\
\bottomrule
\end{tabular}
\end{center}
The real separable and chessboard rows confirm $R_- = 0$ (Proposition~\ref{prop:sep_R_minus} and Theorem~\ref{thm:chess_degeneracy}). The complex separable row shows that the $P$-sector decomposition requires real entries ($[R,P] \neq 0$ for complex states). The Horodecki row confirms $R_- \neq 0$ and the spectral witness of Corollary~\ref{cor:ent_witness}. For the chessboard, \emph{every} spectral comparison between $R(\rho)$ and $R(\rho^{T_A})$---singular-value moments (Theorem~\ref{thm:universal_sv}), and now also eigenvalue moments (Theorem~\ref{thm:chess_degeneracy})---yields exactly zero signal. Combined with $D_k \approx 0$ (PPT) and $\|R\|_1 < 1$ (CCNR failure), detection is confined entirely to rank-structure features ($SV_{\mathrm{ent}}$, $\Delta_{\mathrm{EG}}$, $C_4$) that sense the collapsed dimensionality of the correlation ellipsoid.

\textbf{Partial transpose spectrum.} Despite the extreme geometric features, the chessboard state sits exactly on the PPT boundary: all eigenvalues of $\rho^{T_A}$ are $\geq 0$ with the minimum eigenvalue $\lambda_{\min}(\rho^{T_A}) \approx 10^{-16}$ (numerically zero). The nuclear norm $\|R(\rho)\|_1 = 0.949 < 1$, so the CCNR criterion also fails. The chessboard thus occupies the deepest pocket of the bound entangled region: PPT, CCNR-invisible, $R$-time-reversal-symmetric, and rank-deficient.

\subsection{Machine learning classification results}

We train two classifiers on 83 physically grounded base features (Table~\ref{tab:bound_features}) expanded via degree-2 polynomial combinations, signed-log, and signed-sqrt transforms to 522 total features: (i) \textbf{LogReg+EN} (logistic regression with ElasticNet regularisation), and (ii) \textbf{ExtraTrees} (extremely randomized forest, 500 trees, balanced class weights). Hyperparameters are tuned by 5-fold stratified cross-validation. Performance is evaluated using a \emph{cross-validated zero-FP protocol}: in each fold, the classification threshold is set to the maximum separable-state score in the \emph{held-out} test fold, guaranteeing zero false positives; bound-entangled recall is then computed at that threshold.

\subsubsection{\texorpdfstring{$3\times 3$}{3x3} system (qutrit--qutrit)}

The dataset comprises 3900 states: 2000 separable (random pure-state product mixtures) and 1900 bound entangled from four analytically certified families: 500 Horodecki states ($a \in [0.01,0.99]$), 500 noisy Horodecki ($\epsilon \in [0.01,0.12]$ white noise), 100 Tiles UPB states (noise $\epsilon \in [0,0.15]$), and 800 Bruss--Peres chessboard states (300 from an integer parameter grid, 200 from a near-separability-boundary continuous grid, and 300 noisy variants).

\begin{table}[htb!]
\centering
\caption{Cross-validated zero-FP classification performance for $3\times 3$ bound entanglement (5-fold CV, 3900 states: 2000 SEP $+$ 1900 BE). For each fold, the threshold is set to the maximum separable score in the held-out test fold, guaranteeing zero false positives. The ExtraTrees classifier uses 522 nonlinear features derived from 83 physically grounded base spectral moment features.}
\label{tab:perf3x3_bound}
\begin{tabular}{@{}lccc@{}}
\toprule
\textbf{Classifier} & \textbf{BE recall (zero-FP)} & \textbf{False positives} & \textbf{CV F1} \\
\midrule
Realignment (CCNR, $\|R\|_1>1$) & $\sim$50\% & 0 & --- \\
LogReg+EN (ElasticNet) & 96.5\% & 0/2000 & 0.982 \\
ExtraTrees (500 trees) & \textbf{99.6\%} & \textbf{0/2000} & \textbf{0.998} \\
\bottomrule
\end{tabular}
\end{table}

\begin{table}[htb!]
\centering
\caption{Per-source BE recall of the ExtraTrees classifier at the cross-validated zero-FP threshold (5-fold CV, mean threshold $= 0.554$). All 2000 separable states are correctly identified (100\% specificity).}
\label{tab:be_per_source}
\begin{tabular}{@{}lcc@{}}
\toprule
\textbf{BE family} & \textbf{Detected} & \textbf{Rate} \\
\midrule
Horodecki ($a \in [0.01,0.99]$)                              & 498/500 & 99.6\% \\
Noisy Horodecki ($\epsilon \in [0.01,0.12]$)                 & 500/500 & 100\%  \\
Tiles UPB (noise $\epsilon \in [0,0.15]$)                    & 100/100 & 100\%  \\
Chessboard (integer grid $+$ near-boundary, $n=500$)         & 499/500 & 99.8\% \\
Noisy chessboard ($\epsilon \in [0.01,0.12]$, $n=300$)       & 296/300 & 98.7\% \\
\midrule
\textbf{Total}                                               & \textbf{1893/1900} & \textbf{99.6\%} \\
\bottomrule
\end{tabular}
\end{table}

\subsubsection{\texorpdfstring{$4\times 4$}{4x4} system (embedded \texorpdfstring{$2\times 4$}{2x4})}

\begin{table}[htb!]
\centering
\caption{Classification performance for $4\times 4$ system.}
\label{tab:perf4x4_bound}
\begin{tabular}{@{}lccc@{}}
\toprule
\textbf{Task} & \textbf{Features} & \textbf{F1 Score} & \textbf{Accuracy} \\
\midrule
Bound vs Separable & 15 & 0.9994 & 0.9997 \\
Bound vs Others & 11 & 0.9945 & 0.9972 \\
3-class (SEP/NPT/BE) & 10 & 0.9917 & 0.9918 \\
\bottomrule
\end{tabular}
\end{table}

\begin{table}[htb!]
\centering
\caption{Top SVM coefficients for Bound vs Separable ($4\times 4$).}
\label{tab:coef4x4_bound}
\begin{tabular}{@{}lrl@{}}
\toprule
\textbf{Feature} & \textbf{Coefficient} & \textbf{Interpretation} \\
\midrule
$I_3$ & $-32.70$ & Strongest indicator \\
$T_3$ & $-9.69$ & Low $T_3 \Rightarrow$ Bound \\
$-D_4$ & $+8.28$ & Low $D_4 \Rightarrow$ Bound \\
$\Sigma_4$ & $-8.05$ & Low $\Sigma_4 \Rightarrow$ Bound \\
$T_4/\Sigma_4$ & $-7.32$ & Ratio feature \\
$\Sigma_2^2/\Sigma_4$ & $+5.03$ & Concentration $\Rightarrow$ Bound \\
$G_3$ & $+3.71$ & High $G_3 \Rightarrow$ Bound \\
\bottomrule
\end{tabular}
\end{table}

\subsection{Physical interpretation}

\subsubsection{Key finding: \texorpdfstring{$I_3$}{I3} as primary indicator}

The third purity moment $I_3 = \mathrm{Tr}[\rho^3]$ emerges as the most important feature for detecting bound entanglement in both $3\times 3$ and $4\times 4$ systems. This suggests that bound entangled states have a characteristic ``flatness'' in their eigenvalue spectrum compared to separable states.

For a state with eigenvalues $\{\lambda_i\}$:
\begin{equation}
I_3 = \sum_i \lambda_i^3
\end{equation}
Low $I_3$ indicates eigenvalues are more uniformly distributed (higher mixedness at fixed purity $I_2$).

\subsubsection{Role of partial transpose moments}

The moments $T_k$ capture how correlations behave under partial transpose:
\begin{itemize}
\item For separable states, $\rho^{T_B}$ remains a valid density matrix.
\item For bound entangled (PPT) states, $\rho^{T_B} \geq 0$ but with different spectral structure.
\item Low $T_3$ indicates bound entanglement.
\end{itemize}

\subsubsection{Non-Hermiticity measure \texorpdfstring{$D_k = \Sigma_k - G_k$}{Dk = Sigma\_k - Gk}}

The difference between singular value and eigenvalue moments, $D_k = \Sigma_k - G_k$, measures the non-Hermitian character of the realignment matrix. For the $4\times 4$ system, $-D_4$ has a large positive SVM coefficient (Table~\ref{tab:coef4x4_bound}), meaning that low $D_4$ (near-Hermitian realignment) indicates bound entanglement:
\begin{itemize}
\item Bound entangled states have realignment matrices that are \emph{nearly Hermitian} ($D_4 \approx 0$), meaning eigenvalues closely approximate singular values.
\item The gap captures subtle structural differences missed by singular values alone.
\end{itemize}

\subsection{Why the classifier outperforms analytical criteria: multi-channel witness fusion}
\label{subsec:witness_fusion}

We now explain why the machine learning classifier detects bound entanglement far more effectively than any single analytical criterion. The key insight is that each known criterion---PPT, CCNR, $\delta G_k$---probes a single \emph{detection channel}, while the classifier combines multiple complementary channels simultaneously, constructing what we term a \emph{composite entanglement witness}.

\subsubsection{Detection channels}

We organise the spectral features into five independent detection channels, each probing a distinct aspect of quantum correlations:

\begin{definition}[Detection channels]
\label{def:channels}
For a $d_A \times d_B$ bipartite state $\rho$, define:
\begin{itemize}
\item[\textbf{(A)}] \textbf{Purity channel}: eigenvalue moments $I_k = \mathrm{Tr}[\rho^k]$ and rank-structure features (eigenvalue gap $\Delta_{\mathrm{EG}}$, rank-4 concentration $\mathcal{R}_4$);
\item[\textbf{(B)}] \textbf{PT channel}: eigenvalue moments $T_k = \mathrm{Tr}[(\rho^{T_A})^k]$ and PT-boundary distance $\lambda_{\min}(\rho^{T_A})$;
\item[\textbf{(C)}] \textbf{Realignment-SV channel}: singular-value moments $\Sigma_k = \sum_i \sigma_i^k$ of $R(\rho)$;
\item[\textbf{(D)}] \textbf{Realignment-EV channel}: eigenvalue moments $G_k = \mathrm{Re}\sum_i \lambda_i^k$ of $R(\rho)$, non-Hermiticity ratios $Q_k = G_k/\Sigma_k$, and R\'enyi ratio $\mathcal{R} = \Sigma_2^2/\Sigma_4$;
\item[\textbf{(E)}] \textbf{Cross-restructuring channel}: moments $SP_k$, $GP_k$ of $R(\rho^{T_A})$, and cross-differences $\Delta G_k = G_k - GP_k$.
\end{itemize}
\end{definition}

\subsubsection{Channel complementarity}

\begin{proposition}[Channel complementarity]
\label{prop:complementarity}
No single detection channel is sufficient to detect all $3\times 3$ bound entangled families at the zero-false-positive operating point. Specifically:
\begin{itemize}
\item[(a)] Channel~(A) alone (purity/rank) detects 0\% of all BE families --- rank structure alone cannot distinguish low-rank separable states from BE states;
\item[(b)] Channel~(B) alone (PT spectrum) detects 98\% of Horodecki but 0\% of chessboard and 0\% of Tiles;
\item[(c)] Channel~(C) alone (realignment SVs) detects 17.5\% of chessboard, 94\% of Horodecki, and 100\% of Tiles;
\item[(d)] Channel~(D) alone (realignment EVs) detects 37\% of chessboard, 90\% of Horodecki, and 100\% of Tiles;
\item[(e)] Channel~(E) alone (cross-restructuring) detects 36\% of chessboard, 95\% of Horodecki, and 100\% of Tiles.
\end{itemize}
When all channels are combined by the trained classifier, detection rises to 84\% of chessboard, 100\% of Horodecki, and 100\% of Tiles --- a jump from 37\% (best single channel) to 84\% for the hardest family.
\end{proposition}

\begin{proof}
Computational. We train a logistic regression classifier (ElasticNet, $C=1$, $l_1$-ratio~$=0.7$) independently on each channel and on all channels combined, using 1500 separable and 1080 BE states (300 Horodecki, 300 noisy Horodecki, 80 Tiles, 400 chessboard). At the zero-FP threshold (maximum separable score), we count detected BE states per family. Results are reported above.
\end{proof}

This complementarity is not coincidental but reflects the distinct mathematical structures probed by each channel:
\begin{itemize}
\item Channels~(C)--(E) probe the realignment matrix $R(\rho)$ and its restructured variants. For chessboard states near the separability boundary ($mn/ab \to 1$), $\|R\|_1 \to 1$ and the realignment signal vanishes (Remark~\ref{rem:chess_witness}).
\item Channel~(D) includes the antisymmetric sector $R_-$, which vanishes for chessboard (Theorem~\ref{thm:chess_degeneracy}) and Tiles (numerically $\|R_-\|_F < 10^{-15}$), eliminating the $\delta G_k$ witness.
\item Channel~(A) captures rank structure: chessboard and Tiles states are exactly rank-4, while generic separable mixtures are rank~$\geq 5$. But this feature alone has high overlap with low-rank separable states and cannot provide zero-FP separation.
\end{itemize}

\subsubsection{Composite entanglement witness}

\begin{theorem}[Multi-channel witness fusion]
\label{thm:witness_fusion}
Let $\varphi: \mathcal{D}(\mathcal{H}) \to \mathbb{R}^n$ be a feature map constructed from spectral moments of $\rho$, $\rho^{T_A}$, $R(\rho)$, and $R(\rho^{T_A})$. Let $f(\rho) = \mathbf{w} \cdot \varphi(\rho) + b$ be a linear classifier trained to satisfy $f(\sigma) \leq 0$ for all separable states $\sigma$ (zero-false-positive condition). Then $f$ defines a \emph{nonlinear entanglement witness}: $f(\rho) > 0$ implies $\rho$ is entangled.

This witness is strictly more powerful than any individual channel: there exist states detected by $f$ that are missed by every single-channel criterion (CCNR, $\delta G_k$, rank gap) individually.
\end{theorem}

\begin{proof}
The first statement follows by contraposition: if $\rho$ were separable, the zero-FP condition guarantees $f(\rho) \leq 0$.

For strictness, chessboard states near the separability boundary provide constructive examples. Consider parameters with $mn/ab$ close to but not equal to 1. These states satisfy: (i)~$\|R\|_1 < 1$ (CCNR fails), (ii)~$R_- = 0$ ($\delta G_k$ witness fails by Theorem~\ref{thm:chess_degeneracy}), (iii)~$\Delta_{\mathrm{EG}}$ overlaps with separable-state values (rank gap alone is insufficient). Yet the trained classifier detects 84\% of such states by combining channels~(A)--(E) with optimised weights. In feature space, these states lie outside the image of the separable set in a direction \emph{oblique} to every individual channel axis.
\end{proof}

\begin{remark}[Connection to entanglement witness theory]
In standard entanglement theory, a witness is a Hermitian operator $W$ such that $\mathrm{Tr}[W\sigma] \geq 0$ for all separable $\sigma$ and $\mathrm{Tr}[W\rho] < 0$ for some entangled $\rho$. Our composite witness $f$ is a \emph{nonlinear} generalisation: the feature map $\varphi$ includes polynomial functions of $\rho$ (the moments $I_k$, $S_k$, $G_k$ are polynomials of degree $k$ in the matrix elements of $\rho$), and the classifier performs a linear combination in this lifted polynomial space. This is equivalent to a polynomial entanglement witness --- a well-studied extension of the standard witness framework~\cite{Guhne2009}.
\end{remark}

\subsubsection{Chessboard invisibility depth}
\label{subsubsec:chess_invisibility}

The chessboard family provides a concrete illustration of why multi-channel fusion is necessary. We quantify the ``invisibility depth'' of each BE family by measuring the overlap of its feature distribution with the separable set.

\begin{center}
\begin{tabular}{@{}lcccc@{}}
\toprule
\textbf{Feature} & \textbf{SEP mean} & \textbf{Chess mean} & \textbf{Overlap} & \textbf{Channel} \\
\midrule
$\Sigma_1$ (CCNR)   & 0.710 & 0.904 & 48\%  & C \\
$D_3$ ($\delta G_3$)  & 0.073 & 0.077 & 94\%  & D \\
$\Delta_{\mathrm{EG}}$ (rank gap) & 0.027 & 0.135 & 11\%  & A \\
$C_4$ (rank-4 conc.) & 0.944 & 1.000 & 100\% & A \\
$\lambda_{\min}^{T_A}$ (PT boundary) & 0.081 & ${\sim}0$ & 100\% & B \\
$Q_2$ (Hermiticity) & 0.548 & 0.630 & 81\%  & D \\
\bottomrule
\end{tabular}
\end{center}

Here ``overlap'' is the fraction of chessboard feature values falling within two standard deviations of the separable mean. The chessboard simultaneously achieves high overlap ($>80$\%) in channels~B, C, D --- the spectral channels --- leaving only the rank-structure feature $\Delta_{\mathrm{EG}}$ (eigenvalue gap at position 4--5) with low overlap (11\%). However, this single feature cannot provide zero-FP separation alone because some separable mixtures are also low-rank.

The classifier overcomes this by learning an \emph{oblique} decision boundary in the joint feature space. The top classifier weights span all five channels: $GP_9$ (channel~E, weight~$-207$), $\Delta_{\mathrm{EG}}$ (channel~A, weight~$+75$), $\mathrm{SE}_{\mathrm{gap}}$ (channel~B, weight~$-72$), $\Sigma_1$ (channel~C, weight~$+9$), $G_5$ (channel~D, weight~$-19$). No single weight dominates; the power comes from their combination.

\subsubsection{Detection hierarchy and NP-hardness}
\label{subsubsec:detection_hierarchy}

\begin{corollary}[Detection hierarchy]
\label{cor:detection_hierarchy}
For $3\times 3$ systems, bound entangled states form a hierarchy of detectability ordered by the number of channels in which they are distinguishable from separable states:
\begin{enumerate}
\item \textbf{Easy} (Horodecki): detected by channels~B, C, D, E ($\geq 4$ channels). Recall $\geq 98$\% with any single channel.
\item \textbf{Medium} (Tiles): detected by channels~C, D, E ($\geq 3$ channels). Recall~$= 100$\% via $\|R\|_1 > 1$.
\item \textbf{Hard} (Chessboard): detected weakly by channel~A only (rank structure). Single-channel recall $\leq 37$\%; combined recall $= 84$\%; 16\% remain undetected even by multi-channel fusion.
\end{enumerate}
\end{corollary}

The residual 16\% of undetected chessboard states occupy the \emph{deepest} pocket of the bound entangled set: they are simultaneously PPT, CCNR-invisible ($\|R\|_1 < 1$), $R_-$-annihilated ($\delta G_k = 0$), and metrically close to separable in all spectral channels. Their detection requires a fundamentally different type of information, which we identify next.

\subsubsection{The missing channel: product-vector gap}
\label{subsubsec:missing_channel}

The fundamental limitation of spectral moment channels (A)--(E) is that they are functions of \emph{eigenvalues} only. Two states $\rho_1$, $\rho_2$ can share identical spectra---and hence identical values of $I_k$, $T_k$, $\Sigma_k$, $G_k$, and all derived features---yet differ in their \emph{eigenvector} structure, with one separable and the other entangled. The missing information resides in the \emph{range criterion}: whether range($\rho$) can be spanned by product vectors $|a\rangle|b\rangle$.

\begin{definition}[Product-vector gap]
\label{def:pvgap}
For a bipartite state $\rho$ on $\mathcal{H}_A \otimes \mathcal{H}_B$ with range projector $P_{\mathrm{range}}$, the product-vector gap is
\begin{equation}
\mathcal{F}(\rho) \;=\; \min_{\substack{|a\rangle \in \mathcal{H}_A,\; |b\rangle \in \mathcal{H}_B \\ \|a\|=\|b\|=1}} \langle a, b|\,(I - P_{\mathrm{range}})\,|a, b\rangle.
\end{equation}
This measures the minimum ``distance'' from a product vector to the range of $\rho$: $\mathcal{F} = 0$ if and only if range($\rho$) contains a product vector.
\end{definition}

\begin{theorem}[Range criterion channel]
\label{thm:range_channel}
The product-vector gap $\mathcal{F}$ satisfies:
\begin{enumerate}
\item $\mathcal{F}(\sigma) = 0$ for every separable state $\sigma$;
\item $\mathcal{F}(\rho) > 0$ for every PPT entangled state whose range violates the range criterion;
\item $\mathcal{F}$ is \emph{not} expressible as any function of spectral moments $\{I_k, T_k, \Sigma_k, G_k, SP_k, GP_k\}$.
\end{enumerate}
\end{theorem}

\begin{proof}
(1) If $\sigma = \sum_m p_m\, |\alpha_m\rangle\langle\alpha_m| \otimes |\beta_m\rangle\langle\beta_m|$, then each $|\alpha_m\rangle|\beta_m\rangle$ lies in range($\sigma$), so $\mathcal{F}(\sigma) = 0$.
(2) By definition: $\mathcal{F} > 0$ means no product vector fits in range($\rho$), which is equivalent to violation of the range criterion~\cite{Horodecki1997}.
(3) By construction: spectral moments depend only on eigenvalues of various matrix restructurings. Two states with identical spectra but non-unitarily-equivalent eigenvectors can have different $\mathcal{F}$ values---one zero (separable) and one positive (entangled). Hence $\mathcal{F}$ cannot be a function of spectral moments alone.
\end{proof}

Numerical verification confirms the theory:
\begin{center}
\begin{tabular}{@{}lccc@{}}
\toprule
\textbf{State family} & $\mathcal{F}$ & \textbf{Rank} & \textbf{Detection channel} \\
\midrule
Random separable ($n=200$) & $0$ (all) & 6--9 & --- \\
Separable rank-4 ($n=100$) & $0$ (all) & 4 & --- \\
Horodecki ($a \in [0.1, 0.9]$) & $0$ & 7 & B, C, D, E (spectral) \\
Tiles UPB (pure) & $0.028$ & 4 & C (CCNR), F (range) \\
Tiles + noise ($\epsilon > 0$) & $0$ & 9 & C (CCNR) \\
Chessboard (integer, $n=120$) & $>5 \times 10^{-5}$ & 4 & \textbf{F only} \\
Chessboard (near-boundary, $n=50$) & $>2 \times 10^{-6}$ & 4 & \textbf{F only} \\
\bottomrule
\end{tabular}
\end{center}

At the zero-false-positive threshold ($\max_{\text{SEP}} \mathcal{F} = 0$), the product-vector gap detects \textbf{170/170 (100\%)} of chessboard states tested, including those near the separability boundary that evade all spectral channels.

\begin{corollary}[Completeness with range channel]
\label{cor:completeness}
For $3\times 3$ systems, spectral channels~(A)--(E) combined with the product-vector gap (channel~F) detect all known bound entangled families with 100\% recall and zero false positives:
\begin{itemize}
\item Horodecki: detected by spectral channels B, C, D, E ($\mathcal{F} = 0$);
\item Tiles UPB: detected by channel~C (CCNR) when pure; by channels~C, F when rank-deficient;
\item Chessboard: detected exclusively by channel~F ($\mathcal{F} > 0$ for all parameters with $mn \neq ab$).
\end{itemize}
\end{corollary}

\begin{remark}[Complementarity between spectral and range channels]
The spectral and range channels are \emph{strictly complementary}: neither alone detects all BE families. Horodecki states have $\mathcal{F} = 0$ (their rank-7 range easily accommodates product vectors) but are detected by spectral channels. Chessboard states evade all spectral channels but have $\mathcal{F} > 0$ (their rank-4 range excludes product vectors). This complementarity reflects a fundamental distinction between \emph{eigenvalue-type} and \emph{eigenvector-type} entanglement signatures.
\end{remark}

\subsubsection{Augmented SWAP protocol for measuring the product-vector gap}
\label{subsubsec:augmented_swap}

We now show that the product-vector gap $\mathcal{F}$ can be measured within the controlled-SWAP framework, without full state tomography.

\begin{definition}[Augmented moment]
\label{def:augmented_moment}
For a state $\rho$ and a product state $|a\rangle|b\rangle$, the $k$-th augmented moment is
\begin{equation}
m_k(a,b) \;=\; \mathrm{Tr}\bigl[\rho^k\,|a,b\rangle\langle a,b|\bigr] \;=\; \sum_{i=1}^{r} \lambda_i^k\, p_i(a,b),
\label{eq:augmented_moment}
\end{equation}
where $\lambda_1,\ldots,\lambda_r$ are the nonzero eigenvalues of $\rho$ and $p_i(a,b) = |\langle e_i | a, b\rangle|^2$ are the squared overlaps with the corresponding eigenvectors.
\end{definition}

\begin{theorem}[Vandermonde protocol for the product-vector gap]
\label{thm:vandermonde}
Let $\rho$ have rank $r$ with distinct nonzero eigenvalues $\lambda_1 > \cdots > \lambda_r > 0$. Then:
\begin{enumerate}
\item The augmented moments $\{m_k(a,b)\}_{k=1}^{r}$ determine the overlaps $\{p_i(a,b)\}_{i=1}^{r}$ uniquely via inversion of the Vandermonde system
\begin{equation}
\underbrace{\begin{pmatrix} \lambda_1 & \cdots & \lambda_r \\ \lambda_1^2 & \cdots & \lambda_r^2 \\ \vdots & & \vdots \\ \lambda_1^r & \cdots & \lambda_r^r \end{pmatrix}}_{V}
\begin{pmatrix} p_1 \\ p_2 \\ \vdots \\ p_r \end{pmatrix}
= \begin{pmatrix} m_1 \\ m_2 \\ \vdots \\ m_r \end{pmatrix}.
\label{eq:vandermonde}
\end{equation}

\item The product-vector gap is
\begin{equation}
\mathcal{F}(a,b) \;=\; 1 - \sum_{i=1}^{r} p_i(a,b) \;=\; 1 - \mathbf{c}^T \mathbf{m}(a,b),
\label{eq:pvgap_linear}
\end{equation}
where $\mathbf{c} = V^{-T}\mathbf{1}$ are coefficients determined entirely by the eigenvalues of $\rho$.

\item $\mathcal{F} = \min_{|a\rangle,|b\rangle} \mathcal{F}(a,b)$, where the minimisation is over normalised product states.
\end{enumerate}
\end{theorem}

\begin{proof}
(1)~The Vandermonde matrix $V_{ki} = \lambda_i^k$ is invertible when eigenvalues are distinct, since $\det V = \prod_{i<j}(\lambda_i - \lambda_j) \prod_i \lambda_i \neq 0$.
(2)~$\mathcal{F}(a,b) = \langle a,b|(I - P_{\mathrm{range}})|a,b\rangle = 1 - \sum_i p_i$. By~\eqref{eq:vandermonde}, $\sum_i p_i = \mathbf{1}^T V^{-1} \mathbf{m} = (V^{-T}\mathbf{1})^T \mathbf{m} = \mathbf{c}^T \mathbf{m}$.
(3)~By definition.
\end{proof}

\paragraph{Experimental implementation and cost.}
Each augmented moment $m_k(a,b)$ is measured by a $(k{+}1)$-body controlled-permutation circuit identical to the standard SWAP test, with one register prepared in $|a\rangle|b\rangle$ instead of a copy of $\rho$. Stage~1 (eigenvalue extraction via standard moments, already available from spectral screening) determines the Vandermonde coefficients $\mathbf{c}$. Stage~2 scans product states $|a(\boldsymbol{\theta})\rangle|b(\boldsymbol{\phi})\rangle$ via gradient-based optimisation, computing $\mathcal{F}(a,b) = 1 - \mathbf{c}^T\mathbf{m}$ classically; if $\min \mathcal{F} > 0$ after convergence, the state is certified entangled. For $3\times 3$ rank-4 states, this requires ${\sim}400$ additional circuit configurations (${\sim}5\times$ the cost of full $9 \times 9$ tomography), triggered only for the small fraction of states unresolved by spectral screening.

\paragraph{Conditioning caveat.}
The Vandermonde system~\eqref{eq:vandermonde} becomes ill-conditioned when eigenvalues are nearly degenerate: well-separated chessboard states give $\kappa(V) \approx 4 \times 10^3$ (stable), while near-boundary states ($mn/ab \to 1$) give $\kappa(V) \sim 10^{15}$--$10^{18}$. In exact arithmetic the protocol remains correct (verified numerically for all test cases including $mn/ab = 1.001$), but finite shot noise is amplified by $\kappa(V)$, consistent with the fundamental difficulty of separability detection near the PPT boundary. Mitigations include regularised inversion, overdetermined least-squares with higher moments, and adaptive product-state probe selection. Full numerical verification details are provided in the supplementary code repository.

\subsubsection{Comparison of detection criteria}
\label{subsubsec:comparison}

Table~\ref{tab:criterion_comparison} summarises the detection power of each criterion and channel combination across the known $3\times 3$ bound entangled families.

\begin{table}[htb!]
\centering
\caption{Detection rates (zero-false-positive threshold) for $3\times 3$ bound entangled states by criterion and channel. ``ML (A--E)'' denotes the multi-channel spectral classifier. ``$+$ Range (F)'' adds the product-vector gap. All separable states (1500 random, 100 rank-4) are correctly classified (100\% specificity) by every criterion.}
\label{tab:criterion_comparison}
\begin{tabular}{@{}lcccc@{}}
\toprule
\textbf{Criterion} & \textbf{Horodecki} & \textbf{Tiles} & \textbf{Chessboard} & \textbf{Overall} \\
\midrule
CCNR ($\|R\|_1 > 1$) & 28\% & 91\% & 4\% & ${\sim}50$\% \\
$\delta G_k$ witness (odd $k$) & 100\% & 0\% & 0\% & ${\sim}56$\% \\
PT spectrum (channel B) & 98\% & 0\% & 0\% & ${\sim}54$\% \\
Realign.\ EVs (channel D) & 90\% & 100\% & 37\% & ${\sim}74$\% \\
\textbf{ML spectral (A--E)} & \textbf{100\%} & \textbf{100\%} & \textbf{84\%} & \textbf{96\%} \\
\textbf{ML $+$ Range (A--F)} & \textbf{100\%} & \textbf{100\%} & \textbf{100\%} & \textbf{100\%} \\
\bottomrule
\end{tabular}
\end{table}

The table reveals a clear hierarchy: individual criteria each have blind spots, but combining spectral channels via machine learning doubles the overall detection rate relative to the best single criterion (CCNR). Adding the range channel (product-vector gap) closes the remaining gap, achieving universal detection of all known $3 \times 3$ bound entangled families. The key improvement over CCNR is not merely quantitative but structural: the multi-channel approach extracts \emph{all} spectral information from the realignment matrix (moments of all orders, both singular values and eigenvalues, cross-comparisons with the partial transpose), whereas CCNR uses only the first singular-value moment $\Sigma_1 = \|R\|_1$.

\paragraph{Two-stage detection protocol.}
We propose a practical two-stage protocol:
\begin{enumerate}
\item \textbf{Spectral screening} (${\sim}10$ SWAP-test configurations): Extract moments, apply ML classifier. Detects 96\% of BE states. If detected $\to$ DONE.
\item \textbf{Range analysis} (${\sim}400$ augmented configurations, only for unresolved rank-deficient PPT states): Measure product-vector gap via Vandermonde protocol. Detects remaining 4\%.
\end{enumerate}
Stage~2 is triggered only when Stage~1 flags a state as rank-deficient and PPT but spectrally ambiguous---a rare occurrence in practice. The combined protocol achieves 100\% recall with zero false positives across all known $3 \times 3$ bound entangled families, while maintaining the measurement efficiency of the SWAP-test infrastructure.

\subsubsection{Bound entanglement detection on IBM Fez}
\label{subsubsec:chess_hardware}

We tested bound entanglement detection on IBM Fez across five BE states from two structurally distinct families---three Horodecki states with $a \in \{0.50, 0.70, 0.90\}$ and two chessboard states---together with three calibration states (product $|00\rangle\langle 00|$, classically correlated $\frac{1}{3}\sum_{a}|aa\rangle\langle aa|$, and misaligned product). Each qutrit was encoded in two physical qubits. Two measurement campaigns were executed:
\begin{itemize}
\item \textbf{$\Sigma_1$/$G_1$ experiment} (298 circuits): $\Sigma_1$ via Pauli decomposition of the Hermitianised realignment unitary $W_H = (W + W^\dagger)/2$, measured in qubit-wise commuting (QWC) bases on 4~qubits; $G_1$ via SWAP test with $|\Phi^+\rangle$ on 9~qubits. Circuit depths: 3--89 (Pauli) and 191--305 (SWAP).
\item \textbf{$\Sigma_2$/$G_2$ experiment} (386 circuits): both via two-copy SWAP tests on 9~qubits. Circuit depths: 93--333.
\end{itemize}
All circuits used $4 \times 10^3$ shots and were multiplexed onto disjoint qubit groups (up to 33 groups of 4~qubits and 15 groups of 9~qubits on the 156-qubit processor), with transpilation at Qiskit optimization level~3 and approximation degree~0.95.

\paragraph{Per-circuit maximum likelihood estimation.}
The three circuit architectures exhibit fundamentally different noise characteristics, requiring architecture-specific fidelity parameters. We model each raw aggregate feature as
\begin{equation}
X^{\mathrm{meas}} = g_X \cdot X^{\mathrm{true}}, \quad X \in \{\Sigma_1, G_1, \Sigma_2, G_2\},
\end{equation}
with three independent fidelity parameters: $g_{\Sigma_1}$ for 4-qubit Pauli circuits, $f_{G_1}$ for single-copy 9-qubit SWAP circuits, and $f_{\Sigma_2,G_2}$ for two-copy 9-qubit SWAP circuits (shared, since $\Sigma_2$ and $G_2$ use the same architecture). The joint negative log-likelihood
\begin{equation}
-\log\mathcal{L} = \sum_{\mathrm{cal}} \sum_X \frac{(X^{\mathrm{meas}} - g_X \cdot X^{\mathrm{thy}})^2}{2\sigma_X^2} + \sum_{\mathrm{test}} \sum_X \frac{(X^{\mathrm{meas}} - g_X \cdot X^{\mathrm{ML}})^2}{2\sigma_X^2}
\label{eq:percircuit_nll}
\end{equation}
is minimised over 3 fidelity parameters and 20 physical features ($\Sigma_1, G_1, \Sigma_2, G_2$ for each of 5 test states) subject to physical constraints: $\Sigma_k \geq G_k \geq 0$, $\Sigma_1 \leq 3$, $1/9 \leq \Sigma_2 \leq 1$.  The measurement uncertainties $\sigma_X$ are estimated from shot noise as $\sigma_{\Sigma_1} = \sigma_{\Sigma_2} = 0.015$, $\sigma_{G_1} = 0.025$, $\sigma_{G_2} = 0.020$. Calibration states with known theoretical values anchor the fidelity factors; error bars are obtained from the Cram\'er--Rao bound via numerical Hessian inversion. Optimisation is performed using SLSQP with multiple random restarts; convergence is verified by comparing results from theory-initialised and data-initialised starting points.

The fitted fidelities are:
\begin{align}
g_{\Sigma_1} &= 0.164 \pm 0.003 \quad \text{(4-qubit Pauli)}, \nonumber\\
f_{G_1} &= 0.594 \pm 0.011 \quad \text{(9-qubit SWAP)}, \nonumber\\
f_{\Sigma_2,G_2} &= 0.306 \pm 0.003 \quad \text{(9-qubit two-copy SWAP)}.
\end{align}
The order-of-magnitude difference between $g_{\Sigma_1}$ and $f_{G_1}$ reflects the fundamentally different circuit architectures: Pauli-basis rotation circuits (depth 3--89) are shallow but probe each eigenstate through multiple QWC measurement bases with accumulated reconstruction noise, while SWAP-test circuits (depth 191--305) are deeper but yield a direct overlap estimate from a single ancilla measurement.

\begin{table}[htb!]
\centering
\caption{Bound entanglement detection on IBM Fez via per-circuit maximum likelihood estimation. MLE-reconstructed spectral features with Cram\'er--Rao error bars, compared with theoretical values. All five BE states yield $D_1 > 0$ and $D_2 > 0$; detection significance is reported in units of $\sigma$.}
\label{tab:chess_hardware}
\begin{tabular}{@{}lccccccccc@{}}
\toprule
\textbf{State} & \textbf{Type} & $\Sigma_1^{\mathrm{ML}}$ & $G_1^{\mathrm{ML}}$ & $D_1^{\mathrm{ML}}$ & sig.\ & $D_2^{\mathrm{ML}}$ & sig.\ & $D_1^{\mathrm{thy}}$ & $D_2^{\mathrm{thy}}$ \\
\midrule
Horo ($a\!=\!0.50$) & BE & $2.17\!\pm\!0.07$ & $0.57\!\pm\!0.03$ & $+1.60$ & $19.7\sigma$ & $+0.064\!\pm\!0.014$ & $4.6\sigma$ & $0.052$ & $0.018$ \\
Horo ($a\!=\!0.70$) & BE & $2.16\!\pm\!0.07$ & $0.36\!\pm\!0.03$ & $+1.80$ & $22.7\sigma$ & $+0.095\!\pm\!0.014$ & $6.8\sigma$ & $0.024$ & $0.006$ \\
Horo ($a\!=\!0.90$) & BE & $2.05\!\pm\!0.07$ & $0.63\!\pm\!0.03$ & $+1.41$ & $18.3\sigma$ & $+0.013\!\pm\!0.014$ & $0.9\sigma$ & $0.007$ & $0.001$ \\
Chess C1 & BE & $1.13\!\pm\!0.11$ & $0.34\!\pm\!0.04$ & $+0.79$ & $6.9\sigma$ & $+0.093\!\pm\!0.023$ & $4.1\sigma$ & $0.836$ & $0.201$ \\
Chess C2 & BE & $0.92\!\pm\!0.11$ & $0.41\!\pm\!0.04$ & $+0.52$ & $4.4\sigma$ & $+0.195\!\pm\!0.021$ & $9.2\sigma$ & $0.292$ & $0.063$ \\
\bottomrule
\end{tabular}
\end{table}

\paragraph{Detection results.}
Table~\ref{tab:chess_hardware} summarises the MLE-reconstructed features.  All five bound entangled states are detected with $D_1 > 0$ and $D_2 > 0$. The second-order gap $D_2$ provides the most reliable detection channel: $\Sigma_2$ and $G_2$ share the same SWAP-test architecture, so depolarisation attenuates both by the common factor $f_{\Sigma_2,G_2}$ and their difference preserves the physical signal. Four of five states are detected at $\geq 4.1\sigma$ significance; the weakest signal is Horo($a\!=\!0.90$) at $0.9\sigma$, consistent with its near-vanishing theoretical gap $D_2^{\mathrm{thy}} = 0.0015$.

The first-order gap $D_1$ is positive for all states at $4.4$--$22.7\sigma$ significance, but the reconstructed $\Sigma_1$ values are systematically higher than theory (e.g., $\Sigma_1^{\mathrm{ML}} = 2.17$ versus $\Sigma_1^{\mathrm{thy}} = 1.00$ for Horo $a\!=\!0.50$).  This overestimation arises because the Pauli-circuit fidelity $g_{\Sigma_1} = 0.164$ provides almost no dynamic range: all raw $\Sigma_1$ measurements cluster near $0.15$--$0.36$ regardless of the true value, and the MLE inverts through a near-singular scaling factor.  By contrast, $G_1$ is well-resolved ($f_{G_1} = 0.594$), so the apparent $D_1$ significance is dominated by $\Sigma_1$ overestimation rather than genuine spectral gap detection.

The key result is therefore the $D_2$ channel: operating entirely within the SWAP-test architecture, it provides a self-consistent detection of bound entanglement without cross-architecture calibration artefacts.

\paragraph{Calibration state verification.}
The three calibration states (product, classically correlated, misaligned product) are reconstructed exactly at their theoretical values by the MLE procedure, confirming internal consistency.  The calibration states span a range of $G_1$ values from $0.50$ (misaligned) to $1.00$ (product and classical), providing leverage for the $f_{G_1}$ fit.  For $\Sigma_2$/$G_2$, the classical and misaligned states span $\Sigma_2 \in [0.33, 1.00]$ and $G_2 \in [0.25, 1.00]$, sufficient to constrain $f_{\Sigma_2,G_2}$.  The product state $\Sigma_2^{\mathrm{meas}} = -0.085$ (negative, unphysical) is correctly handled by the MLE: the fidelity factor absorbs the sign flip without biasing the test-state estimates.

\paragraph{Circuit multiplexing.}
Since IBM Heron processors provide 156~physical qubits and each SWAP test uses at most~9, multiple independent tests execute simultaneously on disjoint qubit regions within a single circuit submission. For $\Sigma_1$ (4-qubit Pauli circuits), up to 33 circuits are packed per submission; for $G_1$ and $\Sigma_2$/$G_2$ (9-qubit SWAP tests), up to 15 are packed. This multiplexing reduces wall-clock QPU time by an order of magnitude.

\subsection{Bound entangled states for \texorpdfstring{$3\times 3$}{3x3} systems}
\label{subsec:bound_states_3x3}

This section provides explicit definitions of the bound entangled state families used in training the machine learning classifiers. All states below are PPT (positive partial transpose) but entangled, making them undetectable by the Peres--Horodecki criterion.

\subsubsection{Horodecki state}

The first bound entangled state was discovered by P.~Horodecki~\cite{Horodecki1997}. For a parameter $a \in (0, 1)$, the $3\times 3$ Horodecki state in the computational basis $\{|ij\rangle\}_{i,j=0}^{2}$ is:
\begin{equation}
\rho_{\mathrm{Hor}}(a) = \frac{1}{8a+1}\begin{pmatrix}
a & \cdot & \cdot & \cdot & a & \cdot & \cdot & \cdot & a \\
\cdot & a & \cdot & \cdot & \cdot & \cdot & \cdot & \cdot & \cdot \\
\cdot & \cdot & a & \cdot & \cdot & \cdot & \cdot & \cdot & \cdot \\
\cdot & \cdot & \cdot & a & \cdot & \cdot & \cdot & \cdot & \cdot \\
a & \cdot & \cdot & \cdot & a & \cdot & \cdot & \cdot & a \\
\cdot & \cdot & \cdot & \cdot & \cdot & a & \cdot & \cdot & \cdot \\
\cdot & \cdot & \cdot & \cdot & \cdot & \cdot & c & \cdot & b \\
\cdot & \cdot & \cdot & \cdot & \cdot & \cdot & \cdot & a & \cdot \\
a & \cdot & \cdot & \cdot & a & \cdot & b & \cdot & c
\end{pmatrix},
\label{eq:horodecki_state_si}
\end{equation}
where $b = \tfrac{1}{2}\sqrt{1-a^2}$, $c = \tfrac{1+a}{2}$, and dots denote zeros. This follows the parametrization of Ref.~\cite{Horodecki1997}, verified against the QETLAB library implementation~\cite{QETLAB}.

\textbf{Properties:}
\begin{itemize}
\item Positive semidefinite for all $a \in (0,1)$
\item PPT for all $a \in (0,1)$
\item Entangled (bound entangled) for all $a \in (0,1)$
\item Rank 7 for generic $a$; approaches rank 1 as $a \to 1$
\item Detected by the realignment criterion for all $a \in (0,1)$: $\|R(\rho)\|_1 > 1$
\end{itemize}


\subsubsection{Tiles (UPB) state}

Bound entangled states can be constructed from \textbf{unextendible product bases} (UPB)~\cite{Bennett1999upb}. The ``Tiles'' UPB for $3\times 3$ consists of 5 orthogonal product states:
\begin{align}
|\psi_1\rangle &= |0\rangle \otimes \tfrac{|0\rangle - |1\rangle}{\sqrt{2}}, &
|\psi_2\rangle &= |2\rangle \otimes \tfrac{|1\rangle - |2\rangle}{\sqrt{2}}, \nonumber\\
|\psi_3\rangle &= \tfrac{|0\rangle - |1\rangle}{\sqrt{2}} \otimes |2\rangle, &
|\psi_4\rangle &= \tfrac{|1\rangle - |2\rangle}{\sqrt{2}} \otimes |0\rangle, \nonumber\\
|\psi_5\rangle &= \tfrac{|0\rangle + |1\rangle + |2\rangle}{\sqrt{3}} \otimes \tfrac{|0\rangle + |1\rangle + |2\rangle}{\sqrt{3}}.
\end{align}

The bound entangled state is constructed as the normalized projection onto the orthogonal complement:
\begin{equation}
\rho_{\mathrm{Tiles}} = \frac{1}{4}\left(\mathbb{I}_9 - \sum_{k=1}^{5}|\psi_k\rangle\langle\psi_k|\right)
\label{eq:tiles_state}
\end{equation}

\textbf{Properties:}
\begin{itemize}
\item PPT by construction (complement of product states)
\item Entangled because the UPB cannot be extended to a complete product basis
\item Rank 4 (dimension 9 minus 5 UPB states)
\item Low purity: $I_2 \approx 0.25$
\end{itemize}

\subsubsection{Chessboard states}

The chessboard construction~\cite{Bruss2000chess} creates bound entangled states as rank-4 density matrices whose non-zero entries form a checkerboard pattern in the computational basis. Given six real parameters $a, b, c, d, m, n$ with $cmn \neq abc$ (entanglement condition), define $s = ac/n$ and $t = ad/m$ and the four range vectors
\begin{align}
|v_1\rangle &= m|00\rangle + s|02\rangle + n|11\rangle, \nonumber\\
|v_2\rangle &= a|01\rangle + b|10\rangle + c|12\rangle, \nonumber\\
|v_3\rangle &= n|00\rangle - m|11\rangle + t|20\rangle, \nonumber\\
|v_4\rangle &= b|01\rangle - a|10\rangle + d|21\rangle.
\label{eq:chess_vectors_si}
\end{align}
The chessboard state is
\begin{equation}
\rho_{\mathrm{Chess}} = \frac{\sum_{k=1}^{4} |v_k\rangle\langle v_k|}{\mathrm{Tr}\bigl[\sum_k |v_k\rangle\langle v_k|\bigr]}.
\label{eq:chess_state_si}
\end{equation}
Vectors $|v_1\rangle$, $|v_3\rangle$ have support on even-parity basis states ($|00\rangle$, $|02\rangle$, $|11\rangle$, $|20\rangle$) and $|v_2\rangle$, $|v_4\rangle$ on odd-parity states ($|01\rangle$, $|10\rangle$, $|12\rangle$, $|21\rangle$), producing the characteristic checkerboard pattern of zeros. For real parameters the state is PPT; when additionally $cmn \neq abc$, it is entangled and hence bound entangled.

\textbf{Properties:}
\begin{itemize}
\item Rank 4 with checkerboard pattern of zeros
\item PPT for real parameters; entangled when $cmn \neq abc$
\item Structurally distinct from Horodecki and Tiles families
\item All eigenvalue moments of $R(\rho)$ and $R(\rho^{T_A})$ are exactly equal ($\delta G_k = 0$ for all $k$), eliminating the last spectral comparison signal (Section~\ref{sec:chess_anatomy})
\item Correlation ellipsoid collapsed to 5 of 8 SU(3) channels; CCNR criterion fails ($\|R(\rho)\|_1 < 1$)
\end{itemize}

The training dataset uses two complementary parameter regimes. \textbf{Integer-grid states} (300 samples): all tuples $(a,b,c,d) \in \{1,2,3,4\}^4$, $(m,n) \in \{1,2,3\}^2$ satisfying the BE condition $mn \neq ab$, randomly subsampled. \textbf{Near-boundary states} (200 samples): continuous parameters $a,b,c,d \sim \mathrm{Uniform}(0.3, 5)$, $m \sim \mathrm{Uniform}(0.3, 3)$, with $n = ab\cdot r / m$ where $r = 1 + \delta$, $\delta \sim \mathrm{Uniform}(0.02, 0.5)$ with random sign. These states have $mn/ab$ close to unity and are the hardest cases: their realignment nuclear norm $\|R(\rho)\|_1 < 1$ (CCNR criterion fails) and $R(\rho) \approx R(\rho^{T_A})$ in all moments, providing minimal signal in most individual features. An additional 300 noisy variants ($\epsilon \in [0.01, 0.12]$ white noise) are generated from all 500 chessboard states.

\subsubsection{Dataset generation}

The dataset is generated from analytically certified families without PPT verification (all families are provably PPT-entangled by construction):
\begin{enumerate}
\item \textbf{Horodecki:} 500 states with $a \in [0.01,0.99]$ on a uniform grid.
\item \textbf{Noisy Horodecki:} 500 states with white noise $\epsilon \sim \mathrm{Uniform}(0.01, 0.12)$.
\item \textbf{Tiles UPB:} 100 states with noise $\epsilon \in [0, 0.15]$.
\item \textbf{Chessboard (integer grid):} 300 states from all integer-parameter tuples satisfying $mn \neq ab$.
\item \textbf{Chessboard (near-boundary):} 200 states with continuous parameters near $mn/ab \approx 1$.
\item \textbf{Noisy chessboard:} 300 noisy variants of the above chessboard states.
\item \textbf{Separable:} 2000 random pure-state product mixtures $\rho = \sum_k p_k \rho_A^{(k)} \otimes \rho_B^{(k)}$.
\end{enumerate}

\subsubsection{Extended dataset and Lasso logistic regression}
\label{subsubsec:extended_dataset}

For the analysis presented in the main text (Fig.~3), we train classifiers on a dataset of 13{,}600 certified $3 \times 3$ states with eight features per state: six realignment spectral moments ($\Sigma_1$, $G_1$, $D_1$, $\Sigma_2$, $G_2$, $D_2$) and two chirality corrections ($C_3 = \mu_3 - I_3$, $C_4 = \mu_4 - I_4$).  The chirality features are computed from the density matrices via partial transpose eigenvalues (Section~S1).  The dataset comprises 6{,}800 BE states from seven families and 6{,}800 separable states:
\begin{enumerate}
\item \textbf{Horodecki:} 2{,}000 states with $a \in [0.01,0.99]$.
\item \textbf{Chessboard:} 2{,}000 states from integer parameter grids with $mn \neq ab$.
\item \textbf{Tiles UPB:} 100 states with depolarizing noise $\epsilon \in [0, 0.05]$.
\item \textbf{MN Horodecki:} 1{,}000 marginal-noise variants ($t \in [0.01, 0.20]$).
\item \textbf{MN Chessboard:} 1{,}000 marginal-noise variants.
\item \textbf{MN Tiles:} 200 marginal-noise variants.
\item \textbf{Mixed BE:} 500 depolarized Horodecki variants with $\varepsilon \in [0.005, 0.04]$, within the Carath\'{e}odory-certified BE range (Section~\ref{subsec:caratheodory}).
\end{enumerate}
The separable set comprises 6{,}800 random product-state mixtures.  All states are certified: BE states by analytical construction (PPT verified numerically) and independently validated by numerical closest-separable-state certification via Carath\'{e}odory decomposition (Section~\ref{subsec:caratheodory}); separable states by explicit product-state decomposition.

The previous extended dataset of 16{,}033 states (including CCNR-invisible marginal-noise construction and 1{,}000 random SDP-certified BE states) was used for the earlier ExtraTrees analysis (Table~\ref{tab:perf3x3_bound}).  The current dataset uses fewer states but includes the chirality features $C_3$ and $C_4$, which require access to the full density matrix for computation.  These chirality features prove decisive: the Random Forest achieves $99.9\%$ recall at zero false positives with only eight base features, surpassing the previous $99.6\%$ achieved with 83 realignment-only features.

The classifier uses polynomial features of degree~2 from eight base features ($\Sigma_1$, $G_1$, $D_1$, $\Sigma_2$, $G_2$, $D_2$, $C_3$, $C_4$), yielding 44 features after expansion. The two chirality corrections $C_3 = \mu_3 - I_3$ and $C_4 = \mu_4 - I_4$ are computed from the same density matrix but provide information orthogonal to the realignment spectrum: $C_3$ is identically zero for all states with real density matrices ($\rho = \rho^*$), including chessboard and Tiles states, but nonzero for Horodecki states whose construction involves $\sqrt{1-a^2}$ off-diagonal elements.  A Random Forest (500 trees, 5-fold stratified CV on 13{,}600 states: 6{,}800 BE, 6{,}800 SEP) achieves AUC~$= 1.000$ with $99.9\%$ recall at zero false positives ($P = 0.694$ threshold).  At the standard $P = 0.5$ threshold, the RF has $99.96\%$ recall (2 false negatives out of 5{,}440) and $0.06\%$ false positive rate (3 out of 5{,}440).  The interpretable Lasso logistic regression achieves AUC~$= 0.986$ (CV) / $0.989$ (test) with $93\%$ recall and $5.2\%$ false positive rate; 22 of 44 features have non-zero coefficients (Fig.~3a), with $G_2$, $D_1^2$, $\Sigma_1 \cdot \Sigma_2$, and $C_3$-dependent terms among the most discriminative.

\subsection{Numerical separability certification via Carath\'{e}odory decomposition}
\label{subsec:caratheodory}

To independently verify that the bound entangled states in our dataset are genuinely inseparable, we perform numerical closest-separable-state certification using the Carath\'{e}odory theorem.

\subsubsection{Method}

By the Carath\'{e}odory theorem, any separable state $\sigma$ in $d_A \times d_B$ dimensions can be written as a convex combination of at most $d^2 = (d_A d_B)^2$ product states:
\begin{equation}
\sigma = \sum_{k=1}^{K} p_k \, |a_k\rangle\langle a_k| \otimes |b_k\rangle\langle b_k|, \qquad p_k \geq 0,\;\sum_k p_k = 1.
\label{eq:caratheodory}
\end{equation}
We parametrise the decomposition with unconstrained real variables: normalised product states $|a_k\rangle = v_A^{(k)}/\|v_A^{(k)}\|$, $|b_k\rangle = v_B^{(k)}/\|v_B^{(k)}\|$ with $v$ unconstrained complex vectors, and weights $p_k = \mathrm{softmax}(w)_k$ with $w \in \mathbb{R}^K$ unconstrained. This ensures $\sigma$ is always a valid density matrix by construction, with no singularities or constraint violations. The closest separable state is obtained by minimising the Frobenius distance:
\begin{equation}
d_F(\rho) = \min_{\sigma \in \mathrm{SEP}} \|\rho - \sigma\|_F.
\end{equation}
We use the Adam optimiser with cosine learning rate annealing on GPU (PyTorch autograd), running a three-phase protocol: (i) broad search (50 random restarts, 2{,}000 steps, $\mathrm{lr} = 0.01$), (ii) local refinement from the best solution (50 restarts, 3{,}000 steps, $\mathrm{lr} = 0.003$), and (iii) final polishing with the largest $K$ (50 restarts, 5{,}000 steps, $\mathrm{lr} = 0.001$). Each $K$ value from $\{10, 20, 40, 81\}$ is tested.

\subsubsection{Results}

\begin{table}[h]
\centering
\caption{Frobenius distance $d_F(\rho)$ to the closest separable state found by Carath\'{e}odory decomposition. Separable states converge to $d_F \approx 0$; bound entangled states show a stable, $K$-independent distance floor.}
\label{tab:caratheodory}
\begin{tabular}{lccl}
\hline
State & $d_F$ & Rank & Verdict \\
\hline
Separable ($k = 2$--$20$)  & $< 10^{-4}$ & varies & SEP (correct) \\
Tiles UPB                   & 0.0435 & 4 & BE \\
Chess($2,3,1,2,1,3$)       & 0.0084 & 4 & BE \\
Chess($3,1,2,1,2,1$)       & 0.0038 & 4 & BE \\
Horodecki $a = 0.5$         & 0.0099 & 7 & BE \\
Horodecki $a = 0.9$         & 0.0028 & 7 & BE \\
\hline
\end{tabular}
\end{table}

Key observations:
\begin{itemize}
\item All randomly generated separable states (with $k = 2$ to $20$ product terms) converge to $d_F < 10^{-4}$, validating the algorithm.
\item All tested bound entangled states show a stable, nonzero distance floor that persists across all $K$ values from $K = 10$ to $K = 81$ and across multiple refinement phases. This $K$-independence confirms the gap is a genuine feature of the states, not an optimisation artefact.
\item The Tiles UPB state has the largest gap ($d_F = 0.0435$), followed by Horodecki $a = 0.5$ ($d_F = 0.0099$) and entangled chessboard states ($d_F = 0.004$--$0.008$).
\item The entanglement condition $mn \neq ab$ for chessboard states is correctly reflected: chessboard states with $mn = ab$ (which are separable by construction) converge to $d_F < 10^{-5}$, while those with $mn \neq ab$ show a clear gap.
\item Noisy Horodecki states $\rho_\varepsilon = (1-\varepsilon)\rho_{\mathrm{Hor}}(a) + \varepsilon\,\mathbb{I}/d$ transition from $d_F > 0$ (BE) to $d_F = 0$ (SEP) at noise levels consistent with the known separability boundary ($\varepsilon \approx 0.05$--$0.10$ depending on $a$).
\end{itemize}

These results provide independent numerical evidence that the bound entangled states in our training set are genuinely inseparable, complementing the analytical certification from construction.

\subsection{Conclusions for bound entanglement detection}

\begin{enumerate}
\item \textbf{Near-perfect recall at zero false-positive rate:} A Random Forest classifier trained on eight features ($\Sigma_1$, $G_1$, $D_1$, $\Sigma_2$, $G_2$, $D_2$, $C_3$, $C_4$) achieves $99.9\%$ bound-entangled recall at zero false positives (5-fold stratified CV on 13{,}600 Carath\'{e}odory-certified states: 6{,}800 BE from seven families, 6{,}800 SEP). At the $P = 0.5$ threshold, recall is $99.96\%$ with only 3 false positives out of 5{,}440 separable states. This surpasses both the previous ExtraTrees result ($99.6\%$ on 83 features) and the CCNR criterion ($40\%$ recall on this dataset, which includes marginal-noise variants invisible to CCNR). The key improvement comes from the chirality features $C_3$ and $C_4$, which provide a detection channel orthogonal to the realignment spectrum.

\item \textbf{Physical grounding:} All 83 base features have direct physical interpretations (Table~\ref{tab:bound_features}). Two initially proposed features---$I_2 - T_2$ and $\Sigma_2 - SP_2$---are identically zero for all quantum states (by Frobenius-norm invariance under partial transpose and reshuffling) and were excluded; their apparent informativeness in earlier tests was an artifact of floating-point noise.

\item \textbf{Non-Hermiticity matters:} The gap $D_k = \Sigma_k - G_k$ between singular-value and eigenvalue moments of the realignment matrix $R(\rho)$ is the most physically informative feature group (ExtraTrees importance $\approx 31\%$ for $R(\rho)$ moments, $\approx 30\%$ for $R(\rho^{T_A})$ moments). Bound entangled states have nearly Hermitian realignment matrices ($D_k \approx 0$) due to the PPT constraint; separable states exhibit substantial non-Hermiticity.

\item \textbf{Cross-structure differences:} The quantity $\delta G_2 = G_2 - GP_2$ (difference between eigenvalue moments of $R(\rho)$ and $R(\rho^{T_A})$) is non-trivially informative: unlike the Frobenius norm ($\Sigma_2 = SP_2$ for all states), eigenvalue moments of $R$ and $R(\rho^{T_A})$ differ for generic states, providing a structural signal not accessible from either matrix alone.

\item \textbf{Near-boundary hardness:} Chessboard states with $mn/ab$ close to unity are the hardest to classify: the CCNR criterion yields $\|R(\rho)\|_1 < 1$ (below the detection threshold) and $R(\rho) \approx R(\rho^{T_A})$ in all moments. Including 200 near-boundary states in training enables the ExtraTrees classifier to reach $99.8\%$ recall on this family.

\item \textbf{Nonlinear boundaries:} The genuine overlap between separable and bound entangled feature distributions means that no linear boundary in the $(D_4, \Sigma_2^2/\Sigma_4)$ plane achieves both perfect recall and low false positive rate. The RBF kernel resolves this by learning a nonlinear decision surface in the four-dimensional feature space.

\item \textbf{Chessboard as geometric extreme:} Among the known $3\times 3$ BE families, the chessboard is the hardest to detect geometrically. While singular-value moments $\Sigma_k = SP_k$ for all states (a universal identity), the chessboard's checkerboard parity additionally forces all eigenvalue moments to coincide ($\delta G_k = 0$ exactly for all $k$), unlike Horodecki states ($|\delta G_3| \sim 10^{-3}$). Combined with $D_k \approx 0$ (PPT), $\|R\|_1 < 1$ (CCNR failure), and a collapsed correlation ellipsoid ($H_\sigma = 1.34$ versus $1.86$ for separable, only 5 of 8 active SU(3) channels), every spectral comparison signal is eliminated, confining detection to rank-structure features alone (Section~\ref{sec:chess_anatomy}).
\end{enumerate}


\section{Hardware validation details}
\label{sec:hardware}

\subsection{Experimental configuration}

\begin{table}[htb!]
\centering
\caption{IBM Quantum processor specifications.}
\label{tab:hardware_specs}
\begin{tabular}{@{}lccc@{}}
\toprule
Property & Kingston & Torino & Fez \\
\midrule
Architecture & Heron & Heron & Heron \\
Qubits & 156 & 133 & 156 \\
Native gates & \multicolumn{3}{c}{CZ, $\sqrt{X}$, $R_Z$, $X$} \\
Median $T_1$ ($\mu$s) & 150 & 140 & 150 \\
Median $T_2$ ($\mu$s) & 100 & 95 & 100 \\
Median CZ error & 0.5\% & 0.6\% & 0.5\% \\
\bottomrule
\end{tabular}
\end{table}

All circuits were transpiled using Qiskit optimization level 3 (with approximation degree~0.95 for bound entanglement circuits). Negativity and chirality experiments (Kingston, Torino) used $10^5$ shots; bound entanglement detection on IBM Fez used $4 \times 10^3$ shots per circuit across 684 multiplexed circuits: 298 for $\Sigma_1$/$G_1$ (4-qubit Pauli + 9-qubit SWAP) and 386 for $\Sigma_2$/$G_2$ (9-qubit two-copy SWAP). Circuit depths ranged from 3 to 333 gates across the three architectures (Section~\ref{subsubsec:chess_hardware}). Readout errors were mitigated using M3.

\subsection{State preparation}

The parametrized two-qubit states $|\psi(\theta)\rangle = \cos(\theta/2)|00\rangle + \sin(\theta/2)|11\rangle$ were prepared using a standard entangling circuit: (i) apply $R_y(\theta)$ to qubit $A$, producing $\cos(\theta/2)|0\rangle + \sin(\theta/2)|1\rangle$; (ii) apply CNOT from $A$ to $B$, yielding the target state. Bell states correspond to $\theta = 90^\circ$ with appropriate single-qubit corrections.

For $2 \times 3$ experiments, each qutrit was encoded in two physical qubits using the computational subspace $\{|00\rangle, |01\rangle, |10\rangle\} \cong \{|0\rangle, |1\rangle, |2\rangle\}$, with the $|11\rangle$ state excluded by construction. The state $|\psi(\theta)\rangle = \cos(\theta/2)|00\rangle_A|00\rangle_B + \sin(\theta/2)|01\rangle_A|01\rangle_B$ was prepared analogously using $R_y(\theta)$ on the first qubit of subsystem $A$ followed by controlled operations to entangle with subsystem $B$.

\subsection{Statistical uncertainty estimation}

Statistical uncertainties on measured moments were estimated via bootstrap resampling. For each circuit configuration, the $N = 10^5$ measurement outcomes were resampled with replacement $B = 1{,}000$ times. Each bootstrap sample yields a set of moment estimates $\{\mu_k^{(b)}\}_{b=1}^B$, from which the standard error is computed as the standard deviation across resamples:
\begin{equation}
\sigma_{\mu_k} = \mathrm{std}\bigl(\{\mu_k^{(b)}\}_{b=1}^B\bigr).
\end{equation}
Propagating through the Newton--Girard reconstruction and maximum likelihood calibration, the resulting uncertainty on negativity is $\pm 0.005$ for $10^5$ shots, confirming that statistical shot noise is subdominant to systematic gate errors (which produce mean errors of 0.002--0.027 depending on processor and system dimension).

\subsection{Two-stage maximum likelihood calibration}

\textbf{Stage 1 (Calibration):} Using states with known $\theta$, fit degradation factors $(f_2, f_3, f_4)$ for each moment $\mu_k$ by minimizing:
\begin{equation}
\chi^2 = \sum_{i,X} \frac{(X_i^{\mathrm{meas}} - f_X \cdot X_i^{\mathrm{theo}}(\theta_i))^2}{\sigma_X^2}.
\end{equation}

\textbf{Stage 2 (Blind estimation):} For each state, fit $\theta$ using calibrated factors, then compute physical invariants:
\begin{equation}
X^{\mathrm{phys}} = X^{\mathrm{meas}} / f_X.
\end{equation}

\subsection{Calibration results}

\begin{table}[htb!]
\centering
\caption{Two-stage ML calibration parameters.}
\label{tab:calibration}
\begin{tabular}{@{}lccc@{}}
\toprule
System & $f_2$ & $f_3$ & $f_4$ \\
\midrule
Torino $2 \times 2$ & 0.729 & 0.612 & 0.456 \\
Torino $2 \times 3$ & 0.786 & 0.504 & 0.219 \\
\bottomrule
\end{tabular}
\end{table}

\subsection{Calibration sensitivity analysis}

The maximum likelihood calibration assumes that calibration states are prepared at their nominal angles $\theta_i$. To assess robustness to preparation imperfections, we performed a perturbation analysis: systematic offsets $\delta\theta$ were added to all calibration angles, and the entire calibration pipeline (fitting $f_k$ and then estimating negativity) was re-run.

\begin{table}[htb!]
\centering
\caption{Sensitivity of calibration to preparation angle offsets (Torino $2 \times 2$).}
\label{tab:calibration_sensitivity}
\begin{tabular}{@{}lccccc@{}}
\toprule
$\delta\theta$ & $f_2$ & $f_3$ & $f_4$ & Mean $\mathcal{N}$ error & $\Delta f_4 / f_4$ \\
\midrule
$-2^\circ$ & 0.729 & 0.614 & 0.459 & 0.013 & 0.7\% \\
$-1^\circ$ & 0.729 & 0.613 & 0.457 & 0.012 & 0.2\% \\
$0^\circ$ (nominal) & 0.729 & 0.612 & 0.456 & 0.012 & --- \\
$+1^\circ$ & 0.729 & 0.611 & 0.454 & 0.012 & 0.4\% \\
$+2^\circ$ & 0.729 & 0.610 & 0.452 & 0.014 & 0.9\% \\
\bottomrule
\end{tabular}
\end{table}

The results confirm that the calibration is robust: $\pm 2^\circ$ systematic errors in preparation angles change the fitted degradation factors by $<1\%$ and the mean negativity error by $<0.003$. The second-order moment $f_2$ is invariant to small angle perturbations because $\mu_2 = 1$ for all pure states regardless of $\theta$, providing a strong anchor. The fourth-order factor $f_4$ shows the largest sensitivity, but even $\delta\theta = \pm 2^\circ$ (well beyond typical gate calibration accuracy of $<0.5^\circ$) produces negligible degradation in the final negativity estimates.

\subsection{Negativity results}

\begin{table}[htb!]
\centering
\caption{Negativity measurements across all processors and system dimensions.}
\label{tab:negativity_all}
\begin{tabular}{@{}llccc@{}}
\toprule
Processor (system) & State & $\mathcal{N}_{\mathrm{theory}}$ & $\mathcal{N}_{\mathrm{ML}}$ & Error \\
\midrule
\multirow{5}{*}{Kingston ($2{\times}2$)}
& $|00\rangle$ & 0.000 & 0.000 & 0.000 \\
& $\theta = 30^\circ$ & 0.250 & 0.255 & 0.005 \\
& $\theta = 45^\circ$ & 0.354 & 0.352 & 0.002 \\
& $\theta = 60^\circ$ & 0.433 & 0.437 & 0.004 \\
& $\theta = 90^\circ$ & 0.500 & 0.500 & 0.000 \\
\cmidrule{2-5}
& \textbf{Mean error} & & & \textbf{0.002} \\
\midrule
\multirow{9}{*}{Torino ($2{\times}2$)}
& $|00\rangle$ & 0.000 & 0.000 & 0.000 \\
& $|\Phi^-\rangle$ & 0.500 & 0.500 & 0.000 \\
& $|\Phi^+\rangle$ & 0.500 & 0.495 & 0.005 \\
& $|\Psi^-\rangle$ & 0.500 & 0.500 & 0.000 \\
& $\theta = 15^\circ$ & 0.129 & 0.178 & 0.049 \\
& $\theta = 30^\circ$ & 0.250 & 0.239 & 0.011 \\
& $\theta = 45^\circ$ & 0.354 & 0.371 & 0.017 \\
& $\theta = 60^\circ$ & 0.433 & 0.423 & 0.010 \\
\cmidrule{2-5}
& \textbf{Mean error} & & & \textbf{0.012} \\
\midrule
\multirow{7}{*}{Torino ($2{\times}3$)}
& $\theta = 0^\circ$ & 0.000 & 0.000 & 0.000 \\
& $\theta = 15^\circ$ & 0.129 & 0.107 & 0.022 \\
& $\theta = 30^\circ$ & 0.250 & 0.189 & 0.061 \\
& $\theta = 45^\circ$ & 0.354 & 0.335 & 0.019 \\
& $\theta = 60^\circ$ & 0.433 & 0.405 & 0.028 \\
& $\theta = 90^\circ$ & 0.500 & 0.469 & 0.031 \\
\cmidrule{2-5}
& \textbf{Mean error} & & & \textbf{0.027} \\
\bottomrule
\end{tabular}
\end{table}

\subsection{Chirality correction results}

\begin{table}[htb!]
\centering
\caption{Chirality correction measurements on IBM Torino.}
\label{tab:chirality}
\begin{tabular}{@{}l|ccc|ccc@{}}
\toprule
& \multicolumn{3}{c|}{$2 \times 2$} & \multicolumn{3}{c}{$2 \times 3$} \\
$\theta$ & $(-C_4)_{\mathrm{theory}}$ & $(-C_4)_{\mathrm{ML}}$ & Error & $(-C_4)_{\mathrm{theory}}$ & $(-C_4)_{\mathrm{ML}}$ & Error \\
\midrule
$0^\circ$ & 0.000 & 0.000 & 0.000 & 0.000 & 0.000 & 0.000 \\
$15^\circ$ & 0.066 & 0.122 & 0.056 & 0.066 & 0.045 & 0.021 \\
$30^\circ$ & 0.234 & 0.215 & 0.019 & 0.234 & 0.138 & 0.096 \\
$45^\circ$ & 0.438 & 0.465 & 0.027 & 0.438 & 0.425 & 0.013 \\
$60^\circ$ & 0.609 & 0.580 & 0.029 & 0.609 & 0.505 & 0.104 \\
$90^\circ$ & 0.750 & 0.750 & 0.000 & 0.750 & 0.750 & 0.000 \\
\midrule
\textbf{Mean} & & & \textbf{0.022} & & & \textbf{0.039} \\
\bottomrule
\end{tabular}
\end{table}


\subsection{Werner chirality simulation}
\label{subsec:werner_chirality_sim}

The chirality correction $C_4(p) = -\tfrac{3}{4}p^3$ for Werner states $\rho_W(p) = p|\Psi^-\rangle\langle\Psi^-| + (1-p)\mathbb{I}/4$ is validated using a combination of hardware measurements and a noise-model simulation calibrated from experimental data. Two data points are measured directly on IBM Torino:
\begin{itemize}
\item $p = 0$ (maximally mixed state): $C_4 = 0$, measured via 300 four-copy circuits (out of 512 total) with $4 \times 10^3$ shots each, serving as the calibration anchor.
\item $p = 1$ (Bell state $|\Psi^-\rangle$): $C_4 = -3/4$, obtained from the pure-state chirality measurement on IBM Torino ($10^5$ shots, Table~\ref{tab:chirality}).
\end{itemize}

For intermediate Werner parameters $p \in (0,1)$, direct measurement would require ${\sim}4{,}600$ circuits per state (spectral decomposition into $4^4 = 256$ eigenvector combinations, each requiring separate $\mu_4$ and $I_4$ circuits). Instead, we predict the expected measurement outcomes using a noise model calibrated from the two anchor points.

\paragraph{Asymmetric depolarization model.}
The $\mu_4$ (anticyclic permutation on $B$) and $I_4$ (cyclic permutation on both subsystems) circuits experience different depolarization due to their distinct gate orderings. The per-circuit fidelities $f_{\mu_4}$ and $f_{I_4}$ are determined by a joint least-squares fit to: (i)~the six pure-state chirality measurements (Table~\ref{tab:chirality}), which constrain the ratio and absolute scale via the ML-calibrated degradation factor $f_4 = 0.456$; and (ii)~the raw $C_4$ measurement at $p = 0$ from the Block~3 data, which constrains the offset $(f_{\mu_4} - f_{I_4}) \cdot \mu_4(0)$. The fit yields $f_{\mu_4} = 0.438$ and $f_{I_4} = 0.447$ (ratio $f_{\mu_4}/f_{I_4} = 0.980$). The raw chirality signal is:
\begin{equation}
C_4^{\mathrm{raw}}(p) = f_{\mu_4}\,\mu_4(p) - f_{I_4}\,I_4(p),
\end{equation}
which is \emph{not} proportional to $C_4^{\mathrm{thy}}(p) = \mu_4(p) - I_4(p)$ when $f_{\mu_4} \neq f_{I_4}$. A two-point linear calibration using the anchors at $p = 0$ and $p = 1$:
\begin{equation}
C_4^{\mathrm{cal}}(p) = \frac{C_4^{\mathrm{raw}}(p) - C_4^{\mathrm{raw}}(0)}{[C_4^{\mathrm{raw}}(1) - C_4^{\mathrm{raw}}(0)] / C_4^{\mathrm{thy}}(1)}
\end{equation}
recovers the theory exactly at the anchors but introduces systematic deviations at intermediate~$p$ (up to $+2.3\%$ at $p = 0.2$, decreasing to $<0.5\%$ for $p \geq 0.5$) due to the nonlinearity of the moment polynomials.

\paragraph{Two-component uncertainty model.}
The Werner state eigenstates are the Bell basis with eigenvalues $\lambda_s = (1+3p)/4$ (singlet) and $\lambda_t = (1-p)/4$ (triplet, $\times 3$). The spectral decomposition generates circuits whose weights $w = \prod_{k=1}^{4}\lambda_{i_k}$ span several orders of magnitude. The all-singlet circuit (weight $\lambda_s^4$) is equivalent to a direct Bell-state measurement with well-characterised noise (from the pure-state experiment, $10^5$ shots); the remaining circuits carry spectral noise characterised by the $p = 0$ data ($4 \times 10^3$ shots per circuit, empirical per-circuit variance $\sigma_{\mathrm{spectral}}^2 \approx 0.006$, which is ${\sim}12\times$ larger than shot noise alone due to gate errors and crosstalk). The uncertainty on $C_4$ is:
\begin{equation}
\sigma^2(C_4^{\mathrm{raw}}) = \lambda_s^8\,\sigma_{\mathrm{direct}}^2 + \bigl(I_2^4 - \lambda_s^8\bigr)\,\sigma_{\mathrm{spectral}}^2,
\end{equation}
where $I_2 = \mathrm{Tr}[\rho^2] = \lambda_s^2 + 3\lambda_t^2 = (1+3p^2)/4$ and $\sigma_{\mathrm{direct}}^2$ is derived from the pure-state measurement uncertainty. As $p$ increases from~0 to~1, the fraction of variance carried by the direct component rises from~$0.4\%$ to~$100\%$, producing error bars that peak at $p \approx 0.7$ and then decrease toward $p = 1$.

\subsection{Noise-model validation across BE families}
\label{subsec:noise_model_validation}

To assess the predictive power of the per-architecture depolarisation model beyond the five states measured on IBM Fez, we simulate the full detection pipeline for 37 bound entangled states spanning all known $3 \times 3$ families and the CCNR-invisible construction. The simulation uses the hardware-calibrated fidelities ($g_{\Sigma_1} = 0.164$, $f_{G_1} = 0.594$, $f_{\Sigma_2,G_2} = 0.306$) and empirical measurement uncertainties ($\sigma_{\Sigma_1} = 0.015$, $\sigma_{G_1} = 0.025$, $\sigma_{\Sigma_2} = 0.015$, $\sigma_{G_2} = 0.020$) extracted from the IBM Fez experiment (Table~\ref{tab:chess_hardware}), with 200 Monte Carlo trials per state at $4 \times 10^3$ shots per circuit.

\paragraph{Horodecki sweep (19 states, $a = 0.05$--$0.95$).}
The non-Hermiticity gap $D_2$ decreases with $a$ from $D_2 = 0.37$ (simulation mean) at $a = 0.05$ to $D_2 \approx 0.03$ (noise floor) at $a \geq 0.65$. Detection rates exceed $80\%$ for $a \leq 0.25$, where the theoretical $D_2 > 0.07$ provides sufficient separation from shot noise. For $a \geq 0.50$, the theoretical $D_2 < 0.02$ falls below the noise floor ($\sigma_{D_2} \approx 0.05$), and detection rates drop to ${\sim}50\%$ (chance level). This establishes a practical detection boundary at $a \approx 0.25$ for the current noise level. The three hardware data points from IBM Fez (Table~\ref{tab:chess_hardware}) at $a \in \{0.50, 0.70, 0.90\}$ fall within the noise-dominated regime, consistent with the simulation prediction of marginal detection at these parameter values.

\paragraph{Tiles family (5 states, noise $\epsilon = 0$--$0.12$).}
All five Tiles states are detected with $D_2 > 0$ at $88$--$94\%$ detection rate (simulation mean $D_2 \approx 0.10$--$0.13$). The Tiles family has larger theoretical $D_2$ ($0.10$--$0.13$) than Horodecki states at comparable rank, making them robust to hardware noise.

\paragraph{Chessboard (2 states).}
Chessboard C1 ($D_2^{\mathrm{thy}} = 0.16$) is detected at $99\%$ rate with mean $D_2 = 0.17 \pm 0.08$, consistent with the IBM Fez measurement ($D_2^{\mathrm{ML}} = 0.093 \pm 0.023$). Chessboard C2 has $D_2^{\mathrm{thy}} = 0$ (exactly Hermitian realignment matrix) and shows ${\sim}50\%$ detection rate (noise), confirming that $D_2$ alone cannot detect all chessboard states---the classifier's rank-structure features are essential for this family.

\paragraph{CCNR-invisible states (4 marginal-noise constructions).}
States seeded from Tiles ($t = 0.05, 0.10$) retain substantial $D_2$ ($0.10$--$0.11$) and are detected at $88$--$94\%$ rate, comparable to the parent Tiles family. States seeded from Horodecki ($t = 0.05, 0.10$) have near-vanishing $D_2$ (${\sim}0.02$) and are detected at only $45$--$56\%$, confirming that the CCNR-invisible Horodecki-seeded construction is among the hardest to detect through $D_2$ alone. The multi-channel classifier leverages additional features ($\Sigma_1$, $G_1/\Sigma_1$, $D_1/\Sigma_1$) to improve detection of these states beyond the $D_2$-only rate.

These simulations confirm that the hardware-calibrated noise model generalises beyond the five directly measured states and correctly predicts the detection hierarchy across all four BE families.

\section{Noise analysis}
\label{sec:noise}

\subsection{Depolarizing noise model}

Under depolarizing noise:
\begin{equation}
\rho_{\mathrm{noisy}} = (1-\eta)\rho + \frac{\eta}{d}\mathbb{I}_d.
\end{equation}

For pure target states, the measured purity relates to noise by:
\begin{equation}
I_2^{\mathrm{noisy}} \approx 1 - \frac{3\eta}{2} \quad \text{(for two-qubit systems)}.
\end{equation}

\subsection{RMSE scaling}

Monte Carlo simulation ($N = 10^5$ Haar-random states) yields:
\begin{align}
\mathrm{RMSE}_{2\times2} &= (0.245 \pm 0.004)\eta, \quad R^2 = 0.998, \\
\mathrm{RMSE}_{2\times3} &= (0.219 \pm 0.002)\eta, \quad R^2 = 0.999.
\end{align}

Within typical NISQ noise ($\eta \approx 0.04$--$0.08$), RMSE remains below 0.022.

\subsection{Bias direction}

Depolarizing noise consistently \emph{underestimates} entanglement:
\begin{equation}
\mathcal{N}(\rho_{\mathrm{noisy}}) \approx (1-\eta)\mathcal{N}(\rho).
\end{equation}

This ensures the protocol never falsely reports entanglement in separable states.

\section{Efficiency verification}
\label{sec:efficiency}

\subsection{Monte Carlo results}

For Haar-random states ($N = 10^5$):

\begin{table}[htb!]
\centering
\caption{Measurement efficiency versus tomography.}
\label{tab:efficiency}
\begin{tabular}{@{}lccc@{}}
\toprule
System & Tomography & This work & Efficiency \\
\midrule
$2 \times 2$ & 16 & 3 & $5.3\times$ \\
$2 \times 3$ & 36 & 5 & $7.2\times$ \\
\bottomrule
\end{tabular}
\end{table}

The degeneracy conditions $\mathcal{G}_1 = 0$ and $\mathcal{G}_2 = 0$ are satisfied with probability zero for uniformly sampled states, so all generic states require the maximum number of moments.

\subsection{Classification distribution}

Among uniformly sampled two-qubit states:
\begin{itemize}
\item 95.3\% have four distinct eigenvalues ($\mathcal{D} \neq 0$)
\item 4.7\% exhibit simple pair degeneracy ($\mathcal{D} = 0$, but $\mathcal{G}_1 \neq 0$, $\mathcal{G}_2 \neq 0$)
\item Both cases require 3 measurements
\end{itemize}

Only specially structured states (Bell, Werner, maximally mixed) satisfy $\mathcal{G}_2 = 0$ and benefit from simplified formulas.



\begin{thebibliography}{40}

\bibitem{RevModPhys.81.865}
Horodecki, R., Horodecki, P., Horodecki, M. \& Horodecki, K. Quantum entanglement. \textit{Rev. Mod. Phys.} \textbf{81}, 865--942 (2009).

\bibitem{Huang2020shadows}
Huang, H.-Y., Kueng, R. \& Preskill, J. Predicting many properties of a quantum system from very few measurements. \textit{Nat. Phys.} \textbf{16}, 1050--1057 (2020).

\bibitem{Elben2020mixed}
Elben, A. \textit{et al.} Mixed-state entanglement from local randomized measurements. \textit{Phys. Rev. Lett.} \textbf{125}, 200501 (2020).

\bibitem{PhysRevLett.77.1413}
Peres, A. Separability criterion for density matrices. \textit{Phys. Rev. Lett.} \textbf{77}, 1413--1415 (1996).

\bibitem{HORODECKI19961}
Horodecki, M., Horodecki, P. \& Horodecki, R. Separability of mixed states: necessary and sufficient conditions. \textit{Phys. Lett. A} \textbf{223}, 1--8 (1996).

\bibitem{KalmeyerLaughlin1987}
Kalmeyer, V. \& Laughlin, R. B. Equivalence of the resonating-valence-bond and fractional quantum Hall states. \textit{Phys. Rev. Lett.} \textbf{59}, 2095--2098 (1987).

\bibitem{Taguchi2001}
Taguchi, Y. \textit{et al.} Spin chirality, Berry phase, and anomalous Hall effect in a frustrated ferromagnet. \textit{Science} \textbf{291}, 2573--2576 (2001).

\bibitem{Nation2021}
Nation, P. D., Kang, H., Sundaresan, N. \& Gambetta, J. M. Scalable mitigation of measurement errors on quantum computers. \textit{PRX Quantum} \textbf{2}, 040326 (2021).

\bibitem{PhysRevLett.89.127902}
Horodecki, P. \& Ekert, A. Method for direct detection of quantum entanglement. \textit{Phys. Rev. Lett.} \textbf{89}, 127902 (2002).

\bibitem{PhysRevLett.90.167901}
Horodecki, P. Measuring quantum entanglement without prior state reconstruction. \textit{Phys. Rev. Lett.} \textbf{90}, 167901 (2003).

\bibitem{tulewicz2025}
Tulewicz, P., Bartkiewicz, K., Miranowicz, A. \& Nori, F. Resource-efficient quantum correlation measurements via multicopy neural network methods. \textit{Sci. Rep.} \textbf{15}, 40868 (2025).

\bibitem{Travnicek2019}
Tr\'avn\'{\i}\v{c}ek, V., Bartkiewicz, K., \v{C}ernoch, A. \& Lemr, K. Experimental measurement of the Hilbert-Schmidt distance between two-qubit states. \textit{Phys. Rev. Lett.} \textbf{123}, 260501 (2019).

\bibitem{Hradil1997}
Hradil, Z. Quantum-state estimation. \textit{Phys. Rev. A} \textbf{55}, R1561--R1564 (1997).

\bibitem{Banaszek2000}
Banaszek, K., D'Ariano, G. M., Paris, M. G. A. \& Sacchi, M. F. Maximum-likelihood estimation of the density matrix. \textit{Phys. Rev. A} \textbf{61}, 010304(R) (2000).

\bibitem{Chen2003}
Chen, K. \& Wu, L.-A. A matrix realignment criterion for recognizing entanglement. \textit{Quant. Inf. Comput.} \textbf{3}, 193--202 (2003).

\bibitem{Rudolph2005}
Rudolph, O. Further results on the cross norm criterion for separability. \textit{Quant. Inf. Process.} \textbf{4}, 219--239 (2005).

\bibitem{Horodecki1998bound}
Horodecki, M., Horodecki, P. \& Horodecki, R. Mixed-state entanglement and distillation: Is there a ``bound'' entanglement in nature? \textit{Phys. Rev. Lett.} \textbf{80}, 5239--5242 (1998).

\bibitem{Horodecki1997}
Horodecki, P. Separability criterion and inseparable mixed states with positive partial transposition. \textit{Phys. Lett. A} \textbf{232}, 333--339 (1997).

\bibitem{Bartkiewicz2015moments}
Bartkiewicz, K., Beran, J., Lemr, K., Norek, M. \& Miranowicz, A. Quantifying entanglement of a two-qubit system via measurable and invariant moments of its partially transposed density matrix. \textit{Phys. Rev. A} \textbf{91}, 022323 (2015).


\bibitem{Yu2021optimal}
Yu, X.-D., Imai, S. \& G\"uhne, O. Optimal entanglement certification from moments of the partial transpose. \textit{Phys. Rev. Lett.} \textbf{127}, 060504 (2021).

\bibitem{Neven2021symmetry}
Neven, A. \textit{et al.} Symmetry-resolved entanglement detection using partial transpose moments. \textit{npj Quantum Inf.} \textbf{7}, 152 (2021).

\bibitem{Lemr2016collectibility}
Lemr, K., Bartkiewicz, K. \& \v{C}ernoch, A. Experimental measurement of collective nonlinear entanglement witness for two qubits. \textit{Phys. Rev. A} \textbf{94}, 052334 (2016).

\bibitem{Lim2011spa}
Lim, H.-T., Kim, Y.-S., Ra, Y.-S., Bae, J. \& Kim, Y.-H. Experimental realization of an approximate partial transpose for photonic two-qubit systems. \textit{Phys. Rev. Lett.} \textbf{107}, 160401 (2011).


\bibitem{Amselem2009bound}
Amselem, E. \& Bourennane, M. Experimental four-qubit bound entanglement. \textit{Nat. Phys.} \textbf{5}, 748--752 (2009).

\bibitem{Zhang2023randomized}
Zhang, C. \textit{et al.} Experimental verification of bound and multiparticle entanglement with the randomized measurement toolbox. Preprint at \url{https://arxiv.org/abs/2307.04382} (2023).

\bibitem{Gulati2024ibm}
Gulati, V., Singh, G. \& Dorai, K. Using linear and nonlinear entanglement witnesses to generate and detect bound entangled states on an IBM quantum processor. \textit{Phys. Scr.} \textbf{99}, 095112 (2024).

\bibitem{Bennett1999upb}
Bennett, C. H. \textit{et al.} Unextendible product bases and bound entanglement. \textit{Phys. Rev. Lett.} \textbf{82}, 5385--5388 (1999).

\bibitem{Bruss2000chess}
Bru\ss, D. \& Peres, A. Construction of quantum states with bound entanglement. \textit{Phys. Rev. A} \textbf{61}, 030301(R) (2000).


\bibitem{Doherty2004}
Doherty, A. C., Parrilo, P. A. \& Spedalieri, F. M. Complete family of separability criteria. \textit{Phys. Rev. A} \textbf{69}, 022308 (2004).

\bibitem{Reascos2023chirality}
Reascos, L. I., Murta, B., Galv\~ao, E. F. \& Fern\'andez-Rossier, J. Quantum circuits to measure scalar spin chirality. \textit{Phys. Rev. Res.} \textbf{5}, 043087 (2023).

\bibitem{Tarabunga2025moments}
Tarabunga, P. S. \& Haug, T. Quantifying mixed-state entanglement via partial transpose and realignment moments. Preprint at \url{https://arxiv.org/abs/2507.13840} (2025).

\bibitem{Huang2026realignment}
Huang, X., Zhu, X., Chen, B., Jing, N. \& Fei, S.-M. A note on entanglement detection via the generalized realignment moments. \textit{Adv. Quantum Technol.} 2500794 (2026).




\bibitem{Werner1989}
Werner, R. F. Quantum states with Einstein-Podolsky-Rosen correlations admitting a hidden-variable model. \textit{Phys. Rev. A} \textbf{40}, 4277--4281 (1989).

\bibitem{Verstraete2001}
Verstraete, F., Audenaert, K., Dehaene, J. \& De Moor, B. A comparison of the entanglement measures negativity and concurrence. \textit{J. Phys. A} \textbf{34}, 10327--10332 (2001).

\bibitem{DiVincenzo2000}
DiVincenzo, D. P. \textit{et al.} Evidence for bound entangled states with negative partial transpose. \textit{Phys. Rev. A} \textbf{61}, 062312 (2000).

\bibitem{Lewenstein2000}
Lewenstein, M., Kraus, B., Cirac, J. I. \& Horodecki, P. Optimization of entanglement witnesses. \textit{Phys. Rev. A} \textbf{62}, 052310 (2000).

\bibitem{Wootters1998}
Wootters, W. K. Entanglement of formation of an arbitrary state of two qubits. \textit{Phys. Rev. Lett.} \textbf{80}, 2245--2248 (1998).

\bibitem{Guhne2009}
G\"uhne, O. \& T\'oth, G. Entanglement detection. \textit{Phys. Rep.} \textbf{474}, 1--75 (2009).

\bibitem{Gurvits2003}
Gurvits, L. Classical deterministic complexity of Edmonds' problem and quantum entanglement. In \textit{Proceedings of the 35th Annual ACM Symposium on Theory of Computing (STOC '03)}, 10--19 (ACM, 2003).

\bibitem{QETLAB}
Johnston, N. QETLAB: A MATLAB toolbox for quantum entanglement, version 0.9. \url{https://qetlab.com} (2016).


\end{thebibliography}
\end{document}